\documentclass[11pt,a4paper,reqno]{amsart}
\usepackage{fullpage} 

\usepackage[T1]{fontenc}
\usepackage{textcomp}
\usepackage[english]{babel}
\usepackage{filecontents}
\usepackage{amsthm,amsmath,amssymb,xcolor}
\usepackage{mathtools}
\usepackage[
  pdfborder = {0 0 0},
  colorlinks = true,
  citecolor = blue!70!black,
  linkcolor = blue!70!black]{hyperref}
\usepackage{float}
\usepackage{subcaption}
\usepackage{enumerate}
\usepackage{thmtools}
\usepackage{thm-restate}

\theoremstyle{plain}
\newtheorem{theorem}{Theorem}
\newtheorem{lemma}[theorem]{Lemma}
\newtheorem{corollary}[theorem]{Corollary}

\theoremstyle{definition}
\newtheorem{definition}[theorem]{Definition}

\newtheorem{remark}[theorem]{Remark}

\theoremstyle{plain}

\newtheorem*{thmmaintwo}{Theorem~\ref{thm:main2}}
\newtheorem*{thmDensity}{Theorem~\ref{thm: Density of Ratios}}
\newtheorem*{lemcases}{Lemma~\ref{lem: Many Cases Lemma}}

\DeclareMathOperator{\tr}{tr}
\DeclareMathOperator{\Arc}{Arc}
\DeclareMathOperator{\Real}{Re}
\DeclareMathOperator{\Imag}{Im}

\newcommand{\SQ}[0]{\mathbb{S}_\mathbb{Q}}

\newcommand{\len}[1]{\ell{\left(#1\right)}}
\newcommand{\floor}[1]{\left\lfloor#1\right\rfloor}

\def\imm{\ensuremath{\mathrm{i}}}

\def\Gc{\ensuremath{\mathcal{G}}}
\def\Tc{\ensuremath{\mathcal{T}}}

\def\hQ{\ensuremath{\mathbb{Q}}}
\def\CQ{\ensuremath{\mathbb{C}_{\hQ}}}

\renewcommand\arg{\text{arg}}
\newcommand\Arg{\text{Arg}} 
\newcommand{\size}[1]{\mathrm{size}(#1)}
\newcommand{\conj}[1]{\overline{#1}}
\def\FactorIsing#1{\#\ensuremath{\mathsf{IsingNorm}(\lambda,b,\Delta,#1)}}
\def\FactorIsingb{\#\ensuremath{\mathsf{IsingNorm}(\lambda,b,\Delta)}}
\def\FactorIsingc{\#\ensuremath{\mathsf{IsingNorm}(\lambda,b,3)}}
\def\ArgIsing#1{\#\ensuremath{\mathsf{IsingArg}(\lambda,b,\Delta,#1)}}
\def\PM{\ensuremath{\#\mathsf{PerfectMatchings}}}
\def\numP{\ensuremath{\mathsf{\#P}}}
\def\NP{\ensuremath{\mathsf{NP}}}

\def\polP{\ensuremath{\mathsf{P}}}

\newcommand{\im}{\mathrm{i}}
\newcommand{\emm}{\mathrm{e}}
\newcommand{\pl}{\textup{\texttt{+}}}
\newcommand{\mi}{\textup{\texttt{-}}}
\newcommand{\plm}{\textup{\texttt{\textpm}}}
\newcommand{\goal}{\mathsf{goal}}

\makeatletter
\def\prob#1#2#3{\goodbreak\begin{list}{}{\labelwidth\z@ \itemindent-\leftmargin
                        \itemsep\z@  \topsep6\p@\@plus6\p@
                        \let\makelabel\descriptionlabel}
                \item[\it Name]#1
               \item[\it Instance]                #2
                \item[\it Output]#3
                \end{list}}
\makeatother

\title{Lee-Yang zeros and the complexity of the ferromagnetic  Ising model on bounded-degree graphs}
\date{\today}
\author{Pjotr Buys}
\author{Andreas Galanis}
\author{Viresh Patel}
\author{Guus Regts}

\address[Pjotr Buys, Viresh Patel, Guus Regts]{Korteweg de Vries Institute for Mathematics, University of Amsterdam.}
\email{\{pjotr.buys,guusregts\}@gmail.com, vpatel@uva.nl}

\address[Andreas Galanis]{Department of Computer Science, University of Oxford.}
\email{andreas.galanis@cs.ox.ac.uk}

\begin{document}

\begin{abstract}
We study the computational complexity of approximating the partition function of the ferromagnetic Ising model with the external field parameter $\lambda$ on the unit circle in the complex plane. 
Complex-valued parameters for the Ising model are relevant for quantum circuit computations and phase transitions in statistical physics, but have also been key in the recent deterministic approximation scheme for all $|\lambda|\neq 1$ by Liu, Sinclair, and Srivastava. 
Here, we focus on the unresolved complexity picture on the unit circle, and on the tantalising question of what happens around $\lambda=1$, where on one hand the classical algorithm of Jerrum and Sinclair gives a randomised  approximation scheme on the real axis suggesting tractability, and  on the other hand  the presence of Lee-Yang zeros alludes to computational hardness.

Our main result establishes a sharp computational transition at the point $\lambda=1$, and more generally on the entire unit circle.
For an integer $\Delta\geq 3$ and edge interaction parameter $b\in (0,1)$ we show $\numP$-hardness for approximating the partition function on graphs of maximum degree $\Delta$ on the arc of the unit circle where the Lee-Yang zeros are dense. 
This result contrasts with known approximation algorithms when $|\lambda|\neq 1$ or when $\lambda$ is in the complementary arc around $1$ of the unit circle.
Our work thus gives a direct connection between the presence/absence of Lee-Yang zeros and  the tractability of efficiently approximating the partition function on bounded-degree graphs. 

\end{abstract}

\maketitle

\section{Introduction}

The Ising model is a classical model from statistical physics that arises in multiple sampling and inference tasks across computer science. The model has an edge-interaction parameter $b$ and a vertex parameter $\lambda$, known as the external field. For a graph $G=(V,E)$, configurations of the model are all possible assignments of two spins $\pl,\mi$ to the vertices of $G$. Each  configuration $\sigma:V\rightarrow \{\pl,\mi \}$ has  weight  $\lambda^{|n_\pl(\sigma)|}b^{\delta(\sigma)}$, where $n_\pl(\sigma)$ is the set of vertices that get the spin $\pl$ under $\sigma$ and $\delta(\sigma)$ is the number of edges that get different  spins.\footnote{The parametrisation of the Ising model in terms of $\delta(\sigma)$ follows the closely related works \cite{liu2019ising,PetersRegts2018}; if instead the model is defined in terms of the number of  edges with the same spins the edge-interaction parameter $1/b\in (1,\infty)$ is obtained, whose logarithm is known as the inverse temperature in the physics literature.} The \emph{partition function} is the aggregate weight of all configurations, i.e.,  \[Z_G(\lambda,b)=\sum_{\sigma:V\rightarrow\{\pl ,\mi \}}\lambda^{|n_\pl(\sigma)|}b^{\delta(\sigma)}.\]
In this paper, we consider the problem of approximating the partition function when $b\in (0,1]$, known as the ferromagnetic case, and when the parameter $\lambda$ is in the complex plane. Complex parameters for the Ising model have been studied in the computation of probability amplitudes of quantum circuits, see, e.g., \cite{Cuevas,quant,bremner}. Somewhat surprisingly, and this is one of the main motivations behind this work, complex parameters are also fundamental in understanding the complexity of approximation even for real-valued parameters.

In particular, many of the recent advances on the development of approximation algorithms for counting problems have been based on viewing the partition function as a polynomial of the underlying parameters in the complex plane, and using refined interpolation techniques from \cite{Barvinokbook,PR17} to obtain fully polynomial time approximation schemes ($\mathsf{FPTAS}$, see below for the technical definition), even for real values~\cite{GLLZ,GLL,LSP,bencs2018zerofree,barvinokregts2019, liufisher, shao2019contraction, PetersRegts2018,peters2019}. The bottleneck of this approach is establishing zero-free regions in the complex plane of the polynomials, which in turn requires an in-depth understanding of the models with complex-valued parameters. This framework of  designing approximation algorithms aligns with the classical statistical physics perspective on phase transitions, where zeros in the complex plane have long been studied in the context of phase transitions (see, e.g., \cite{yang1952statistical,Heilmann1972}), and several of these classical results have  recently been used to obtain efficient approximation algorithms~(\cite{liu2019ising,PR17}). 

In particular, the celebrated Lee-Yang theorem~\cite{lee1952statistical2} says that, when regarding the partition function of the ferromagnetic Ising model as a polynomial in the external field parameter $\lambda$, all its zeros, referred to as \emph{Lee-Yang zeros}, lie on the unit circle in the complex plane. (These Lee-Yang zeros have actually been observed in quantum experiments~\cite{exper}.) The Lee-Yang theorem was recently used by Liu, Sinclair, and Srivastava~\cite{liu2019ising} to obtain an $\mathsf{FPTAS}$ for approximating the partition function for values $\lambda\in \mathbb{C}$ that do not lie on the unit circle. 
This result can be viewed as a derandomisation of the Markov chain based randomised algorithm by Jerrum and Sinclair \cite{JS:ising} for $\lambda>0$ (see also \cite{guo2018,Collevecchio2016}), solving a longstanding problem.\footnote{Notably, the correlation decay approach, which also yields deterministic approximation algorithms and was key in the full classification of antiferromagnetic 2-spin systems \cite{LLY,SST,SlySun,GSVanti}, somewhat surprisingly does not perform as well for ferromagnetic systems, see \cite{GuoLu} for the state-of-the-art on this front.} 

As noted in \cite[Remark p.290]{liu2019ising}, the ``no-field'' case $|\lambda|=1$ is unclear, since on the one hand we have the algorithm by \cite{JS:ising} for $\lambda=1$, on the other hand it is known that Lee-Yang zeros are dense on the unit circle. 
The density picture was further explored in \cite{PetersRegts2018} for graphs of bounded maximum degree $\Delta$, by establishing for each $b\in (0,1)$ a symmetric arc around $\lambda=1$ on the unit circle where the partition function does not vanish for all graphs of maximum degree at most $\Delta$ and showing density of the Lee-Yang zeros on the complementary arc. See also~\cite{Chio2019} for the density result.

In this paper we resolve the complexity picture of the ferromagnetic Ising model.
We show that for graphs of maximum degree $\Delta$ approximately computing the partition function is $\numP$-hard\footnote{Roughly, \#P is the counting version of problems in NP, see, e.g., \cite{VALIANT} for details.}, on the arc of the unit circle where the Lee-Yang zeros are dense. See Theorem~\ref{thm:main2} below for a precise statement of our main results.
Since on the complementary arc there exists an $\mathsf{FPTAS}$, by the results of \cite{PetersRegts2018} (in combination with~\cite{Barvinokbook,PR17}), this gives a direct connection between hardness of approximation and the presence of Lee-Yang zeros. Combined with the results of~\cite{JS:ising,liu2019ising}, our work therefore  classifies the complexity of approximating the partition function of the ferromagnetic Ising model on the complex plane. 

It should be noted that the existence of zeros does not imply hardness in a straightforward manner.\footnote{For example, the graphs in \cite{PetersRegts2018} whose partition function is shown to be zero are trees, and these can be clearly detected in polynomial time. More generally, it is hard to imagine a construction of graphs with vanishing partition function which can directly yield hardness. In any case, our results, following the framework of \cite{GJ, BGGS, matchings} show hardness for a relaxed version of the problems where zeros do not need to be detected, making all these considerations irrelevant.} 
We obtain the connection between the Lee-Yang zeros and computational complexity via tools from complex dynamical systems. The partition function on trees naturally gives rise to a dynamical system, cf. Section~\ref{sec:dynamics}. Both the hardness of approximating the partition function as well as the density of the Lee-Yang zeros originate from chaotic behavior of the dynamics, while normal behavior is linked to absence of zeros and hence the existence of efficient approximation algorithms~\cite{PetersRegts2018}.

Our work falls into the broader context of showing how zeros in the complex plane actually relate to the existence and design of approximation algorithms. 
This connection has been well studied for general graphs, see, e.g., \cite{GJ,GG,galanis2020complexity}; for bounded-degree graphs, the picture is less clear, but the key seems to lie in understanding  the underlying complex dynamical systems \cite{BGGS, buys2019location, Chio2019, PetersRegts2018, bencs2019leeyang, matchings}. 
A general theory is so far elusive, but it seems that the chaotic behavior of the underlying complex dynamical system is linked to the presence of zeros of the partition function and to the $\numP$-hardness of approximation.

Establishing hardness results for ferromagnetic spin systems is notoriously challenging \cite{LLZ, GuoLu}, we therefore expect our techniques to be applicable in a wider framework. 
We will explain in Section~\ref{sec:outline} the obstacles that arise relative to previous works for antiferromagnetic spin systems.

\subsection{Our results}\label{sec:results}
To state our inapproximability results, we first formally define the computational problems that we consider.  For $z\in \mathbb{C}$, we let $|z|$ be the norm of $z$, $\Arg(z)$ be its argument in the interval $[0,2\pi)$, and  $\arg(z)=\{\Arg(z)+2k\pi\mid k\in \mathbb{Z}\}$ be the set of all of its arguments. We will consider the problems of approximating the norm of the partition function $Z_G(\lambda,b)$ within a rational factor $K>1$ and its argument within an additive rational constant $\rho>0$. For the computational problems, we moreover assume that $b\in (0,1)$ is rational and $\lambda$ has rational real and imaginary parts. The rationality assumption is mainly for convenience (representation issues), and it simplifies some of the proofs.

 \prob{
$\FactorIsing{K}$.} {A graph $G=(V,E)$ with maximum degree $\leq \Delta$.  } 
{   If $Z_G(\lambda,b)=0$, the algorithm may output any rational. Otherwise,
 it must return a rational $\widehat{N}$ such that
$\widehat{N}/K \leq
|Z_{G}(\lambda,b)|\leq K \widehat{N}$. } 

We remark here that the explicit constant $K>1$ in the problem definition above is only for convenience, having $K=2^{n^{1-\epsilon}}$ for any constant $\epsilon>0$ does not change the complexity of the problem using standard powering arguments (here, $n$ is the size of the input graph).

  \prob{
$\ArgIsing{\rho}$.} {A graph $G=(V,E)$ with maximum degree $\leq \Delta$. 
 }
{ If $Z_G(\lambda,b)=0$, the algorithm may output any rational. Otherwise, it must return  
a rational  $\widehat{A}$ such that $|\widehat{A} - a| \leq \rho$ for some $a\in\arg(Z_{G}(\lambda,b))$.  }

A \emph{fully polynomial time approximation scheme} ($\mathsf{FPTAS}$) for approximating $Z_G(\lambda,b)$ for given $\lambda$ and $b$  and positive integer $\Delta$ is an algorithm that for any $n$-vertex graph $G$ of maximum degree at most $\Delta$ and any rational $\varepsilon>0$ solves both probems $\FactorIsing{1+\varepsilon}$ and $\ArgIsing{\varepsilon}$ in time polynomial in $n/\varepsilon$.

 We use $\mathbb{Q}$ to denote the set of rational numbers and $\CQ$ to denote the set of complex numbers with rational real and imaginary part.
We denote by $\mathbb{S}$ the unit circle in the complex plane, and $\SQ=\mathbb{S}\cap \CQ$. It is well-known that numbers in $\SQ$ are dense on the unit circle.\footnote{See for example the upcoming Lemma~\ref{lem:rationalangle}.}
For $\theta\in (0,\pi)$ we denote 
\[I(\theta):=\{e^{i\vartheta}\mid -\theta<\vartheta<\theta\}.\]
For $\Delta\geq 3$ and $b\in (\frac{\Delta-2}{\Delta},1)$ we denote by $\theta_b\in (0,\pi)$ the angle from \cite[Theorem A]{PetersRegts2018} for which the following holds: 
\begin{itemize}
\item[(i)] for any graph $G$ of maximum degree at most $\Delta$ and any $\lambda\in I(\theta_b)$,  $Z_G(\lambda,b)\neq 0$; 
\item[(ii)]for each $\lambda\in \mathbb{S}\setminus I(\theta_b)$ there exists $\lambda'\in \mathbb{S}$ arbitarily close to $\lambda$ and a tree $T$ of maximum degree $\Delta$ for which $Z_T(\lambda',b)=0$.
\end{itemize}
Our main result is as follows.

\newcommand{\statethmmaintwo}{Let $\Delta\geq 3$ be an integer and let $K=1.001$ and $\rho=\pi/40$.
\begin{itemize}
\item[(a)] Let $b \in\big(0,\frac{\Delta-2}{\Delta}\big]$ be a rational, and  $\lambda\in \SQ$ such that  $\lambda\neq \pm 1$. Then the problems $\FactorIsing{K}$ and $\ArgIsing{\rho}$ are $\numP$-hard.
\item[(b)] Let $b\in \big(\frac{\Delta-2}{\Delta},1\big)$ be a rational. 
Then the collection of complex numbers $\lambda\in \SQ$ for which $\FactorIsing{K}$ and $\ArgIsing{\rho}$ are $\numP$-hard is dense in the arc $\mathbb{S}\setminus I(\theta_b)$.
\end{itemize}}

\begin{theorem}\label{thm:main2}
\statethmmaintwo
\end{theorem}
Combined with \cite{liu2019ising}, part (a) of our main theorem completely classifies the hardness of approximating the partition function $Z_G(\lambda,b)$ (as per the two computational problems stated above), for $b\in (0,\frac{\Delta-2}{\Delta}]$.
Combined with \cite[Corollary 1]{PetersRegts2018}, part (b) of our main theorem `essentially' classifies the hardness of approximating the partition function for $b\in(\frac{\Delta-2}{\Delta},1)$ and answers a question from~\cite{PetersRegts2018}.
Technically, we do not rule out that there may be an efficient algorithm for these problems for some $\lambda\in \SQ\setminus I(\theta_b)$, but such an algorithm must be specifically tailored to such a particular $\lambda$ (unless of course $\polP=\numP$). 
We in fact conjecture that, when $b\in(\frac{\Delta-2}{\Delta},1)$, approximating the partition function (as in Theorem 1) is $\numP$-hard for all non-real $\lambda \in  \SQ\setminus I(\theta_b)$. 
See Remark~\ref{rem:real} below for a discussion of the antipodal case $\lambda=-1$.

We should further remark that the open interval $b\in \big(0,\frac{\Delta-2}{\Delta}\big)$ for positive $\lambda$ corresponds to the so-called non-uniqueness region of the infinite $\Delta$-regular tree. For the antiferromagnetic Ising model and positive $\lambda$, non-uniqueness leads to computational intractability~\cite{SlySun,GSV2016}, in contrast to the ferromagnetic case. As we explain in Section~\ref{sec:outline}, the phenomenon which underpins our proofs for $|\lambda|=1$ with $\lambda\neq \pm 1$ is the chaotic behaviour of the underlying complex dynamical system, which resembles in rough terms a complex-plane analogue of non-uniqueness. Interestingly, at criticality, i.e., when $b=\frac{\Delta-2}{\Delta}$, while the model is in uniqueness for $\lambda=1$, the chaotic behaviour is nevertheless present for non-real $\lambda$, and we show $\numP$-hardness for this case too. 

\begin{remark}\label{rem:real}
We further discuss the real cases $\lambda=\pm 1$ which are not explicitly covered by Theorem~\ref{thm:main2}. The case $\lambda=1$ admits an $\mathsf{FPRAS}$  \cite{JS:ising,guo2018,Collevecchio2016}, but the existence of a deterministic approximation scheme is open. We study the case $\lambda=-1$ in more detail in Section~\ref{sec:minusone}  where we show that the problem is not $\numP$-hard (assuming $\numP\neq \NP$): using the ``high-temperature'' expansion of the model, we show an odd-subgraphs formulation of the partition function (Lemma~\ref{lem:connection}), which is then used to conclude (Theorem~\ref{thm:lambdaminusone})  that the sign of the partition function can be determined trivially, while the problem of approximating the norm of the partition function for all $\Delta\geq 3$ is equivalent to the problem of approximately counting the number of perfect matchings (even on unbounded-degree graphs); the complexity of the latter is an open problem in general, but it can be approximated with an NP-oracle \cite{Jerrum1986}, therefore precluding $\numP$-hardness.
\end{remark}

In the next section, we give an outline of the key pieces to obtain our inapproximability results; the details of these pieces will be filled in the forthcoming sections (see also the upcoming Section~\ref{sec:out-line} for the organisation of the paper).

\section{Proof outline}\label{sec:outline}

Let $\Delta\geq 3$ be an integer, $b\in (0,\frac{\Delta-2}{\Delta}]$, and $\lambda\in \SQ$ with $\lambda\neq \pm 1$. It will be convenient to work sometimes with $d=\Delta-1$. For $z_1, z_2 \in \mathbb{S}$ let $\Arc{[z_1,z_2]}$ and $\Arc{(z_1,z_2)}$ denote the counterclockwise 
arc in $\mathbb{S}$ from $z_1$ to $z_2$ including and excluding the endpoints, respectively. For an 
arc $A$ on the unit circle $\mathbb{S}$, we let $\ell(A)$ denote the length of $A$. We use $\overline{z}$ to denote the conjugate of $z$.

\subsection{Rooted-tree gadgets and complex dynamical systems} \label{sec:dynamics}
Our reductions are based on gadgets that are rooted trees, whose analysis will be based on understanding the dynamical behaviour of certain complex  maps on the unit circle, given by\footnote{\label{footS}Note that, for real $b$ and $\lambda\in \mathbb{S}$, if $z\in \mathbb{S}$ then $f_{\lambda,k}(z)\in \mathbb{S}$ as well.}
\begin{equation}\label{eq:flambdam}
	f_{\lambda,k}: z \mapsto \lambda \cdot \Big(\frac{z + b}{b z + 1}\Big)^k, \mbox{ for integers $k=1,\hdots, d$.}
\end{equation}
We will sometimes drop $\lambda$ when it is clear from the context.
To connect these maps with rooted-tree gadgets, for a graph $G=(V,E)$ and a vertex $u$ of $G$, we let $Z_{G,\pl u}$ be the contribution to the partition function from configurations  with $\sigma(u)=\pl$, i.e., 
\[Z_{G,\pl u}(\lambda,b):=\sum_{\sigma: V\rightarrow\{\pl,\mi \}; \sigma(u)=\pl} \lambda^{|n_\pl(\sigma)|}b^{\delta(\sigma)},\]
and we define analogously $Z_{G,\mi u}$.
\begin{definition}
Let $\lambda,b$ be arbitrary numbers and $T$ be a tree rooted at $r$. We say that $T$ implements the field   $\lambda'$ if $Z_{T,\mi r}(\lambda,b)\neq 0$ and $\lambda'=\tfrac{Z_{T,\pl r}(\lambda,b)}{Z_{T,\mi r}(\lambda,b)}$. We call $\lambda'$ the field of $T$.
\end{definition}

The next lemma explains the relevance of the maps $f_{\lambda,1},\hdots, f_{\lambda,d}$ for implementing fields.
\begin{lemma}
\label{lem: Tree Building}
	Let $b \in (0,1)$ and $\lambda \in \mathbb{S}$.
	Let $T_1,T_2$ be rooted trees with roots $r_1,r_2$ and fields $\xi_1,\xi_2\in \mathbb{S}$, respectively. Then, the tree $T$ rooted at $r_2$ consisting of $T_2$ and $k$ distinct copies 
	of $T_1$ which are attached to $r_2$ via an edge between $r_2$ and $r_1$, implements the field $\xi=f_{\xi_2,k}(\xi_1)\in \mathbb{S}$.
\end{lemma}
\begin{proof}
Omitting for convenience the arguments $\lambda,b$ from the partition functions, we have 
\begin{align*}
Z_{T,\pl r_2}=Z_{T_2,\pl r_2} ( Z_{T_1,\pl r_1}+b Z_{T_1,\mi r_1})^{k}, \quad Z_{T,\mi r_2}=Z_{T_2,\mi r_2} (b Z_{T_1,\pl r_1}+ Z_{T_1,\mi r_1})^{k}.
\end{align*}
Dividing these yields the result (note, $Z_{T_2,\mi r_2}\neq 0$ and $\xi_1=\tfrac{Z_{T_1,\pl r_1}}{Z_{T_1,\mi r_1}}\in \mathbb{S}$, so $Z_{T,\mi r_2}\neq 0$); the fact that $\xi\in \mathbb{S}$ follows from footnote~\ref{footS}. 
\end{proof}
Note in particular that all fields implemented by trees lie on the unit circle $\mathbb{S}$, cf. footnote~(\ref{footS}). The following theorem, which lies at the heart of the construction of the gadgets, asserts that throughout the relevant range of the parameters we can in fact implement a field arbitrarily close to any number in $\mathbb{S}$. We use $\Tc_{d+1}$ to denote the set of all rooted trees with maximum degree $\leq d+1$ whose roots have degree  $\leq d$.
\begin{definition}
Given $b\in (0,1)$ and $d\geq 2$ we denote by $\SQ(d,b)$ the collection of $\lambda\in \SQ$ for which the set of fields implemented by trees in $\Tc_{d+1}$, whose roots have degree $1$, is dense in~$\mathbb{S}$.
\end{definition}

\newcommand{\statethmDensity}{Let $d\geq 2$ be an integer.
\begin{itemize}
\item[(a)] Let  $b \in \big(0, \frac{d-1}{d+1}\big]$ be a rational. Then $\SQ(d,b)=\SQ\setminus\{\pm1\}$.
\item[(b)] Let  $b \in \big(\frac{d-1}{d+1},1\big)$ be a rational. Then $\SQ(d,b)$ is dense in $\mathbb{S}\setminus I(\theta_b)$.
\end{itemize}
}
\begin{theorem} \label{thm: Density of Ratios}
\statethmDensity
\end{theorem}
Theorem~\ref{thm: Density of Ratios} (b) is in stark contrast to what happens for $\lambda\in I(\theta_b)$, where it is known that fields are confined in an arc around 1~\cite{PetersRegts2018}.
We conjecture that in part (b) it is true that $\SQ(d,b)=\SQ\setminus (I(\theta_b)\cup\{-1\})$. 
Moreover, while in Theorem~\ref{thm: Density of Ratios} we focus on rational $b$, which is most relevant for our computational problems, we note that for any real $b\in (0,1)$ $\SQ(d,b)$ is dense in $\mathbb{S}$ in case (a) and dense in $\mathbb{S}\setminus I(\theta_b)$ in case (b).
We suspect that case (a) is true when $\SQ$ is replaced by the collection of algebraic numbers on the unit circle and $b\in (0, \frac{d-1}{d+1}]$ is algebraic, but this seems to be challenging to prove.

Later, in Section~\ref{sec:fast}, we bootstrap Theorem~\ref{thm: Density of Ratios} to obtain fast algorithms to implement fields with arbitrarily small error, see Lemma~\ref{lem:mainlemma} for the exact statement. Roughly, these fields are then used as ``probes'' in our reductions to compute exactly the ratio $\tfrac{Z_{G,\pl v}(\lambda,b)}{Z_{G,\mi v}(\lambda,b)}$ for any graph $G$ and vertex $v$; we say more about this in Section~\ref{sec:redoutline}. For now, we focus on the key Theorem~\ref{thm: Density of Ratios} and the ideas behind its proof.

\subsection{Hardness via Julia-set density }\label{sec:hardJulia} 
To prove Theorem~\ref{thm: Density of Ratios}, we will be interested in the set of values obtained by successive composition of the maps $f_{\lambda,k}$ in \eqref{eq:flambdam} starting from the point $z=1$; the main challenge in our setting is to prove that, for  $\lambda,b$ as in Theorem~\ref{thm: Density of Ratios}, these values are dense on the unit circle $\mathbb{S}$. 
Part (b) is relatively easy to prove, but the real challenge lies in proving part (a).

To understand the reason that this is challenging let us consider the properties of the map $f_{\lambda,k}$ for some root degree $k\geq 1$ viewed as a dynamical system, cf. the upcoming Lemmas~\ref{lem:Julia} and~\ref{lem: properties f} for details. Then, for all $b\in (0,1)$ the following hold.
\begin{enumerate}
\item \label{it:blambda} \textbf{The ``well-behaved'' regime:} When $b\in (\frac{k-1}{k+1},1)$, there exists $\lambda_k=\lambda_k(b)\in \mathbb{S}$ with $\Imag{\lambda_k}>0$ such that for all  $\lambda$ in an arc around 1 given by $\Arc[\overline{\lambda_k},\lambda_k]$,  the iterates of the point $z=1$ under the map $f_{\lambda,k}$ converge to a value $R_k(\lambda)\in \mathbb{S}$. In fact, the map $f_{\lambda,k}$ has nice convergence/contracting properties in an arc around $z=1$: the iterates of any point in  $\Arc{[1,R_k(\lambda)]}$ converge to  $R_k(\lambda)$.
\item \label{it:blambda1} \textbf{The ``chaotic'' regime:} Instead, when  $b\in (\frac{k-1}{k+1},1)$ and $\lambda\in\Arc(\lambda_k, \overline{\lambda_k})$ or $b\in (0,\frac{k-1}{k+1}]$,  all points in $\mathbb{S}$ belong to the so-called Julia set of the map; roughly, this means that the iterates under $f_{\lambda,k}$ of two distinct but arbitrarily close points in  $\mathbb{S}$ will be separated by some absolute constant infinitely many times. In other words, the map $f_{\lambda,k}$ has a chaotic behaviour on $\mathbb{S}$.
\end{enumerate}
For $b\in (0,1)$, we use $\Lambda_k(b)$ to denote the set of $\lambda\in \mathbb{S}$ where the degree-$k$ map $f_{\lambda,k}$ exhibits the behaviour in~(\ref{it:blambda1}), see the relevant Definition~\ref{def:Lambdak} and Lemma~\ref{lem:Julia}. Based on item~(\ref{it:blambda}), it was shown in \cite{PetersRegts2018} that the iterates of the point $z=1$ under the successive composition of the maps in \eqref{eq:flambdam} stay  ``trapped'' in a small arc around 1 when $b\in (\tfrac{d-1}{d+1},1)$ and $\lambda\in \Arc(\overline{\lambda_d},\lambda_d)$. 

Instead, our goal is to tame the chaotic behaviour in item~(\ref{it:blambda1}) to get density on $\mathbb{S}$ for fixed $b\in (0,\tfrac{d-1}{d+1}]$ and $\lambda\in \mathbb{S}\backslash\{\pm 1\}$. We should emphasise  here that, in the range of $b,\lambda$ we consider, the map $f_{\lambda,d}$ has the chaotic behaviour described in item~(\ref{it:blambda1}) throughout $\mathbb{S}$, so by default it is hopeless to aim for any fine analytical understanding, and this is the major technical obstacle we need to address; note, the use of the map $f_{\lambda,d}$ is mandatory to cover all $b$ in $(0,\tfrac{d-1}{d+1}]$.

An analogous setting has been previously considered in \cite{BGGS}, in the context of approximating the independent set polynomial. The bottleneck of showing the desired density is to first argue density around  a point $x^*$ in the Julia set of the degree-$d$ map. Once this is done, the chaotic behaviour of the degree-$d$ map around the Julia-set point $x^*$ can be utilised to bootstrap the density to the whole complex plane. The key challenge here is  arguing the initial density around the Julia-set point of the degree-$d$ map, since the degree-$d$ map itself is useless for creating density in the Julia set. In \cite{BGGS}, an auxiliary Fibonacci-style recursion was used  to converge to such a point $x^*$; the density around $x^*$ was then achieved  by utilising the convergence to further obtain a set of contracting maps around a neighbourhood $N$ of $x^*$, such that the images of $N$ under the maps formed a covering of $N$. 

While the contracting/covering maps framework  can be adapted to our setting (see Lemma~\ref{lem: Density in Arc}), the bottlenect step of obtaining the initial density around the Julia-set point requires a radically different argument: the convergence of the recursion in~\cite{BGGS} relies on a certain linearisation property, which is not present in the case of the ferromagnetic Ising model; even worse, the recursion does not converge for all the relevant range of $b,\lambda$.\footnote{In fact, determining the range of $\lambda$'s where the corresponding recursion for the Ising model converges to a Julia-set point is, to the best of our knowledge, beyond known complex dynamics methods.} 

\subsection{Our approach to obtain density around a Julia-set point}\label{sec:ourapproach} We devise a new technique to tackle the problem of showing density around a point in the Julia set of $f_{\lambda,d}$.  The main idea is to exploit the chaotic behaviour of the iterates of $f_{\lambda,k}$ when $\lambda\in \Lambda_k(b)$ to obtain an iterate $\xi$ of 1 with an expanding derivative, i.e., $|f_{\lambda,k}'(\xi)|>1$. The existence of $\xi$ follows by general arguments from the theory of complex dynamical systems, see the relevant Lemma~\ref{lem: expanding orbit}. The lower bound on the derivative is then used in careful inductive constructions to obtain families  of contracting maps that cover an appropriate arc of the circle. 

To illustrate the main idea of this inductive construction, let us assume that the degree $d+1$ is odd. Then, using Lemma~\ref{lem: expanding orbit} and the fact that $\lambda\in \Lambda_d(b)$, we obtain  an iterate of the point $z=1$ under the map $f_{\lambda,d}$, say $\xi$, so that  $|f_{\lambda,d}'(\xi)|>1$. The key point is to consider the map $f_{\xi,k}$ for $k= d/2$. On one hand, if it happens that $\xi\notin \Lambda_k(b)$ lies in the ``well-behaved'' regime of the degree-$k$ map, it can be shown that the maps $f_{\xi,k},f_{\xi,k+1}$ are contracting/covering maps in an appropriate arc of $\mathbb{S}$, yielding the required density as needed  (details of this argument can be found in Lemma~\ref{lem: easy odd case}). On the other hand, if $\xi\in \Lambda_k(b)$ lies in the ``chaotic'' regime of the degree-$k$ map, then we can proceed inductively by finding an iterate $\nu$ of 1 under the map $f_{\xi,k}$ with expanding derivative $|f_{\xi,k}'(\nu)|>1$ and recurse. 

Technically, to carry out this inductive scheme we have to address the various integrality issues, while at the same time being careful to maintain the degrees of the trees bounded by $\Delta$. More importantly, we need to consider  pairs/triples/quadruples of maps to ensure the contraction/covering property in the inductive step;  to achieve this, we need to understand the dependence of the derivative at the fixpoint of the $k$-ary map with $k$. Here, things turn out to be surprisingly pleasant, since it turns out that $|f_{\lambda,k}'(z)|$ depends linearly on the degree $k$ and is independent of $\lambda$, see item~(\ref{enum: lem1}) of Lemma~\ref{lem: properties f}; this fact is exploited in the arguments of Section~\ref{sec:dependence}. These considerations cover almost all cases, but a few small degrees $d$ remain that we cover by a Cantor-style construction, see Section~\ref{sec:cantor} for details.  

\subsection{The reduction}\label{sec:redoutline} The arguments discussed so far can be used to show that rooted trees in $\Tc_{d+1}$ implement any field $\xi$ on the unit circle $\mathbb{S}$ within arbitrarily small error $\epsilon>0$, see  Lemma~\ref{lem:mainlemma} for the form that we actually need. We now discuss in a bit more detail the high-level idea behind the final reduction argument in Section~\ref{sec:reduction}.

The key observation to utilise the gadgets is that for any graph $G$ and vertex $v$ with $Z_{G,\mi v}(\lambda,b) \not= 0$, the ``field'' at a vertex $v$ satisfies $\tfrac{Z_{G,\pl v}(\lambda,b)}{Z_{G,\mi v}(\lambda,b)}\in \mathbb{S}$, cf. Lemma~\ref{lem:conjugation}, and hence we can use our rooted-tree gadgets as probes to compute exactly the ratio $Q_{G,v}:=\tfrac{Z_{G,\pl v}(\lambda,b)}{Z_{G,\mi v}(\lambda,b)}$; the straightforward way to do this would be to attach a tree on $v$ which implements a field $x\in \mathbb{S}$ and use oracle calls to either $\FactorIsing{K}$ and $\ArgIsing{\rho}$  and look for $x=x^*$ that makes the partition function of the resulting graph equal to zero; from the key observation earlier, we know that such an $x^*$ exists, namely $x^*=-1/Q_{G,v}$, and, to determine it, we can use binary search.

This is the main idea behind the reduction, though there are a couple of caveats. First of all, there is no way to know whether  the ratio $Q_{G,v}$ is well-defined, i.e., whether $Z_{G,\mi v}(\lambda,b)\neq 0$, even using oracle calls to the approximation problems we study: $Z_{G,\mi v}(\lambda,b)$ is not a partition function of a graph (since $v$'s spin is fixed), and even if we managed to cast this as a partition function, the oracles cannot detect zeros (cf. Section~\ref{sec:results}). The second caveat is that attaching the tree increases the degree of $v$ which is problematic when, e.g., $G$ is $\Delta$-regular, and the  peeling-vertices argument does not quite work since there is no simple way to utilise the oracles after the first step.

The first point is addressed by replacing the edges of $G$ with paths of appropriate length which has the effect of  ``changing'' the value of the parameter $b$ to some value $\hat{b}$ close but not equal to $1$ where the partition function is zero-free (we actually need to attach to internal vertices of the paths rooted trees with fields close to $1/\lambda$ so that the complex external field $\lambda$ is almost cancelled). Then, using oracle calls to $\FactorIsing{K}$ or $\ArgIsing{\rho}$, our algorithm aims to determine the value of $Z_G(\lambda,\hat{b})$ which is a $\numP$-hard problem (\cite{KowCai16}, Theorem~1.1); the key-point is that now we have zero-freeness of the partition function which allows us to assert that the quantities we compute during the course of the algorithm are actually well-defined. 

The second point is addressed by doing the peeling-argument at the level of edges by trying to figure out, for an edge $e$ of $G$, the value of the ratio $\hat{Q}_{G,e}=\tfrac{Z_{G}(\lambda,\hat{b})}{Z_{G,\backslash e}(\lambda,\hat{b})}$. We do this by subdividing the edge and use a field gadget on the middle vertex; this has the benefit that it does not increase the maximum degree of the graph but certain complications arise since instead of  $\hat{Q}_{G,e}$ we retrieve a slightly different ratio, see Lemma~\ref{lem:maincomp} in Section~\ref{sec:reductionstep}, and some extra work is required to finish off the proof of Theorem~\ref{thm:main2}, see Section~\ref{sec:proofmain} for details.

\subsection{Outline}\label{sec:out-line} The next section details the  dynamical properties of the maps $f_{\lambda,k}$ and  elaborates on the inductive proof of Theorem~\ref{thm: Density of Ratios}, which is based on the upcoming Lemma~\ref{lem: Many Cases Lemma}. Section~\ref{sec:contracting} explains in more detail the contracting/covering maps framework and how we utilise the degree/derivative interplay to cover the bulk of the cases in Lemma~\ref{lem: Many Cases Lemma}. Section~\ref{sec:special} contains the remaining pieces needed to complete the proof of Lemma~\ref{lem: Many Cases Lemma}, which is given in Section~\ref{sec: Proof of the main lemma}. In Section~\ref{sec:fast}, we bootstrap Theorem~\ref{thm: Density of Ratios} to obtain fast algorithms to implement fields on the unit circle with arbitrarily small precision error, which is used in the reduction arguments of Section~\ref{sec:reduction}, where the proof of Theorem~\ref{thm:main2} is completed. Finally, in Section~\ref{sec:minusone}, we study the case $\lambda=-1$ (cf. Remark~\ref{rem:real}), and show the equivalence with the problem of approximately counting perfect matchings.

\section{Complex Dynamics Preliminaries and the Inductive Step in Theorem~\ref{thm: Density of Ratios}}

In this section, we set up some preliminaries about the maps $f_{\lambda,k}$ in \eqref{eq:flambdam} that will be used to prove Theorem~\ref{thm: Density of Ratios}. We first consider the general case $k\geq 1$ in Section~\ref{sec:genm} and then further study  the $k = 1$ case separately in Section~\ref{sec:flone}. In Section~\ref{sec:expanding}, we use these properties to obtain points with expanding derivatives using tools from complex dynamics.  Then, in Section~\ref{sec:induction}, we give the main lemma that lies at the heart of the inductive proof of Theorem~\ref{thm: Density of Ratios} and conclude the proof of the latter.

\subsection{Results on $f_{\lambda,k}$ for general $k$}\label{sec:genm}
This section contains relevant properties of the maps $f_{\lambda,k}: z \mapsto \lambda \cdot \big(\frac{z + b}{b z + 1}\big)^k$ that we will need; these were discussed informally in Section~\ref{sec:hardJulia}, and here we formalise them. Almost all results of this section follow from arguments in \cite{PetersRegts2018}.

We begin by defining formally the set $\Lambda_k(b)$.
\begin{definition}\label{def:Lambdak}
Let $k\geq 1$ be an integer and $b\in (0,1)$. We let $\Lambda_k(b)$ be the set of $\lambda\in \mathbb{S}$ such that all fixed points $z$ of the map $f_{\lambda,k}$ with $z\in \mathbb{S}$ are repelling, i.e., $|f'_{\lambda,k}(z)|>1$.
\end{definition}

The following lemma gives a description of the set $\Lambda_k(b)$ and characterises the Julia set of $f_{\lambda,k}$. We have described informally the dynamical properties of the map $f_{\lambda,k}$ that the Julia set captures, see item~(\ref{it:blambda1}) in Section~\ref{sec:hardJulia}. The reader is referred to \cite[Chapter 4]{Milnor2006} for more details on the general theory.
\begin{lemma}\label{lem:Julia}
Let $k\geq 1$ be an integer. Then, 
\begin{itemize}
\item if $b\in (0,\frac{k-1}{k+1})$, $\Lambda_k(b)=\mathbb{S}$. For $b=\frac{k-1}{k+1}$, $\Lambda_k(b)=\mathbb{S}\backslash \{+1\}$.
\item if $b\in (\frac{k-1}{k+1}, 1)$, there is $\lambda_k=\lambda_k(b) \in \mathbb{S}$ with  $\Imag{(\lambda_k)} > 0$ such that  $\Lambda_k(b)=\Arc(\lambda_k,\overline{\lambda_k})$.
\end{itemize}
Moreover, if $k>1$, then for all $\lambda\in \Lambda_k(b)$, the Julia set of $f_{\lambda,k}$ is equal to the unit circle $\mathbb{S}$.
\end{lemma}
\begin{proof}
For $b\in (\frac{k-1}{k+1}, 1)$, the range of $\Lambda_k(b)$  follows from \cite[Theorem~14]{PetersRegts2018}. For $b\in (0,\frac{k-1}{k+1}]$, the range of $\Lambda_k(b)$  follows from item~(\ref{enum: lem1}) in Lemma~\ref{lem: properties f} below. The characterisation of the Julia set for $\lambda\in \Lambda_k(b)$ is shown in \cite[Proof of Proposition 17]{PetersRegts2018}.
\end{proof}

\begin{remark}
The $\lambda_k(b)$ of the lemma is equal to $e^{i\theta_b}$ from the statement of Theorem~\ref{thm:main2} (where we replace $\Delta$ by $k+1$).
\end{remark}

Let $A$ be an arc of $\mathbb{S}$. A map $f:A\rightarrow \mathbb{S}$ is orientation-preserving if for any $z,z_1,z_2\in A$ with $z\in \Arc{[z_1,z_2]}$ it holds that  $f(z)\in \Arc{[f(z_1),f(z_2)}]$. The orbit of a point $z_0\in \mathbb{S}$ under the map $f_{\lambda,k}$ is the sequence of the iterates $\{f_{\lambda,k}^n(z_0)\}_{n\geq 0}$. A fixed point $z$ of the map $f_{\lambda,k}$ is called attracting if $|f'_{\lambda,k}(z)|<1$ and parabolic if $|f'_{\lambda,k}(z)|=1$.

The following lemma captures properties of the maps $f_{\lambda,k}$ when $k\in\{1,\hdots,d-1\}$ in the regime $b \in \big(\frac{d-2}{d},\frac{d-1}{d+1}\big]$, which turns out to be the hard part of the proof of Theorem~\ref{thm: Density of Ratios} (the lemma is stated more generally for $b \in \big(\frac{d-2}{d},1\big)$).
\begin{lemma}
	\label{lem: properties f}
	Let $d \in \mathbb{Z}_{\geq 2}$ and let $b \in \big(\frac{d-2}{d},1\big)$.
	Then the following holds:
	\begin{enumerate}[(i)]
		\item
		\label{enum: lem1}
		For all $\lambda \in \mathbb{S}$ and $k \in \mathbb{Z}_{\geq 1}$ the map 
		$f_{\lambda,k}: \mathbb{S} \to \mathbb{S}$ is orientation-preserving. Also,  
		the magnitude of the derivative at a point $z \in \mathbb{S}$ does not 
		depend on $\lambda$ and equals $|f_{k}'(z)|$ where
		\begin{equation}
			\label{eq: Derivative}
			\left| f_{k}'(z) \right| = k\cdot \left| f_{1}'(z) \right| 
			= \frac{k(1-b^2)}{b^2 + 2b \cdot\Real{(z)} + 1}.
		\end{equation}
		\item
		\label{enum: lem2}
		For $k \in \{1,\dots, d-1\}$, let $\lambda_k=\lambda_k(b)\in \mathbb{S}$ be as in Lemma~\ref{lem:Julia}. Then,   
		\begin{itemize}
		\item if $\lambda \in \Arc{(\overline{\lambda_k},\lambda_k)}$, then $f_{\lambda,k}$ 		has a unique attracting fixed point $R_k(\lambda) \in \mathbb{S}$.  
		\item if $\lambda =\overline{\lambda_k}$ or $\lambda_k$, then $f_{\lambda,k}$ 		has a unique parabolic fixed point $R_k(\lambda) \in \mathbb{S}$. 
		\end{itemize}
		\item
		\label{enum: lem3}
		The fixed point maps $R_k: \Arc{[\overline{\lambda_k},\lambda_k]} \to \mathbb{S}$ are continuously
		differentiable on $\Arc{(\overline{\lambda_k},\lambda_k)}$ and orientation-preserving with 
		the property that $R_k(1) = 1$ and $R_k(\overline{\lambda}) = \overline{R_k(\lambda)}$.
		\item
		\label{enum: lem4}
		For $\lambda \in \Arc{(1,\lambda_k]}$, the fixed point $R_k(\lambda)$ is in 
		the upper half-plane. For $z_0 \in \Arc{[1,R_k(\lambda)]}$ the orbit of
		$z_0$ under iteration of $f_{\lambda,k}$ converges to $R_k(\lambda)$ and is contained in
		$\Arc{[z_0,R_k(\lambda)]}$.
		\item
		\label{enum: lem5}
		The following inequalities hold:
		\[
			\Arg{(\lambda_{d-1})} < \Arg{(\lambda_{d-2})} < \cdots < \Arg{(\lambda_1)},
		\]
		while for $\lambda \in \Arc{(1,\lambda_{m}]}$, with $m\leq d-1$, we have
		\[
			\Arg{(R_1(\lambda))} < \Arg{(R_2(\lambda))} < \cdots < \Arg{(R_m(\lambda))},
		\]
		with the additional property that, for $i \in \{1, \dots, m-2\}$,  
		\begin{equation}
			\label{eq: increasing intervals}
			\ell\left(\Arc{[R_{i}(\lambda),R_{i+1}(\lambda)]}\right) 
				\leq
			\ell\left(\Arc{[R_{i+1}(\lambda),R_{i+2}(\lambda)]}\right).
		\end{equation}
	\end{enumerate} 
\end{lemma}

\begin{proof}
We refer to \cite{PetersRegts2018} for proofs of items (\ref{enum: lem1})--(\ref{enum: lem4}). Specifically, item~(\ref{enum: lem1}) follows from \cite[Lemma~8 \& Equation~(3.1)]{PetersRegts2018}, item~(\ref{enum: lem2})  from \cite[Lemma~13, Theorem~14, Proof of Proposition~17]{PetersRegts2018}, item~(\ref{enum: lem3})  from~\cite[Proof of Theorem~14]{PetersRegts2018}, and item~(\ref{enum: lem4})  from~\cite[Theorem~14, Proof of Theorem 5(i)]{PetersRegts2018}.

We will prove item (\ref{enum: lem5}). By
taking the derivative of both sides of the equality $f_{\lambda,k}(R_k(\lambda)) = R_k(\lambda)$
with respect to $\lambda$ and rewriting we obtain 
\[
	R_k'(\lambda) = \frac{R_k(\lambda)}{\lambda\Big(1 - f_{\lambda,k}'\big(R_k(\lambda)\big)\Big)}.
\]
Using equation (\ref{eq: Derivative}) for $z = R_m(1) = 1$ we obtain that $R_{i+1}'(1) > R_{i}'(1)$ 
for $i \in \{1, \dots, d-2\}$ and thus for $\lambda \in \mathbb{S}$ in the upper half-plane near $1$ 
we find that $\Arg{(R_{i+1}(\lambda))} > \Arg{(R_{i}(\lambda))}$. 

The derivative at a fixed point of a map of the unit circle to itself is real (see also \cite[Lemma 11]{PetersRegts2018}). Furthermore, if such a map is orientation-preserving the derivative
at a fixed point is positive. Because the map $f_{\lambda,i}$
is orientation-preserving with attracting fixed point $R_i(\lambda)$ we find that 
$f_{\lambda,i}'(R_i(\lambda)) \in (0,1)$ for $\lambda \in \Arc{(\overline{\lambda_i}, \lambda_i)}$. 
From this we deduce that we can write 
\begin{equation}
	\label{eq: R_m der}
	\left|R_k'(\lambda)\right| = \frac{1}{1 - f_{\lambda,k}'\left(R_k(\lambda)\right)}.
\end{equation}
From equation (\ref{eq: Derivative}) it can be seen that $\left|f_{i}'(z)\right|$ is increasing 
both with respect to $\Arg{(z)}$ when $\mathrm{Im}(z)>0$ and with respect to the index $i$ and thus, as long as $R_{i}(\lambda)$ and $R_{i+1}(\lambda)$
are both defined and $\Arg(R_{i+1}(\lambda)) > \Arg(R_{i}(\lambda))$, we deduce
that $|R_{i+1}'(\lambda)| > |R_{i}'(\lambda)|$. Since $\Arg(R_{i+1}(\lambda)) > 
\Arg(R_{i}(\lambda))$ for $\lambda$ in the upper half-plane close to $1$, we conclude that there 
cannot be any $\lambda$ in the upper half-plane such that $\Arg(R_{i+1}(\lambda)) \leq \Arg(R_{i}(\lambda))$.

Now suppose that there is some index $i$ such that $\Arg{(\lambda_i)} \leq \Arg{(\lambda_{i+1})}$.
Note that $R_i(\lambda_i)$ is a parabolic fixed point of $f_{\lambda_i,i}$, which means that
$f_{\lambda_i,i}'(R_i(\lambda_i)) = 1$. 
Because we assumed that $\Arg{(\lambda_i)} \leq \Arg{(\lambda_{i+1})}$
we see from item (\ref{enum: lem2}) that $R_{i+1}(\lambda_i)$ must be well defined. We already deduced that 
$\Arg{(R_{i+1}(\lambda_i))} > \Arg{(R_{i}(\lambda_i))}$ and thus 
$f_{\lambda_i,i+1}'(R_{i+1}(\lambda_i)) > f_{\lambda_i,i}'(R_i(\lambda_i)) = 1$, which contradicts 
the fact that $R_{i+1}(\lambda_i)$ is an attracting fixed point of $f_{\lambda_i,i+1}$. This
concludes the proof of the first two claims of item (\ref{enum: lem5}).

Finally, we show the final claim of item (\ref{enum: lem5}). For indices $0\leq i\leq j\leq m$, it will be convenient to denote by $A_{i,j,\lambda}$ the arc $\Arc{[R_{i}(\lambda),R_{j}(\lambda)]}$, under the convention that $R_0(\lambda)=1$. For $\lambda \in \Arc{(1,\lambda_{m}]}$, our goal is hence to show that 	$\ell\left(A_{i,i+1,\lambda}\right) \leq 	\ell\left(A_{i+1,i+2,\lambda}\right)$ for all $i\in \{1, \dots, m-2\}$. 

 For any $\lambda \in \Arc{(1,\lambda_{k}]}$ we observe for $i = 1, \dots, k$ that
\[
	\ell(A_{0,i,\lambda}) = \int_{\Arc{[1,\lambda]}} |R_i'(z)| \,|dz|
\]
and thus for $i \in \{1, \dots, k-1\}$ we have $\ell(A_{i,i+1,\lambda}) = \int_{\Arc{[1,\lambda]}} |R_{i+1}'(z)| - |R_{i}'(z)| \,|dz|$.

We  first show item (\ref{enum: lem5}) for $\lambda$ near $1$. As we let $\lambda$ approach $1$ along the circle, we obtain that 
\begin{align*}
	\lim_{\lambda \to 1} \frac{\ell(A_{i,i+1,\lambda}) }{\len{\Arc{[1,\lambda]}}}
	&= 
	|R_{i+1}'(1)| - |R_{i}'(1)| = \frac{1}{1 - f_{\lambda,i+1}'(1)} - \frac{1}{1 - f_{\lambda,i}'(1)} \\
	&=\frac{(1+b)/(1-b)}{\big(i - (1+b)/(1-b)\big)\big(i-2b/(1-b)\big)}.
\end{align*}
The second equality can be obtained by using (\ref{eq: R_m der}) and the 
third by using (\ref{eq: Derivative}) and simplifying. If we denote this expression by 
$g(i)$ then it is not hard to see that $g(i+1) > g(i)$ as long as $i + 1 < 2b/(1-b)$.
Because $b\in (\tfrac{d-2}{d},1)$, we have  $2b/(1-b)>d-2$. 
So indeed $g(i+1) > g(i)$ for $i \in \{1, \dots, d-3\}$, which contains
$\{1, \dots, k-2\}$. This shows that inequality in (\ref{eq: increasing intervals}) is 
true for $\lambda$ near $1$. 

Now suppose that there is $\lambda\in \mathbb{S}$ and index $i$ for which the inequality
in (\ref{eq: increasing intervals}) does not hold. Then, by continuity,
because the inequality does hold near $1$, there is $\lambda\in \mathbb{S}$ for which the inequality is an equality, i.e.,
\begin{equation}
\label{eq: assume equality}
	\ell\left(A_{i,i+1,\lambda}\right) 
		=
	\ell\left(A_{i+1,i+2, \lambda}\right).
\end{equation}
For convenience, we will henceforth drop the subscript $\lambda$ from the notation for the arcs $A_{i,j,\lambda}$ and simply write $A_{i,j}$. The maps $f_{\lambda,j}$ are orientation-preserving for any $j$ and thus
\[
\len{\Arc[\lambda,R_j(\lambda)]} = \len{f_{\lambda,j}(A_{0,j})}
= \int_{A_{0,j}}  |f_j'(z)|\, |dz|.
\]
Using this equality we can write
\[
\len{A_{j,j+1}} 
= \int_{A_{0,j+1}} |f_{j+1}'(z)|\, |dz| - \int_{A_{0,j}}|f_j'(z)|\, |dz|.
\]
We use this equality for $j= i$ and $j= i+1$ and rearrange \eqref{eq: assume equality}
to obtain
\[
2\int_{A_{0,i+1}}  |f_{i+1}'(z)|\, |dz| =  \int_{A_{0,i+2}} |f_{i+2}'(z)|\, |dz| +
\int_{A_{0,i}}  |f_i'(z)|\, |dz|.
\]
We use \eqref{eq: Derivative} to rewrite the left-hand side of this equation as
\[
2(i+1) \int_{A_{0,i}}  |f_{1}'(z)|\, |dz|
+ 2 \int_{A_{i,i+1}} |f_{i+1}'(z)|\, |dz|
\]
and we rewrite the right-hand side as
\begin{align*}
(i+2)\int_{A_{0,i}}  |f_{1}'(z)|\, |dz| +
\int_{A_{i,i+2}}  |f_{i+2}'(z)|\, |dz| +
i\cdot \int_{A_{0,i}}  |f_{1}'(z)|\, |dz|.
\end{align*}
These two being equal implies that
\[
2 \int_{A_{i,i+1}}  |f_{i+1}'(z)|\, |dz| 
=
\int_{A_{i,i+2}} |f_{i+2}'(z)|\, |dz|.
\]
We will show that this yields a contradiction. We rewrite the right-hand side as 
\[
\int_{A_{i,i+2}}  |f_{i+2}'(z)|\, |dz| 
= 
\int_{A_{i,i+1}}  |f_{i+2}'(z)|\, |dz|
+ 
\int_{A_{i+1,i+2}}  |f_{i+2}'(z)|\, |dz|.
\]
and we will show that both summands are greater than $\int_{A_{i,i+1}} |f_{i+1}'(z)|\, |dz|$,
which will yield the contradiction. The inequality for the first summand follows
easily from the fact that $|f_{i+2}'(z)| > |f_{i+1}'(z)|$ for all $z \in \mathbb{S}$, cf. \eqref{eq: Derivative}. 
The second inequality uses the fact that  $|f_{i+2}'(z)|$ increases as
$\Arg(z)$ increases (when $\Imag{z} > 0$) and thus
\begin{align*}
\int_{A_{i+1,i+2}} |f_{i+2}'(z)|\, |dz|
&>
|f_{i+2}'(R_{i+1}(\lambda))| \cdot \len{A_{i+1,i+2}} > |f_{i+1}'(R_{i+1}(\lambda))| \cdot \len{A_{i,i+1}} \\
&> \int_{A_{i,i+1}} |f_{i+1}'(z)|\, |dz|.
\end{align*}
The second inequality of this derivation uses the assumed equality in~\eqref{eq: assume equality}. 
This yields the desired contradiction.
\end{proof}

\begin{remark} \label{rem:increase}
For any $\lambda \in \mathbb{S}$  and $k \in \mathbb{Z}_{\geq 1}$ for which the fixed point $R_k(\lambda)$ is defined we have $|f_k'(R_k(\lambda))| >
|f_k'(\lambda)|$. To see this when $\Imag{\lambda}>0$, note  from 
item~{(\ref{enum: lem4})} and the fact that the maps $f_{\lambda,k}$ are orientation-preserving that $\Arg{(R_k(\lambda))}\in (\Arg{(\lambda)},\pi)$ and hence by \eqref{eq: Derivative} that  $|f_k'(R_k(\lambda))| >
|f_k'(\lambda)|$. When $\Imag{\lambda}<0$, the inequality follows from the above since $\overline{R_k(\lambda)}=R_k(\overline{\lambda})$ from  item~(\ref{enum: lem3}) and the expression in \eqref{eq: Derivative} depends only the real part of $z$.

\end{remark}

\subsection{Results on $f_{\lambda,1}$}\label{sec:flone}
Note that $f_{\lambda,1}$ is a M\"obius transformation; we will extract some relevant information about it using the theory of M\"obius transformations, following \cite[Section 4.3]{Beardon1995}. 

  There is a natural way to relate each M\"obius transformation $g$ with a $2\times 2$ matrix $A$. Formally, let $\textrm{GL}_2(\mathbb{C})$ be the group of $2\times 2$ invertible matrices with complex entries (with the multiplication operation) and $\mathcal{M}$ be the group of M\"obius transformations (with the composition operation $\circ$).  The following map  gives a surjective homomorphism between the groups $\textrm{GL}_2(\mathbb{C})$ and $\mathcal{M}$:
\[
	\Phi:\textrm{GL}_2(\mathbb{C}) \to \mathcal{M},\quad 
	\Big(
	\begin{array}{cc}
	 a & b \\
	 c & d \\
	\end{array}
	\Big) \mapsto (z \mapsto \frac{a z + b}{c z + d}).
\]
For $g \in \mathcal{M}$, let $A \in \textrm{GL}_2(\mathbb{C})$ such that $\Phi(A) = g$ and 
define $\tr^2(g) = \tr(A)^2/\det(A)$. This value does not depend on the choice of preimage
and thus $\tr^2$ is a well defined operator on $\mathcal{M}$. In the following theorem it is
stated how this operator is used to classify M\"obius transformations. We say that $f,g\in \mathcal{M}$ are conjugate if there is
$h\in \mathcal{M}$ such that $f=h\circ g\circ h^{-1}$.

\begin{theorem}[{\cite[Theorem 4.3.4]{Beardon1995}}]
\label{thm: Mobius classification}
	Let $g \in \mathcal{M}$ not equal to the identity, then $g$ is conjugate to
	\begin{enumerate}
		\item \label{it:case33434}
	 	a rotation $z \mapsto e^{i \theta} z$ for some $\theta \in (0,\pi]$ if and 
		only if $\tr^2(g) \in [0,4)$, in which case $\tr^2(g) = 2 \cdot \left(\cos (\theta )+1\right);$
		\item \label{it:case33434b}
		a multiplication $z \mapsto e^{\theta} z$ for some $\theta \in \mathbb{R}_{> 0}$ if and 
		only if $\tr^2(g) \in (4,\infty)$, in which case $\tr^2(g) = 2 \cdot \left(\cosh (\theta )+1\right).$
	\end{enumerate}
\end{theorem}
In case~(\ref{it:case33434}), $g$ is said to be \emph{elliptic}, while in case~(\ref{it:case33434b}) $g$ is 
called \emph{hyperbolic}. If $\tr^2(g) = 4$ the map is called \emph{parabolic}.

\begin{corollary}
\label{cor: mobius parameters}
	Let $b \in (0,1)$ and let $\lambda_1=\lambda_1(b) \in \mathbb{S}$ be as in Lemma~\ref{lem:Julia}. The map $f_{\lambda,1}$ is hyperbolic when $\lambda \in \Arc{(\overline{\lambda_1},\lambda_1)}$ and $f_{\lambda,1}$ is elliptic when $\lambda \in \Arc{(\lambda_1, \overline{\lambda_1})}$.
\end{corollary}

\begin{proof}
Write $\lambda = x + i y$ with $x,y \in \mathbb{R}$ such that $x^2 + y^2 = 1$. A short 
calculation gives that
\[
	\tr^2(f_{\lambda,1}) = \frac{2 \left(x + 1\right)}{1 - b^2}.
\]
The value of $tr^2(f_{\lambda,1})$ strictly increases from $0$ to $4/(1-b^2)$ as $x$ increases 
from $-1$ to $1$. It follows that there is a unique value $x \in (-1,1)$ such that
$\tr^2(f_{\lambda,1}) = 4$. This value must coincide with $\textrm{Re}(\lambda_1)$, where $\lambda_1=\lambda_1(b) \in \mathbb{S}$ is as in Lemma~\ref{lem:Julia}, completing the proof.
\end{proof}

\begin{lemma}
\label{lem: elliptic -> irrational rotation}
	Let $b \in (0,1)$ be a rational. Suppose that $\xi \in \SQ$ with $\xi\neq \pm 1$ is such 
	that $f_{\xi,1}$ is elliptic. Then $f_{\xi,1}$ is conjugate to an irrational rotation.
\end{lemma}



\begin{proof}
Let $\xi = x + i y$ with $x,y \in \mathbb{Q}$ such that $x^2 + y^2 = 1$. Because $f_{\xi,1}$
is elliptic it is conjugate to a rotation $z \mapsto e^{i \theta} z$ with
\begin{equation}
	\label{eq: Conjugation equality}
	2 \cdot \left(\cos (\theta )+1\right) = \frac{2 \left(x + 1\right)}{1 - b^2}.
\end{equation}
Let $t = 2 \cdot \left(\cos (\theta )+1\right)$.
Suppose $\theta$ is an angle corresponding to a rational rotation, i.e., if we let $z = e^{i \theta}$ 
then there is a natural number $n$ such that $z^n = 1$. It follows that then also $\overline{z}^n = 1$ 
and thus both $z$ and $\overline{z}$ are algebraic integers. Therefore $z + \overline{z} = 2\cos(\theta)$ 
is an algebraic integer. It follows that $t$ is an algebraic integer, while the right-hand side of 
(\ref{eq: Conjugation equality}) shows that $t$ must also be rational. Because the only rational algebraic 
integers are integers we can conclude that $t$ must be an integer and thus $t \in \{0,1,2,3\}$. If $t=0$ we 
see that $\xi = x = -1$, which we excluded, so only three possible values of $t$ remain. Let
$X = t(1+b)/(1-b)$ and $Y = 2ty/(1-b)^2$ then $(X,Y)$ is a rational point on the elliptic curve
$E_t$ given by the following equation:
\[
	E_t: Y^2 = X^3 - (t-2)t\cdot X^2 + t^2 \cdot X.
\]
The set of rational points of an elliptic curve together with an additional point has a group
structure that is isomorphic to $\mathbb{Z}^r \times \mathbb{Z}/N\mathbb{Z}$. The number $r \geq 0$
is called the rank of the curve and the subgroup isomorphic to $\mathbb{Z}/N\mathbb{Z}$ is called 
the torsion subgroup. The rank and the torsion subgroup of a particular curve can be found
using a computer algebra system. Using \emph{Sage}, if the variable \verb|t| is declared
to be either $1,2$ or $3$, the curve $E_t$ can be defined with the code
\verb|Et = EllipticCurve([0, -(t-2)*t, 0, t**2, 0])|.
The rank and the torsion subgroup can subsequently be found with the commands \verb|Et.rank()|
and \verb|Et.torsion_subgroup()|. We find that $E_t$ has rank $0$ for all $t \in \{1,2,3\}$. For
$t \in \{1,3\}$ the torsion subgroup is isomorphic to $\mathbb{Z}/2\mathbb{Z}$ and for $t=2$
it is isomorphic to $\mathbb{Z}/4\mathbb{Z}$. This means that there is one rational point on
$E_1$ and $E_3$, which we can see is the point $(0,0)$, and there are three rational points on 
$E_2$, namely $\{(0,0), (2, \pm 4)\}$. These points do not correspond to 
values of $b$ within the interval $(0,1)$, which means that $\theta$ cannot correspond to a 
rational rotation.
\end{proof}

\subsection{Obtaining points with expanding derivatives}\label{sec:expanding}
In this section, we use the dynamical study of the maps $f_{\lambda,k}$ from previous sections, to conclude the existence of points with expanding derivatives. More precisely, we show the following. 
\newcommand{\statelemexpanding}{Let $b \in (0,1)$, $k\geq 1$ be an integer,	and $\xi\in \Lambda_k(b)$ with $\xi\neq -1$. 	Let $z_0 \in \mathbb{S}$ and let $z_n = f_{\xi,k}^n(z_0)$ for $n>0$. 
	Then there is some index $m$ such that $|f_{\xi,k}'(z_m)| > 1$.}
\begin{lemma}\label{lem: expanding orbit}
\statelemexpanding
\end{lemma}

\begin{proof}
For $k = 1$ it follows from Corollary~\ref{cor: mobius parameters} that 
$f_{\xi,1}$ is conjugate to a rotation. If $f_{\xi,1}$ is conjugate to
an irrational rotation then the orbit of any initial point $z_0$
will get arbitrarily close to $-1$ for which $|f_{1}'(-1)| = \frac{1+b}{1-b} > 1$. 
Otherwise, if $f_{\xi,1}$ is conjugate to a rational rotation, there is 
an integer $N > 1$ such that $f_{\xi,1}^N(z) = z$ for all $z$; consider the smallest such integer $N$. Let $\theta \in (0,\pi]$
be the angle such that $f_{\xi,1}$ is conjugate to the rotation
$z \mapsto e^{i\theta}\cdot z$. Equation~\eqref{eq: Conjugation equality}
then states that 
\[
2 \cdot \left(\cos (\theta )+1\right) = \frac{2 \left(\Real(\xi) + 1\right)}{1 - b^2}.
\]
If $N=2$, then $\theta = \pi$ and thus $\Real(\xi) = -1$ contradicting $\xi \neq -1$. Hence, $N >2$. From $f_{\xi,1}^N(z) = z$, we obtain
\begin{equation}
\label{eq: derivative rational rotation}
 \prod_{n=0}^{N-1} f_{\xi,1}'(z_n) = (f_{\xi,1}^N)'(z_0) = 1.
\end{equation}
From~{(\ref{eq: Derivative})} there 
are precisely two values of $w\in \mathbb{S}$ such that $|f_{\xi,1}'(w)| = 1$. Because $N > 2$ and $N$ is the smallest integer such that $f_{\xi,1}^N(z_0) = z_0$, we conclude there is at least one term, say with index $m$, 
of the product in {(\ref{eq: derivative rational rotation})}
for which $|f_{\xi,1}'(z_m)| > 1$.

Consider now the case $k \geq 2$ and denote $f = f_{\xi,k}$. By Lemma~\ref{lem:Julia}, for $\xi\in \Lambda_{k}(b)$ the Julia set of $f$ is the circle $\mathbb{S}$.
In \cite[Proof of Proposition 17]{PetersRegts2018}, it is shown that the two Fatou components of $f$, denoted by $\mathbb{D}$ and $\overline{\mathbb{D}}^c$,
are attracting basins and contain the critical points $-b$ and $-1/b$. From \cite[Theorem 19.1]{Milnor2006}, we therefore conclude that the map $f$ is hyperbolic, i.e., there
exists a conformal metric $\mu$ on a neighborhood of $\mathbb{S}$ such that 
$||D_z f||_\mu \geq \kappa > 1$ for a constant $\kappa$ and all $z \in \mathbb{S}$. Because 
$\mathbb{S}$ is compact and the metric $\mu$ is conformal there is a 
constant $c > 0$ such that $|g'(z)| > c \cdot ||D_z g||_\mu$
for all $z \in \mathbb{S}$ and maps $g: \mathbb{S} \to \mathbb{S}$. If follows that
for all $N > 0$
\[
 \prod_{n=0}^{N-1} |f'(z_n)| = |(f^N)'(z_0)| > c \cdot ||D_{z_0} f^N||_\mu \geq c \cdot \kappa^N.
\]
There is an $N > 0$ such that the right-hand side of this equation is greater than $1$. 
The product on the left-hand side of the equation shows that for such an $N$ there must
be at least one index $m \in \{0,\dots,N-1\}$ such that $|f'(z_m)| > 1$.
\end{proof}

\begin{lemma}
\label{lem: tree beyond lambda_k}
	Let $b \in (0,1)$, $\lambda \in \mathbb{S}\setminus\{\pm 1\}$ and
	$d, k \in \mathbb{Z}_{\geq 1}$.
	Supppose there exists a rooted tree in $\Tc_{d+1}$ whose root degree $m$ is at most $d-k$ and which implements a field
	$\xi\in \Lambda_k(b)\backslash\{-1\}$. 
	
	Then there is $\sigma \in \mathbb{S}$ with $|f_k'(\sigma)| > 1$ and a sequence of 
	rooted trees $\{T_n\}_{n \geq 1}$ in $\Tc_{d+1}$  with root degrees at most $m+k$ which implement a sequence of fields $\{\zeta_n\}_{n \geq 1}$ such 
	that $\zeta_n$ approaches $\sigma$ without being equal to $\sigma$.
\end{lemma}

\begin{proof}
Consider the orbit 
\[
	\mathcal{S} = \left\{f_{\lambda,1}^n(1): n \geq 1 \right\}.
\]
Note that the elements of $\mathcal{S}$ are fields of paths.
We have seen in Section~\ref{sec:flone} that either $f_{\lambda,1}$ is conjugate to an irrational 
rotation or the orbit of $1$ tends towards an attracting or a parabolic fixed point. 
In either case there is $\sigma_0 \in \mathbb{S}$
such that $\sigma_0 \not \in \mathcal{S}$ and the elements of $\mathcal{S}$ 
accumulate on $\sigma_0$. It follows from Lemma~\ref{lem: expanding orbit}
that there is a positive integer $N$ such that $\sigma := f_{\xi, k}^N(\sigma_0)$ 
has the property $|f_{k}'(\sigma)| > 1$.
Now define
\[
	\mathcal{R} = \left\{f_{\xi, k}^N(s): s \in \mathcal{S} \right\}.
\]
By assumption, $\xi$ can be implemented by a rooted tree in $\Tc_{d+1}$ with root degree $m\leq d-k$, so by applying inductively Lemma~\ref{lem: Tree Building}, the elements of $\mathcal{R}$ are fields of trees in $\Tc_{d+1}$ whose root degrees is $m+k\leq d$. There is a sequence
$\{\zeta_n\}_{n \geq 1} \subseteq \mathcal{R}$ accumulating on $\sigma$
without being equal to $\sigma$, which is what we wanted to show.
\end{proof}

\subsection{The main lemma to carry out the induction: proof of Theorem~\ref{thm: Density of Ratios}}\label{sec:induction}
We are now ready to state the following lemma that will imply Theorem~\ref{thm: Density of Ratios}.
In this section we will show how Theorem~\ref{thm: Density of Ratios} follows from this lemma, and the next couple of sections are dedicated to proving Lemma~\ref{lem: Many Cases Lemma}.

\newcommand{\statelemcases}{Let $k,d\in \mathbb{Z}_{\geq 2}$ with $k \leq d$,
	$b \in \big(\frac{d-2}{d}, \frac{d-1}{d+1}\big]\cap \mathbb{Q}$
	and $\lambda \in \SQ\setminus \{\pm 1\}$. 
	Suppose there exists a rooted tree in $\Tc_{d+1}$ with root degree  at most $d-k$ that implements a field $\xi\neq 1$ 
	with the property that $|f_k'(\xi)|\geq 1$ and
	$\xi \in \Arc{[\overline{\lambda_{\floor{k/2}}},\lambda_{\floor{k/2}}]}$. 
	Then the set of fields implemented by  trees in $\Tc_{d+1}$ is dense in $\mathbb{S}$.}
\begin{lemma}\label{lem: Many Cases Lemma}
\statelemcases
\end{lemma}

Using this lemma we can prove Theorem~\ref{thm: Density of Ratios}, which we restate here for convenience.
\begin{thmDensity}
\statethmDensity
\end{thmDensity}
\begin{proof}
We start with the proof of part (b).
Let $\lambda'\in \mathbb{S}\setminus I(\theta_b)$.
By \cite[Corollary 4 and Theorem 5]{PetersRegts2018} it follows that there exists $\lambda\in \mathbb{S}$ arbitrarily close to $\lambda'$ for which there exists a tree $T\in \mathcal{T}_{d+1}$ such that $Z_T(\lambda,b)=0$. 
Choose such a tree $T$ with the minimum number of vertices and let $v$ be a leaf of $T$, from now on referred to as the root of $T$.
Denote $T'=T-v$ and let $u$ be the unique neighbour of $v$ in $T$.
Then $Z_{T,\mi v}(\lambda,b)\neq 0$. 
Indeed, if $Z_{T,\mi v}(\lambda,b)= 0$, then $Z_{T,\pl v}(\lambda,b)= 0$. Since
\[\left (\begin{array}{lr}
\lambda &\lambda b
\\ b&1
\end{array}\right)
 \left (\begin{array}{l} Z_{T',\pl u}(\lambda,b)\\Z_{T',\mi u}(\lambda,b)\end{array}\right)
 =
 \left (\begin{array}{l} Z_{T,\pl v}(\lambda,b)\\Z_{T,\mi v}(\lambda,b)\end{array}\right)
 \]
 and since the matrix is invertible (as $|b|\neq 1$) this would imply $Z_{T',+u}(\lambda,b)=Z_{T',-u}(\lambda,b)=0$ and hence $Z_{T'}(\lambda,b)=0$ contradicting the minimality of $T$.
 Therefore $R_{T,v}=-1$.
 
Since the map $\lambda\mapsto\xi(\lambda) :=\frac{Z_{T,\pl v}(\lambda,b)}{Z_{T,\mi v}(\lambda,b)}$ is holomorphic, it follows that there exists $\lambda''\in \SQ$ arbitrarily close to $\lambda'$ such that $ \xi=\xi(\lambda'')\in \Arc(\lambda_1(b),\overline{\lambda_1(b)})\setminus\{-1\}$.
 Therefore, by Lemma~\ref{lem: elliptic -> irrational rotation} and Theorem~\ref{thm: Mobius classification} the orbit $\{f^n_{\xi,1}(1)\}$ is dense in $\mathbb{S}$. 
 So from Lemma~\ref{lem: Tree Building}, by using paths with the tree $T$ attached to all but one of its vertices at the root $v$ of $T$, we obtain a collection of trees contained in $\mathcal{T}_{d+1}$ whose fields are dense in $\mathbb{S}$.

We next prove part (a) for all $d\geq 2$ and $b \in (\frac{d-2}{d},\frac{d-1}{d+1}]$. The case $b \in (0,\frac{d-2}{d}]$ follows from this by invoking smaller values for $d$. Let $\lambda\in \SQ\setminus \{\pm1\}$.

Let $k_0 = d$ and $m_0 = 0$ and define the sequences $k_n$ and $m_n$ by 
$k_{n+1} = \floor{\frac{k_n}{2}}$ and $m_{n+1} = m_n + k_{n+1}$. Inductively 
we show that $m_n \leq d - k_n$: we have $m_0 = d - k_0$
and then 
\[
	m_{n+1} = m_n + k_{n+1} \leq d- k_n + k_{n+1} = d - \big(k_n - \floor{\tfrac{k_n}{2}}\big)
	\leq d - \floor{\tfrac{k_n}{2}} = d - k_{n+1}.
\]
Clearly, there is an 
integer $N$ such that $k_{N+1} = 1$. We claim that for every $n \in \{0, \dots, N\}$
there is a rooted tree in $\Tc_{d+1}$ with root degree $m_n$ that implements a field $\xi_n$ so that $|f'(\xi_n)|>1$ and at least one of the following holds.
\begin{enumerate}
	\item
	\label{stat: Statement 2}
	There is a tree in $\Tc_{d+1}$ with root degree $m_n$ that implements a field inside $\Arc{(\lambda_{k_{n+1}},\overline{\lambda_{k_{n+1}}})}\setminus\{-1\}$, or else
	\item
	\label{stat: Statement 1}
	The set of fields implemented by trees in $\Tc_{d+1}$  is dense in $\mathbb{S}$.
\end{enumerate}

To show this for $n = 0$ we consider the tree consisting of a single vertex. This tree implements the field $\lambda$ and its root degree is $0$. By equation \eqref{eq: Derivative} of Lemma~\ref{lem: properties f} and since $b\leq \frac{d-1}{d+1}$, we have that  $|f_{d}'(z)| > 1$ for all $z \in \mathbb{S} \setminus \{1\}$, and in particular we have
$|f_{d}'(\lambda)| > 1$. If $\lambda \in \Arc[\overline{\lambda_{k_1}},\lambda_{k_1}]$ then
we apply Lemma~\ref{lem: Many Cases Lemma} to obtain Item~{(\ref{stat: Statement 1})}.
If $\lambda \in \Arc{(\lambda_{k_{1}},\overline{\lambda_{k_{1}}})}$ then the tree consisting of a single
vertex satisfies the conditions of Item~{(\ref{stat: Statement 2})}.

Now suppose that we have shown the claim for $n-1$ for some $n\geq 1$, and assume that we are in the case of Item~{(\ref{stat: Statement 2})} (otherwise we are done), i.e., there is a tree in $\Tc_{d+1}$  with root degree
$m_{n-1}$ that implements a field $\xi$
inside $\Arc{(\lambda_{k_{n}},\overline{\lambda_{k_{n}}})}\setminus\{-1\}$. We can apply Lemma~\ref{lem: tree beyond lambda_k} to obtain a 
 tree $T$ in $\Tc_{d+1}$ with root degree at most
$k_n + m_{n-1} = m_n$ that implements a field $\zeta\neq -1$ such that $|f_{k_{n}}'(\zeta)| > 1$.
If $\zeta \in \Arc[\overline{\lambda_{k_{n+1}}},\lambda_{k_{n+1}}]$ we can apply
Lemma~\ref{lem: Many Cases Lemma} to obtain Item~{(\ref{stat: Statement 1})}. 
We can apply this lemma because $m_n \leq d - k_n$. If $\zeta \in
\Arc{(\lambda_{k_{n+1}},\overline{\lambda_{k_{n+1}}})}$ then $T$ itself satisfies the conditions of Item~{(\ref{stat: Statement 2})}, which proves the claim.

To finish the proof, it remains to consider the case of Item~{(\ref{stat: Statement 2})}, where we can find a tree in $\Tc_{d+1}$ with root
degree at most $m_N < d - 1$ which implements a field
$\xi$ inside $\Arc{(\lambda_{1},\overline{\lambda_{1}})}\setminus\{-1\}$. We have shown in Lemma~\ref{lem: elliptic -> irrational rotation}
that $f_{\xi,1}$ is conjugate to an irrational rotation and thus the orbit $\{f_{\xi,1}^n(1)\}_{n \geq 1}$
is dense in $\mathbb{S}$. The elements of this orbit correspond to rooted trees in $\Tc_{d+1}$, and hence we
can conclude Item~{(\ref{stat: Statement 1})} in this case as well.

 Let $\mathcal{R}$ denote the set of all fields implemented by rooted trees in $\Tc_{d+1}$. 
Let $\zeta \in \mathcal{R}$ and  $T$ be a tree in $\Tc_{d+1}$ that implements $\zeta$.  We construct the tree $\tilde{T}$ with root $r$ obtained by attaching 
$r$ to the root of $T$ with an edge. Then, the root of $\tilde{T}$ has degree 1 and the field implemented by $\tilde{T}$ is $f_{\lambda,1}(\zeta)$. So the set of fields implemented  by  rooted 
trees in  $\Tc_{d+1}$ whose root degrees are $1$ contains
$f_{\lambda,1}(\mathcal{R})$. Since   $f_{\lambda,1}(\mathbb{S}) = \mathbb{S}$ and  $\mathcal{R}$ is dense in $\mathbb{S}$, we conclude
that $f_{\lambda,1}(\mathcal{R})$ is dense in $\mathbb{S}$ as well.
\end{proof}

\begin{remark}
We note that our proof of part (b) rests on the existence of zeros for trees proved in~\cite{PetersRegts2018}, which in turn depends on the chaotic behaviour of the map $f_{d,\lambda}$. 
Alternatively one could also prove part (b) directly from Lemma~\ref{lem: expanding orbit}.
The same proof also yields a dense set of $\lambda\in \SQ$ for which the collection of fields  of trees in $\mathcal{T}_{d+1}$ with root degree $1$ is dense in $\mathbb{S}$ when $b\in (0,\frac{d-1}{d+1}$].
\end{remark}

\section{Contracting maps that cover via degree-derivative interplay}\label{sec:contracting}
In this section we adapt the contracting/covering maps framework of \cite{BGGS} in our setting and show how to apply it using the degree-derivative inteplay alluded to in Section~\ref{sec:ourapproach}. Section~\ref{sec:covering} gives the details of the framework, and Section~\ref{sec:dependence} gives the main lemmas that exploit this interplay.

\subsection{Density on circular arcs via contracting maps that cover}\label{sec:covering}
The contracting maps that cover framework is captured by the following lemma on the interval $[0,1]$, which yields Corollary~\ref{cor: Density in Arc} on circular arcs of the unit circle $\mathbb{S}$.
\begin{lemma}
\label{lem: Density in Arc}
	Let $f_1, \dots, f_k$ be continuously differentiable maps from the interval $[0,1]$ to itself
	such that $0 < f_m'(x) < 1$ for each index $m$ and $x \in (0,1)$ and such that $\bigcup_{m=1}^k {f_m \left( \left[0,1\right] \right)} = [0,1]$.

	Then for any open interval $J \subseteq [0,1]$ there is a sequence of indices $m_1, \dots,
	m_N$ such that
	\[
		\left(f_{m_1} \circ \cdots \circ f_{m_N}\right){\left( \left[0,1\right] \right)} \subset J.
	\]
\end{lemma}

\begin{proof}
For $m \in \{1,\dots, k\}$ define the closed interval $I_m = f_m\left([0,1]\right) = [f_m(0),f_m(1)]$
and note that $f_m: [0,1] \to I_m$ is bijective with a differentiable inverse. We define a sequence of 
intervals in the following way. Let $J_0 = J$ and as long as there exists an index $m$ such that 
$J_n \subseteq I_m$ we define $J_{n+1} = f_m^{-1}(J)$. We will show that this can not be done
indefinitely, i.e., there will be some interval $J_n$ such that $J_n \not \subseteq I_m$ for all $m$.

For an interval $I \subseteq [0,1]$ let $\ell(I)$ denote the length of the interval and denote $\ell(J)$
by $\epsilon$. For each index $m$ choose a partition $I_m = I_{m,L} \cup I_{m,M} \cup I_{m,R}$, where 
$I_{m,L}, I_{m,M}, I_{m,R}$ are of the form $[f_m(0),a), [a,b], (b, f_m(1)]$ respectively for a choice
of $a,b \in \mathrm{int}(I_m)$ such that $a<b$ and both $\ell(I_{m,L})$ and $\ell(I_{m,R})$ are less than 
$\epsilon/4$. We can choose $C > 1$ such that ${f_m^{-1}}'(x) > C$ for all indices $m$ and $x \in I_{m,M}$.
We will show inductively that for all $n \geq 0$ for which $J_n$ is defined it is the case that 
$\ell(J_{n}) \geq \epsilon \cdot (1+C^{n})/2$. For $n = 0$ the statement is true. Suppose that the 
statement is true for $n \geq 0$ for which $J_{n+1}$ is defined. By definition there is an index $m$
such that $J_n \subseteq I_m$ and $J_{n+1} = f_m^{-1}(J_n)$. We find
\begin{align*}
	\ell(J_{n+1}) &= \ell(f_m^{-1}(J_n \cap I_{m,M})) + 
			\ell(f_m^{-1}(J_n \cap I_{m,L})) + \ell(f_m^{-1}(J_n \cap I_{m,R}))\\
			&\geq 
			C \cdot \ell(J_n \cap I_{m,M}) + \ell(J_n \cap I_{m,L}) + \ell(J_n \cap I_{m,R})\\
			&=
			C \left(\ell(J_n) - \ell(J_n \cap I_{m,L}) - \ell(J_n \cap I_{m,R})\right)
			+ \ell(J_n \cap I_{m,L}) + \ell(J_n \cap I_{m,R}),
\end{align*}
where we have used that ${f_m^{-1}}'(x) \geq 1$ for $x \in I_m$. Because
$\ell(J_n \cap I_{m,L}) + \ell(J_n \cap I_{m,R}) \leq \epsilon/2$ this is again at least equal to
\[
	C (\ell(J_n) - \epsilon/2) + \epsilon/2 \geq 
	C (\epsilon \cdot (1+C^{n})/2 - \epsilon/2) + \epsilon/2
	= \epsilon \cdot (1+C^{n+1})/2.
\]
It follows that there is an index $n$ such that $J_n$ is not totally contained inside $I_m$ for any
index $m$. This means that there is an $m$ such that $J_n$ contains at least one of the endpoints
of $I_m$, without loss of generality we can assume that $J_n$ contains the left endpoint of $I_m$.
It follows that there is an $a > 0$ such that $f_m([0,a]) \subset J_n$ and thus there is 
a sequence $m_1, \dots, m_n$ such that $(f_{m_1} \circ \cdots \circ f_{m_n} \circ f_m)([0,a]) \subset J$.
We complete the proof by showing that for at least one of the maps $f_i$ there is an index $N_a$ for 
any $a > 0$ such that $f_i^{N_a}([0,1]) \subset [0,a]$.

Observe that there must be at least one map $f_i$ such that $f_i(0) = 0$. We obtain 
an inclusion of intervals $[0,1] \supset f_i([0,1]) \supset f_i^2([0,1]) \supset \cdots$,
where $f_i^N([0,1]) = [0, f_i^N(1)]$. This shows that the sequence $\{f_i^N(1)\}_{N \geq 0}$ is decreasing
and thus has a limit $L$. If $L \neq 0$ we would have $f_i([0,L]) = [0,L]$, which contradicts the fact 
that $f_i'(x) < 1$ for all $x \in (0,L)$, so $L = 0$. This concludes the proof. 
\end{proof}

\begin{corollary}
\label{cor: Density in Arc}
	Let $A \subset \mathbb{S}$ be a closed circular arc and let $f_1, \dots, f_k$ be 
	orientation preserving continuously differentiable maps from $A$, such that $\bigcup_{m=1}^k {f_m \left(A\right)} = A$ and $0 < |f_m'(x)| < 1$ 
	for each index $m$ and $x \in A$
	not equal to either of the endpoints of $A$.

	Then for any open circular arc $J \subseteq A$ there is a sequence of indices $m_1, \dots,
	m_N$ such that
	\[
		\left(f_{m_1} \circ \cdots \circ f_{m_N}\right){\left( A \right)} \subset J.
\]
\end{corollary}

\subsection{Exploiting the dependence of derivatives on the degrees}\label{sec:dependence}
In this section, we show a few key lemmas that demonstrate how we employ the contracting maps that cover idea, by exploiting the dependence of derivatives on the degrees.
\begin{lemma}
\label{lem: easy odd case}
	Let $k \in \mathbb{Z}_{\geq 1}$ and $b \in \big[\frac{k}{k+2},1\big)$.
	Let $\xi \in \Arc{[\overline{\lambda_{k+1}},\lambda_{k+1}]}$ with $\xi\neq 1$ be such that
	$|f_{2k+1}'(R_k(\xi))| \geq 1$. Then there is an arc $A$ of $\mathbb{S}$ such that the 
	orbit of $1$ under the action of the semigroup generated by $f_{\xi,k+1}$ and $f_{\xi,k}$ 
	is dense in $A$.
\end{lemma}

\begin{proof}
We can assume that $\xi$ lies in the upper half-plane. Since all maps in this argument 
use the parameter $\xi$ we will write $f_{m}$ instead of $f_{\xi,m}$ for all $m$. Define the 
arc $A = \Arc{[R_{k}(\xi), R_{k+1}(\xi)]}$. Using equation~{(\ref{eq: Derivative})} we
find that for every $m$
\[
	|f_m'(R_{k}(\xi))| 	= \frac{m}{2k + 1} \cdot |f_{2k+1}'(R_{k}(\xi))| 
						\geq \frac{m}{2k + 1}.
\]
By using the fact that $R_{k+1}(\xi)$ is either a parabolic or an attracting fixed point
of $f_{k+1}$ we deduce that for all $z \in A$
\[
	|f_{k}'(z)| < |f_{k+1}'(z)| \leq |f_{k+1}'(R_{k+1}(\xi))| \leq 1,
\]
where the second inequality is strict when $z \neq R_{k+1}(\xi)$. It follows that for all
$z \in A$ not equal to $R_{k+1}(\xi)$ we have $k/(2k+1) \leq |f_{k}'(z)| < 1$
and $(k+1)/(2k+1) \leq |f_{k+1}'(z)| < 1$.
Therefore:
\[
	\ell(A) > \ell(f_{k}(A)) > \frac{k}{2k+1}\ell(A)
		\quad\text{ and }\quad
	\ell(A) > \ell(f_{k+1}(A)) > \frac{k+1}{2k+1}\ell(A).
\]
From this we deduce that $\ell(f_{k}(A)) + \ell(f_{k+1}(A)) > \ell(A)$.
Thus, because $f_{k}(A)$ is of the form $\Arc{[R_{k}(\xi), a]}$ and 
$f_{k+1}(A)$ is of the form $\Arc{[b, R_{k+1}(\xi)]}$ for some 
$a,b \in A$, we conclude that $f_{k}(A) \cup f_{k+1}(A) = A$. 

It follows from item~(\ref{enum: lem4}) of Lemma~\ref{lem: properties f} that there is 
some $M$ such that $f_{k+1}^M(1) \in A$. Let $J \subseteq A$ be any open arc.
According to Corollary~\ref{cor: Density in Arc} there is a sequence of indices
$m_1,\dots, m_N\in \{k,k+1\}$ such that $\left(f_{m_1} \circ \cdots \circ f_{m_N} 
\circ f_{k+1}^M\right){\left(1\right)} \in J$. The fact that $J$ was chosen as an arbitrary 
open arc in $A$ concludes the proof.
\end{proof}

\begin{lemma}
\label{lem: easy even case}
	Let $k \in \mathbb{Z}_{\geq 1}$ and $b \in \big[\frac{k-1}{k+1},1\big)$. 
	Let $\xi_1, \xi_2 \in \Arc{[\overline{\lambda_k},\lambda_k]}$
	such that $\xi_1, \xi_2$ are distinct and lie in the same half-plane, i.e., both in 
	the upper or lower half-plane, and such that $|f_{2k}'(R_k(\xi_i))| \geq 1$ for $i \in \{1,2\}$.
	Then there is an arc $A$ of $\mathbb{S}$ such that the orbit of $1$ under the action of the semigroup 
	generated by $f_{\xi_1,k}$ and $f_{\xi_2,k}$ is dense in $A$.
\end{lemma}

\begin{proof}
	We can assume that $\xi_1$ and $\xi_2$ lie in the upper half-plane and that
	$\Arg(\xi_1) < \Arg(\xi_2)$. Let $A = \Arc{[R_k(\xi_1),R_k(\xi_2)]}$, then for all
	$z \in A$ we have 
	\[
	\frac{1}{2} \leq \frac{1}{2}\cdot \left|f_{2k}'(R_k(\xi_1))\right| 
	= \left|f_{k}'(R_k(\xi_1))\right| \leq \left|f_{k}'(z)\right| \leq \left|f_{k}'(R_k(\xi_2))\right| \leq 1.
	\]
	where the second to last inequality is strict when $z \neq R_{k}(\xi_2)$. Therefore 
	for $i \in \{1,2\}$ we have
	\[
		\ell(A) > \ell(f_{\xi_i,k}(A)) > \frac{1}{2} \cdot\ell(A)
	\]
	and from this we deduce that $\ell(f_{\xi_1,k}(A)) + \ell(f_{\xi_2,k}(A)) > \ell(A)$. 
	The rest of the proof proceeds exactly as the proof of Lemma~\ref{lem: easy odd case}.
\end{proof}

\begin{lemma}
\label{lem: Lower Bound three maps}
	Let $k \geq 5$, $b \in \big(\frac{k-1}{k+1},1\big)$ and $\xi \in \Arc{[\overline{\lambda_k},\lambda_k]}$ with $\xi\neq 1$ 
	such that there is an integer $2k \leq p \leq 3k - 5$ for which $|f_p'(\xi)| \geq 1$. Then at least one
	of the following two statements holds: 
	\begin{enumerate}[(i)]
		\item
		The orbit of $1$ under the action of the semigroup generated by $f_{\xi,k-2}, f_{\xi,k-1}$ and $f_{\xi,k}$ is 
		dense in an arc of $\mathbb{S}$.
		\item
		\label{item: statement2}
		We have $|f_k'(R_{k}(\xi))| > 1-\frac{p-k+2}{p} \cdot \frac{p-2k+1}{k}$.
	\end{enumerate}
\end{lemma}

\begin{proof}
W.l.o.g., we may assume that $\xi$ lies in the upper half plane. We write $f_{m} = f_{\xi,m}$ for all indices $m$.
Define the arcs $A_1 = \Arc{[R_{k-2}(\xi),R_{k-1}(\xi)]}$ and $A_2 = \Arc{[R_{k-1}(\xi),R_{k}(\xi)]}$.
Analogously to the proofs of Lemmas \ref{lem: easy odd case} and \ref{lem: easy even case} we can 
use Corollary \ref{cor: Density in Arc} to show that the orbit of $1$ under the action of the semigroup generated by 
$f_{k-2}, f_{k-1}$ and $f_{k}$ is dense in $A_1 \cup A_2$ if 
\begin{equation}
	\label{eq: covering}
	f_{k-2}(A_1 \cup A_2) \cup f_{k-1}(A_1 \cup A_2) \cup f_{k}(A_1 \cup A_2) = A_1 \cup A_2.
\end{equation}
We will assume that this is not the case and show that this leads to statement (\ref{item: statement2}). First 
we will show that the left-hand side of Equation (\ref{eq: covering}) does cover $A_1$. For any 
arc $A$ in the upper half-plane such that $\Arg{(x)} \geq \Arg{(\xi)}$ for all $x \in A$ and index $m$ we have
\begin{equation}
	\label{eq: inequality1}
	\len{f_m(A)} > f_m'(\xi) \cdot \len{A} = \frac{m\cdot f_p'(\xi)}{p} \cdot \len{A} \geq \frac{m}{p} \cdot \len{A}.
\end{equation}
We use this and the fact that $\len{A_2} \geq \len{A_1}$, which follows from item (\ref{enum: lem5})
of Lemma \ref{lem: properties f}, to conclude the following
\begin{align*}
\len{f_{k-2}(A_1 \cup A_2)} + \len{f_{k-1}(A_1)} 
	&\geq \frac{k-2}{p} \len{A_1 \cup A_2} + \frac{k-1}{p} \len{A_1}\geq \frac{2(k-2)}{p}  \len{A_1} + \frac{k-1}{p}  \len{A_1}\\
	& = \frac{3k - 5}{p} \len{A_1} \geq \len{A_1}.
\end{align*}
Because $f_{k-2}(A_1 \cup A_2)$ is of the form $\Arc{[R_{k-2}(\xi),a]}$ and $f_{k-1}(A_1)$ is 
of the form $\Arc{[b,R_{k-1}(\xi)]}$ we have that $A_1$ is covered by $f_{k-2}(A_1 \cup A_2) \cup 
f_{k-1}(A_1)$. Our assumption can be formulated as 
\[
	\len{A_2} \geq \len{f_{k-1}(A_2)} + \len{f_{k}(A_1 \cup A_2)}.
\]
Note that 
\[
	\len{f_{k-1}(A_2)} + \len{f_{k}(A_1 \cup A_2)} \geq \frac{k-1}{p} \len{A_2} + \frac{k}{p} \left(\len{A_1} + \len{A_2}\right).
\]
Combining the previous two inequalities we get
\begin{equation}
\label{eq: inequality2}
\len{A_1} \leq \frac{p-2k+1}{k} \cdot \len{A_2}.
\end{equation}
Let $A_0 = \Arc{[1,R_{k-2}(\xi)]}$. By using the fact that $R_m(\xi)$ is a fixed point of $f_m$ and $f_m(1) = \xi$
for every $m$ we see that $f_{k-2}(A_0) = \Arc{[\xi, R_{k-2}(\xi)]}$, $f_{k-1}(A_0 \cup A_1) = \Arc{[\xi, R_{k-1}(\xi)]}$
and $f_{k}(A_0 \cup A_1 \cup A_2) = \Arc{[\xi, R_{k}(\xi)]}$. It follows from the relation between 
the derivative of different maps given in item (\ref{enum: lem1}) of Lemma \ref{lem: properties f} that for 
any arc $A$ on which $f_{m_1}$ and $f_{m_2}$ are injective we have
\[
	\len{f_{m_1}(A)} = m_1 \cdot \len{f_{1}(A)} = \frac{m_1}{m_2} \cdot \len{f_{m_2}(A)}.
\] 
These observations can be used to write $\len{A_1}$ and $\len{A_2}$ as follows: 
\begin{align*}
	\len{A_1} 	&= \len{f_{k-1}(A_0 \cup A_1)} - \len{f_{k-2}(A_0)}\\
				&= \frac{k-2}{k-1}\len{f_{k-1}(A_0 \cup A_1)} + \frac{1}{k-1}\len{f_{k-1}(A_0 \cup A_1)}
					- \frac{k-2}{k-1}\len{f_{k-1}(A_0)}\\
				&= \frac{k-2}{k-1}\len{f_{k-1}(A_1)} + \frac{1}{k-1}\len{f_{k-1}(A_0 \cup A_1)}
\end{align*}
and 
\begin{align*}
	\len{A_2} = \len{f_{k}(A_0 \cup A_1 \cup A_2)} - \len{f_{k-1}(A_0 \cup A_1)}  = \len{f_{k}(A_2)} + \frac{1}{k-1} \cdot \len{f_{k-1}(A_0 \cup A_1)}.
\end{align*}
By considering our way of writing $\len{A_1}$ and the inequalities given in $(\ref{eq: inequality1})$ 
and (\ref{eq: inequality2}) we obtain the following inequalities
\begin{align*}
	\frac{1}{k-1}\cdot\len{f_{k-1}(A_0 \cup A_1)} 
		&= \len{A_1} - \frac{k-2}{k-1}\cdot\len{f_{k-1}(A_1)}< \len{A_1} - \frac{k-2}{k-1} \cdot \frac{k-1}{p} \len{A_1} \\
		&= \frac{p-k+2}{p} \cdot \len{A_1}< \frac{p-k+2}{p} \cdot  \frac{p-2k+1}{k} \cdot \len{A_2}.
\end{align*}
It follows from the fact that $\Arg{(R_{k}(\xi))} \geq \Arg{(x)}$ for all $x \in A_2$
that $\len{f_{k}(A_2)} < f_k'(R_{k}(\xi)) \cdot \len{A_2}$. By using this inequality and the previous inequality 
we obtain 
\begin{align*}
	\len{A_2} 	&= \len{f_{k}(A_2)} + \frac{1}{k-1} \cdot \len{f_{k-1}(A_0 \cup A_1)} \\
				&< f_k'(R_{k}(\xi)) \cdot \len{A_2} + \frac{p-k+2}{p} \cdot \frac{p-2k+1}{k} \cdot \len{A_2}.
\end{align*}
We can cancel $\len{A_2}$ and rewrite to obtain:
\[
 	f_k'(R_{k}(\xi)) > 1-\frac{p-k+2}{p} \cdot \frac{p-2k+1}{k},
\]
which is what we set out to prove.
\end{proof}

\begin{corollary}
\label{cor: remaining cases}
	Let $m$ be a positive integer, $b \in \big(\frac{m-1}{m+1},1\big)$ and 
	$\xi \in \Arc{[\overline{\lambda_m},\lambda_m]}$ with $\xi\neq 1$ such that either of 
	the following holds:
	\begin{enumerate}[(a)]
		\item
		\label{item: cor1}
		$m \geq 8$ and $|f_{2m}'(\xi)| \geq 1$;
		\item
		\label{item: cor2}
		$m \geq 9$ and $|f_{2m+1}'(\xi)| \geq 1$.
	\end{enumerate}
	Then the orbit of $1$ under the action of the semigroup generated by
	$f_{\xi,m-3},f_{\xi,m-2},f_{\xi,m-1}$ and $f_{\xi,m}$ is dense in an arc of $\mathbb{S}$.
\end{corollary}

\begin{proof}
We will again assume that $\xi$ lies in the upper half-plane. We apply Lemma \ref{lem: Lower Bound three maps}
with $k = m-1$ and $p = 2m$ for item (\ref{item: cor1}) and $p = 2m + 1$ for item (\ref{item: cor2}). If 
the first statement of that lemma holds we see that orbit of $1$ under the action of $f_{m-3},f_{m-2}$ and $f_{m-1}$ 
generates an arc in which case we are done. If we assume that the second statement holds we obtain
\[
	f_{m-1}'(R_{m-1}(\xi)) > 1-\frac{2m-(m-1)+2}{2m} \cdot \frac{2m-2(m-1)+1}{m-1} > \frac{1}{2}
\]
in the case where $p = 2m$ and 
\[
	f_{m-1}'(R_{m-1}(\xi)) > 1-\frac{2m+1-(m-1)+2}{2m+1} \cdot \frac{2m+1-2(m-1)+1}{m-1} > \frac{1}{2}
\]
in the case where $p = 2m+1$. It follows that for $x \in \Arc{[R_{k-1}(\xi),R_{k}(\xi)]}$ we obtain
$1 >|f_{m}'(x)| > |f_{m-1}'(x)| > 1/2$. Therefore, with $A = \Arc{[R_{k-1}(\xi),R_{k}(\xi)]}$, we get 
\[
	\len{f_m\left(A\right)} + 
	\len{f_{m-1}\left(A\right)}
	\geq
	\len{A}.
\]
This, together with Corollary \ref{cor: Density in Arc}, implies
that the orbit of $1$ under the action of the semigroup generated by $f_{m-1}$ and $f_{m}$ is 
dense in $A$.
\end{proof}

\section{Proof of Lemma~\ref{lem: Many Cases Lemma} for Some Special Cases}\label{sec:special}
The arguments of this section will be used to cover some left-over cases in the proof of Lemma~\ref{lem: Many Cases Lemma} that are not directly covered by the results of the previous section.
\subsection{Proof of Lemma~\ref{lem: Many Cases Lemma} for powers of two} The following lemma will be used in the proof of Lemma~\ref{lem: Many Cases Lemma} for those
values of $k$ for which either $k$ or $k+1$ is a power of two, see the proof in Section~\ref{sec: Proof of the main lemma} for details.
\begin{lemma}
\label{lem: density powers of 2}
	Let $d\geq 2,k\geq 0$ be integers,
	$b \in \big(0,\frac{d-1}{d+1}\big]\cap \mathbb{Q}$, $\lambda \in \SQ\setminus\{\pm 1\}$
	and $\xi\in \Lambda_{2^k}(b)\cap \SQ$ with $\xi \neq \pm1$. 	Suppose there is a tree in $\Tc_{d+1}$ with root degree at most $d-(2^{k+1}-1)$ that implements the field~$\xi$.  	Then the set of fields implemented by rooted 
	trees in  $\Tc_{d+1}$ is dense in $\mathbb{S}$.
\end{lemma}

\begin{proof}
We will prove this by induction on $k$. For $k=0$ the field $\xi$ has to lie in
$\Arc(\lambda_1,\overline{\lambda_1})\setminus\{-1\}$ and the root degree of the  
tree in $\Tc_{d+1}$ implementing $\xi$ is at most $d-1$. From Corollary~\ref{cor: mobius parameters} and 
Lemma~\ref{lem: elliptic -> irrational rotation}, we have that $f_{\xi,1}$ is conjugate to an irrational rotation 
and thus the orbit of any initial point $z_0\in\mathbb{S}$ is dense in $\mathbb{S}$. By Lemma~\ref{lem: Tree Building}, every element of 
the set $\{f_{\xi,1}^n(\lambda)\}_{n \geq 1}$ is the field implemented by a tree in $\Tc_{d+1}$, and hence we obtain  the theorem for $k=0$.

Now suppose that $k \geq 1$ and that we have proved the statement for
$k-1$. If $b < (2^{k-1}-1)/(2^{k-1}+1)$, then we must have $k>1$ and we can immediately apply the induction hypothesis with $\xi=\lambda$ and tree consisting of a single vertex.  
So, assume that $b \geq (2^{k-1}-1)/(2^{k-1}+1)$ and observe that the parameter $\lambda_{2^{k-1}}\in \mathbb{S}$ from Lemma~\ref{lem: properties f} exists. It follows from Lemma~\ref{lem: tree beyond lambda_k}
that there is $\sigma\in \mathbb{S}$ with $|f_{2^k}'(\sigma)| > 1$ 
and a set $\mathcal{R}= \{\zeta_n\}_{n \geq 1}$ accumulating on $\sigma$ such 
that each $\zeta_n$ is implemented by a tree in $\Tc_{d+1}$  whose root degree is at most $d - (2^{k+1} - 1) + 2^k = d-(2^{k} - 1)$. If $\mathcal{R}$
has a non-empty intersection with $\Arc{(\lambda_{2^{k-1}},\overline{\lambda_{2^{k-1}}})}\setminus \{-1\}$ 
we can apply the induction hypothesis to the tree corresponding to the field in this intersection. 
Therefore we assume that the elements of $\mathcal{R}$ accumulate on $\sigma$ from 
inside $\Arc{[\overline{\lambda_{2^{k-1}}},\lambda_{2^{k-1}}]}$. It follows that we can find two 
distinct elements $r_1,r_2 \in \mathcal{R}$ such that they both lie in either $\Arc{(\overline{\lambda_{2^{k-1}}},1)}$
or in $\Arc{(1,\lambda_{2^{k-1}})}$ and such that $|f_{2^k}'(r_i)|>1$ for $i =1,2$. 
By Remark~\ref{rem:increase}, we have $|f_{2^k}'(R_{2^{k-1}}(r_i))|>|f_{2^k}'(r_i)|>1$ and  thus 
we can apply Lemma~\ref{lem: easy even case} to conclude that the following set is dense in an arc $A$ of the circle:
\[
	\mathcal{A} = 
	\left\{
		(f_{r_{i_1},2^{k-1}} \circ \cdots \circ f_{r_{i_n},2^{k-1}})(1) : 
		n\in \mathbb{Z}_{\geq 1} \text{ and } i_1, \dots, i_n \in \{1,2\} 
	\right\}.
\]
Since $r_1,r_2$ are implemented by  trees  in $\Tc_{d+1}$ whose root degrees are at most $d - (2^{k+1} - 1) + 2^k = d-(2^{k} - 1)$, by Lemma~\ref{lem: Tree Building}, every element of $\mathcal{A}$ is implemented by a tree in $\Tc_{d+1}$ whose root degree is bounded by
$d-(2^{k} - 1)+ 2^{k-1} = d - (2^{k-1} -1 ) \leq d$. We have seen
that \eqref{eq: Derivative} implies that for $b \leq  (d-1)/(d+1)$ it holds that $|f_d'(z)| > 1$ for all $z \in \mathbb{S}\setminus\{1\}$. 
This implies that there is some $N \in \mathbb{Z}_{\geq 1}$ such that $f_d^N(A) = \mathbb{S}$. 
It follows that the set $\{f_{\lambda,d}^N(a): a \in \mathcal{A}\}$ is dense in $\mathbb{S}$, finishing the proof since every element of this set corresponds to the field of a tree in $\Tc_{d+1}$ (using again Lemma~\ref{lem: Tree Building}).
\end{proof}

\subsection{Proof of Lemma~\ref{lem: Many Cases Lemma} for small cases}\label{sec:cantor}
In this section, we give the main lemma needed to cover certain small cases of Lemma~\ref{lem: Many Cases Lemma}. Interestingly, the proof uses a Cantor-style construction, explained in detail in the next subsection.
\subsubsection{Near-arithmetic progressions}

Let $\alpha \in (0,1)$ and define the maps from the unit interval to itself 
given by $\phi_0(x) = \alpha x$ and $\phi_1(x) = \alpha x + (1-\alpha)$.
Let $\Omega = \cup_{n=0}^\infty \{0,1\}^n$ be the set of finite binary sequences.
For $\omega \in \Omega$ we let $|\omega|$ denote the length of $\omega$ and 
for $\omega_1, \omega_2 \in \Omega$ we let $\omega_1 \oplus \omega_2 \in \Omega$ 
denote the concatenation of the two sequences. For $\omega \in \Omega$
of the form $(\omega^1, \dots,\omega^n)$ and two maps $f_0, f_1$ we let
$f_\omega = f_{\omega^1} \circ \cdots \circ f_{\omega^n}$ and if $|\omega| = 0$
we let $f_\omega$ denote the identity map.
The properties of the semigroup generated by $\phi_0$ and $\phi_1$ for certain 
parameters of $\alpha$ is a topic that has been studied extensively. 
For $\alpha \in (0,\frac{1}{2})$ the set
\[
	\mathcal{C}_\alpha 
	= \bigcap_{n=0}^\infty \bigcup_{\substack{\omega \in \Omega\\|\omega| = n}} \phi_{\omega}([0,1])
\]
is a Cantor set, with $\mathcal{C}_{1/3}$ being the Cantor ternary set.
We will not use the properties of Cantor sets, so we do not define them. 
First we state some easy to prove properties 
of this semigroup to describe a construction that will help us to prove 
Lemma~\ref{lem: Many Cases Lemma} for small cases of $k$.

\begin{lemma}
	Let $\omega \in \Omega$ and $\alpha \in (0,1)$. Then $\phi_\omega([0,1])$ 
	is an interval of length $\alpha^{|\omega|}$, furthermore
	the intervals $\phi_{\omega\oplus (0)}([0,1])$ and $\phi_{\omega\oplus (1)}([0,1])$ are
	subintervals of $\phi_\omega([0,1])$ sharing the left and right boundary respectively. 
\end{lemma}

\begin{proof}
Because the derivative of $\phi_i$ is constantly equal to $\alpha$ for $i = 1,2$ 
it follows that the length of $\phi_\omega([0,1])$ is $\alpha^{|\omega|}$. 
The maps $\phi_i$ are increasing and thus we can write
$\phi_\omega([0,1]) = [\phi_\omega(0),\phi_\omega(1)]$ and also $\phi_{\omega\oplus (0)}([0,1])=
[\phi_{\omega \oplus (0)}(0),\phi_{\omega \oplus (0)}(1)]=[\phi_{\omega}(0),\phi_{\omega \oplus (0)}(1)]$. 
Therefore the left boundaries of $\phi_\omega([0,1])$ and $\phi_{\omega \oplus (0)}([0,1])$ 
are equal. The length of the latter interval is $\alpha^{|\omega| + 1}$, which is less than the 
length of $\phi_\omega([0,1])$ and thus $\phi_{\omega \oplus (0)}([0,1])$
is indeed contained in $\phi_\omega([0,1])$. The stated property of
$\phi_{\omega \oplus (1)}([0,1])$ follows completely analogously.
\end{proof}

For two sets $A, B \subseteq \mathbb{R}$ we will let $A + B = \{a+b: a \in A, b \in B\}$.
A famous property of the Cantor ternary set is that $\mathcal{C}_{1/3} + \mathcal{C}_{1/3} = [0,2]$.
More generally one can show that
$\mathcal{C}_\alpha + \mathcal{C}_\alpha = [0,2]$ for all $\alpha \in [\frac{1}{3},1)$.
In \cite{MendesEtAl1994} an overview is given of the possible structures of
$\mathcal{C}_{\alpha_1} + \mathcal{C}_{\alpha_2}$
for pairs of $\alpha_1,\alpha_2 \in (0,1)$. Similar methods to those used 
in \cite{MendesEtAl1994} can be used to show the following.

\begin{lemma}
\label{lem: Near-AP linear}
	Let $\alpha \in [\frac{1}{3},1)$ and $\epsilon > 0$. Then there are 
	sequences $\omega_1, \omega_2, \omega_3 \in \Omega$
	such that for all triples $p_1, p_2, p_3$ with $p_i \in \phi_{\omega_i}([0,1])$
	\begin{equation}
	\label{eq: Near-AP}
		\left| \frac{p_2-p_1}{p_3 - p_2} - 1\right| < \epsilon.
	\end{equation}
\end{lemma}

\begin{proof}
First assume that $\alpha \in [\frac{1}{2},1)$. Then $\phi_0([0,1])\cup \phi_1([0,1]) = [0,1]$. 
It follows from Lemma~\ref{lem: Density in Arc} that for any $\delta > 0$ there are elements
$\omega_1,\omega_2$ and $\omega_3$
in $\Omega$ such that 
\[
	\phi_{\omega_1}\left([0,1]\right) \subseteq [0, \delta],\quad
	\phi_{\omega_2}\left([0,1]\right) \subseteq [1/2-\delta, 1/2 + \delta]\quad{\text{ and }}\quad
	\phi_{\omega_3}\left([0,1]\right) \subseteq [1-\delta, 1].
\]
By choosing $\delta$ small enough we can guarantee the inequality in (\ref{eq: Near-AP}). 

Assume now that $\alpha \in [\frac{1}{3},\frac{1}{2})$. We will first show that if there 
are $\omega_i \in \Omega$ with $|\omega_i| = n$ and $q_i \in \phi_{\omega_i}([0,1])$ 
for $i = 1,2,3$ such that $q_1 + q_3 = 2q_2$, then there are choices of indices $k_i \in \{0,1\}$ such 
that there exist $\tilde{q_i} \in \phi_{\omega_i \oplus (k_i)}([0,1])$ for which
$\tilde{q_1} + \tilde{q_3} = 2\tilde{q_2}$. Suppose that we are given such $\omega_i$ and $q_i$. 
Let $I_i = \phi_{\omega_i}([0,1])$ and $I_i^{k} = \phi_{\omega_i\oplus(k)}([0,1])$ 
for $i = 1,2,3$ and $k=0,1$.
We will show that
\begin{equation}
	\label{eq: sum of intervals}
	I_1 + I_3 = (I_1^0 + I_3^0) \cup (I_1^1 + I_3^0) \cup (I_1^0 + I_3^1).
\end{equation}
Let $a_1$ and $a_3$ be the left boundary of $I_1$ and $I_3$ respectively. 
Because $|\omega_1| = |\omega_3| = n$ it follows that
$I_1 = [a_1, a_1 + \alpha^n]$ and $I_3 = [a_3, a_3 + \alpha^n]$ and 
thus $I_1 + I_3 = [a_1 + a_3, a_1 + a_3 + 2\alpha^n]$,
which we can denote as $a_1 + a_3 + \alpha^n \cdot [0,2]$.
Now 
\begin{alignat*}{2}
	I_1^0 + I_3^0 	&= 
				\left(a_1 + \alpha^{n}\cdot [0,\alpha]\right) + \left(a_3 + \alpha^{n}\cdot [0,\alpha]\right) 
			&&= a_1 + a_3 + \alpha^n \cdot [0,2\alpha]\\
	I_1^1 + I_3^0 	&= 
				\left(a_1 + \alpha^{n}\cdot [1-\alpha,1]\right) + \left(a_3 + \alpha^{n}\cdot [0,\alpha]\right) 
			&&= a_1 + a_3 + \alpha^n \cdot [1-\alpha,1+\alpha]\\
	I_1^1 + I_3^1 	&= 
				\left(a_1 + \alpha^{n}\cdot [1-\alpha,1]\right) + \left(a_3 + \alpha^{n}\cdot [1-\alpha,1]\right) 
			&&= a_1 + a_3 + \alpha^n \cdot [2-2\alpha,2].
\end{alignat*}
Because $\alpha \in [\frac{1}{3},1)$ it follows that 
\[
[0,2] = [0,2\alpha] \cup [1-\alpha,1+\alpha] \cup [2-2\alpha,2],
\]
thus showing (\ref{eq: sum of intervals}). Because there are $q_i \in I_i$ 
such that $q_1 + q_3 = 2 q_2$ we know that $I_1 + I_3$ is not disjoint from $2I_2$. 
These two intervals have the same length and thus at least one of the boundary 
points of $2I_2$ lies in $I_1 + I_3$ therefore there is a $k_2 \in \{0,1\}$ 
such that $2I_2^{k_2}$ is not disjoint from $I_1 + I_3$ because
the intervals $I_2^{0}$ and $I_2^{1}$ contain the respective 
boundary points of $I_2$. This means that
$2I_2^{k_2}$ is not disjoint from $(I_1^0 + I_3^0) \cup (I_1^1 + I_3^0) \cup (I_1^0 + I_3^1)$ 
and thus there are also choices of $k_1, k_3 \in \{0,1\}$ such that $I_1^{k_1} + I_3^{k_3}$ 
is not disjoint from $2I_2^{k_2}$. It follows that there are 
$\tilde{q_i} \in I_i^{k_i}$ such that $\tilde{q_1} + \tilde{q_3} = 2\tilde{q_2}$.

Let $\omega_1 = (0,0)$, $\omega_2 = (0,1)$ and $\omega_3 = (1,0)$.
Note that $0 \in \phi_{\omega_1}([0,1])$ and
$1-\alpha \in \phi_{\omega_3}([0,1]) = [1-\alpha, 1-\alpha + \alpha^2]$. 
Furthermore it can be checked, using the fact that $\alpha \in [\frac{1}{3},\frac{1}{2})$, 
that $(1-\alpha)/2  \in \phi_{\omega_2}([0,1]) = [\alpha-\alpha^2,\alpha]$ and thus there are 
$q_i \in \phi_{\omega_i}([0,1])$ such that $q_1 + q_3 - 2q_2 = 0$. 
From the previous considerations it follows that there are $\tilde{\omega}_i \in \Omega$ 
of arbitrary length such that there are $\tilde{q}_i \in \phi_{\omega_i \oplus \tilde{\omega}_i}([0,1])$
for which $\tilde{q}_1 + \tilde{q}_3 - 2\tilde{q}_2 = 0$. 
See Figure~\ref{fig: cantor} for an illustration of the construction described in this proof.
By taking the length of $\tilde{\omega}_i$ large enough, the lengths
of the intervals can be made arbitrarily small and thus we can guarantee that 
\[
	\left|p_3 - p_2 \right| \cdot \left| \frac{p_2-p_1}{p_3 - p_2} - 1\right| 
		= \left|p_1 + p_3 - 2p_2 \right| < \epsilon \cdot (1 - 2\alpha)
\]
for all triples $p_i \in \phi_{\omega_i \oplus \tilde{\omega}_i}([0,1])$. 
Because $\phi_{\omega_i \oplus \tilde{\omega}_i}([0,1]) \subseteq 
\phi_{\omega_i}([0,1])$, we conclude that $p_3 - p_2$ is at least $1 - 2\alpha$. 
The inequality in (\ref{eq: Near-AP}) follows.
\end{proof}

\begin{figure}
  \centering
    \includegraphics[width=\textwidth]{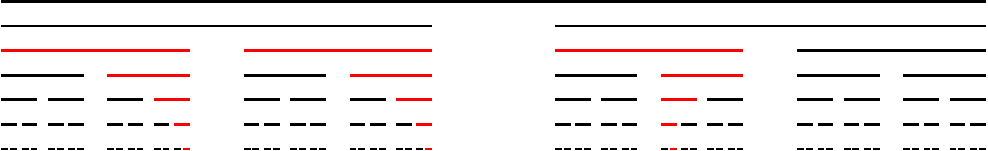}
    \caption{An illustration of the union of $\phi_\omega([0,1])$, where $\omega\in\Omega$ runs over all
    sequences of length $n$ for $n = 0,1, \dots, 6$ for $\alpha = 7/16$. At each level, starting at level two, 
    three red intervals are highlighted containing 
    elements $q_1, q_2$ and $q_3$ respectively such that $q_1 + q_3 = 2q_2$.}
    \label{fig: cantor}
\end{figure}

\begin{lemma}
	Let $\alpha \in [\frac{1}{3},1)$, $\epsilon > 0$ and $f_0,f_1$ 
	differentiable maps from $[0,1]$ to itself with fixed points $0$ and $1$ respectively.
	Then there is a constant $\delta > 0$ such that if $|f_i'(x) - \alpha| < \delta$ 
	for $i = 0,1$ and all $x \in [0,1]$ then there 
	are $\omega_1, \omega_2, \omega_3 \in \Omega$ such that for all 
	triples $p_1, p_2, p_3$ with $p_i \in f_{\omega_i}([0,1])$ it holds that
\begin{equation}
	\label{eq: Near-AP2}
	\left| \frac{p_2 - p_1}{p_3 - p_2} - 1\right| < \epsilon.
\end{equation}
\end{lemma}

\begin{proof}
Suppose that $|f_i'(x) - \alpha| < \delta$ for $i = 0,1$ and all $x \in [0,1]$. 
For any $x \in [0,1]$ we can write 
\[
	f_0(x) = \int_0^x f_0'(t) dt \quad\text{ and }\quad 
	f_1(x) = 1 - \int_x^1 f_1'(t) dt. 
\]
We show inductively, that for all $x \in [0,1]$ and $\omega \in \Omega$ we have 
$|f_\omega(x) - \phi_\omega(x)| \leq |\omega| \cdot \delta$. When $|\omega| = 0$ the 
statement is clear, so we suppose that $|\omega| > 0$. Assume that the first entry
of $\omega$ is a $0$ so we write $\omega = (0)\ \oplus\ \omega'$ for some
$\omega' \in \Omega$ with $|\omega| = |\omega'| + 1$. Let $x \in [0,1]$, we assume that 
we have shown that $|f_{\omega'}(x) - \phi_{\omega'}(x)| < \delta \cdot |\omega'|$. We
denote $f_{\omega'}(x)$ by $y$ and $\phi_{\omega'}(x)$ by $y + r$, where 
$|r| \leq \delta \cdot |\omega'|$. Now
\begin{align*}
	\left|f_\omega(x) - \phi_\omega(x)\right| 	
		&= \left|f_0(y) - \phi_0(y + r)\right|= \left|\int_0^y f_0'(t) dt - \alpha \cdot (y+r) \right| = \left|\int_0^y \left(f_0'(t) - \alpha\right) dt - \alpha r \right|\\ &\leq \int_0^y |f_0'(t)-\alpha| dt + \alpha |r| \leq y \delta + \alpha \delta |\omega'|< \delta (|\omega'| + 1) = \delta |\omega|.
\end{align*}
If the first entry of $\omega$ is a $1$ the calculation is analogous. 

Let $\omega_1, \omega_2, \omega_3 \in \Omega$ such that for all 
triples $p_1, p_2, p_3$ with $p_i \in \phi_{\omega_i}([0,1])$ it holds that $\big| \tfrac{p_2-p_1}{p_3 - p_2} - 1\big| < \epsilon/2$. 
These choices of $\omega_i$ exist by Lemma~\ref{lem: Near-AP linear}. 
For this inequality to hold it must be the case that
$\phi_{\omega_2}([0,1]) \cap \phi_{\omega_3}([0,1]) = \emptyset$ and thus, 
since the map $(p_1,p_2,p_3) \to (p_2-p_1)/(p_3-p_2)$ is continuous in all points 
where $p_2 \neq p_3$, we can find three open intervals $I_1,I_2,I_3$ 
with $\phi_{\omega_i}([0,1]) \subseteq I_i$ such that for all triples $q_i \in I_i$ 
we have
\[
	\left| \frac{q_2-q_1}{q_3 - q_2} - 1\right| < \epsilon.
\]
We showed that by making $\delta$ small enough we obtain bounds on the difference 
between $f_\omega(x)$ and $\phi_\omega(x)$ uniformly over all 
$x \in [0,1]$ and $\omega$ of bounded length. Therefore we can make $\delta$ 
sufficiently small such that $f_{\omega_i}([0,1]) \subset I_i$ for $i \in \{1,2,3\}$, 
which is enough to conclude the statement of the lemma.
\end{proof}

\begin{corollary}
\label{cor: Cantor on Circle}
	Let $\alpha \in [\frac{1}{3},1)$, $\epsilon > 0$. There is $\delta>0$ 
	such that the following holds for any closed circular arc $A \subseteq \mathbb{S}^1$ and two
	maps $f_0, f_1:A\rightarrow A$ with the respective endpoints of $A$ as 
	fixed points with the property that $||f_i'(z)|-\alpha| < \delta$ for all $z\in A$.
	
	There exist
	$\omega_1, \omega_2, \omega_3 \in \Omega$ such that for all triples
	$p_i$ with $p_i \in f_{\omega_i}(A)$ we have that $\Arc[p_1,p_2]$ and $\Arc[p_2,p_3]$ 
	are subsets of $A$ satisfying
\[
	\left|\frac{\ell(\Arc[p_1,p_2])}{\ell(\Arc[p_2,p_3])} - 1\right| < \epsilon.
\]
\end{corollary}

We are now ready to prove the following lemma.
\begin{lemma}
\label{lem: cantor implementation}
	Let $d \in \mathbb{Z}_{\geq 2}$, $k \in \mathbb{Z}_{\geq 1}$,
	$b \in \big[\frac{k-1}{k+1},\frac{d-1}{d+1}\big]$ with $b \neq 0$ and $\lambda \in \mathbb{S}$. 
	Let $\xi \in \mathbb{S}\setminus\{-1\}$
	with $|f'_{3k}(\xi)| > 1$. Let $\{\xi_n\}_{n \geq 1}$ be a sequence in $\mathbb{S}$ 
	 converging to $\xi$ and not equal to $\xi$ such that for all positive integers $n$ there is a rooted 
	tree $T_n$ in $\Tc_{d+1}$, with root degree $m \leq d-2k$ implementing the field $\xi_n$. 
	Then at least one of the following is true.
	\begin{enumerate}
		\item
		\label{item: Case1 lemma}
		The set of fields implemented by rooted trees in $\Tc_{d+1}$ is dense in $\mathbb{S}$.
		\item
		\label{item: Case2 lemma}
		Given $\epsilon > 0$, there is a rooted tree in $\Tc_{d+1}$ with root degree at most $m+k$ that implements the field  
		$r \in \Arc{(\lambda_k,\overline{\lambda_k})}\setminus\{-1\}$ with
		$|f_{3k}'(r)| > |f'_{3k}(\xi)| - \epsilon$.
\end{enumerate}
\end{lemma}

\begin{proof}
We distinguish the following three cases.
\begin{enumerate}[(i)]
\item
\label{item: Case1}
$\xi \in \Arc(\lambda_k,\overline{\lambda_k})$.
\item
\label{item: Case2}
$\xi \in \Arc[\overline{\lambda_k},\lambda_k]$ and $R_k(\xi) \in \Arc(\lambda_k,\overline{\lambda_k})$.
\item
\label{item: Case3}
 $\xi \in \Arc[\overline{\lambda_k},\lambda_k]$ and $R_k(\xi) \in \Arc[\overline{\lambda_k},\lambda_k]$.
\end{enumerate}
Suppose first we are in case~{(\ref{item: Case1})}. Then, since $\xi_n \to \xi$ and thus $f_{3k}'(\xi_n) \to f_{3k}'(\xi)$, given $\epsilon$,
there is an integer $n$ such that $\xi_n \in \Arc(\lambda_k,\overline{\lambda_k})\setminus\{-1\}$ and $|f_{3k}'(\xi_n)| > |f'_{3k}(\xi)| - \epsilon$.
The rooted tree $T_n$ satisfies the requirements of statement~{(\ref{item: Case2 lemma})} of the lemma.

To prove the lemma for cases (\ref{item: Case2}) and (\ref{item: Case3}) we define the following set
\[
	\mathcal{R} = \left\{f_{\xi_n,k}^N(\xi_n): n,N \geq 1 \right\}.
\]
By repeatedly applying Lemma~\ref{lem: Tree Building} we see that every element of $\mathcal{R}$ corresponds to 
the field implemented by a rooted tree in $\Tc_{d+1}$ whose root degree is at most $m + k$. The following limits follow from continuity
\[
	\lim_{N \to \infty} \lim_{n \to \infty} f_{\xi_n,k}^N(\xi_n) = \lim_{N \to \infty} f_{\xi,k}^N(\xi) = R_k(\xi).
\]
Therefore $R_k(\xi)$ is an accumulation of $\mathcal{R}$ and in fact there is a sequence $\{\zeta_n\}_{n \geq 1}$ of elements 
in $\mathcal{R}$ converging to $R_k(\xi)$ but not equal to $R_k(\xi)$. If we are in case~{(\ref{item: Case2})}, by Remark~\ref{rem:increase} we can take
$\zeta_n$ sufficiently close to $R_k(\xi)$ so that $\zeta_n \in \Arc{(\lambda_k,\overline{\lambda_k})}\setminus\{-1\}$. Since $|f_{3k}'(R_k(\xi))| > |f_{3k}'(\xi)|$ by Remark~\ref{rem:increase},  we can further ensure that $|f_{3k}'(\zeta_n)| > |f_{3k}'(\xi)|$. The corresponding tree with
field $\zeta_n$ satisfies the condition of statement~{(\ref{item: Case2 lemma})} of the lemma.

Suppose now we are in case~{(\ref{item: Case3})} and suppose first that $R_k(\xi) \in \{\overline{\lambda_k},\lambda_k\}$.
If a subsequence $(\zeta_n)$ converges to $R_k(\xi)$ along the arc $\Arc(\lambda_k,\overline{\lambda_k})$, 
we obtain a $\zeta_n \in \Arc(\lambda_k,\overline{\lambda_k})$ and by the same reasoning as in the previous case we can 
conclude that statement~{(\ref{item: Case2 lemma})} of the lemma holds. So we can assume that for large enough $n$
all $\zeta_n$ lie in $\Arc(\overline{\lambda_k},\lambda_k)$. In this case we find that for sufficiently high $n$
the elements $\zeta_n$ get arbitrarily close to either $\lambda_k$ or $\overline{\lambda_k}$ and thus $|f_{k}'(R_k(\zeta_n))|$
gets arbitrarily close to $1$. It follows that we can find $n_1$ and $n_2$ such that $\zeta_{n_1}$ and $\zeta_{n_2}$ lie
in the same half plane and such that $|f_{2k}'(R_k(\zeta_{n_i}))| > 1$ for $i = 1,2$. It follows then from Lemma~{\ref{lem: easy even case}}
that, if we let $g_0 = f_{\zeta_{n_1},k}$ and $g_1 = f_{\zeta_{n_2},k}$, the set
\[
\mathcal{R}_1 = \{g_{\omega}(1): \omega \in \Omega, |\omega| \geq 1\}
\] 
is dense in an arc $A \subseteq S$. By applying Lemma~\ref{lem: Tree Building} we observe that every $r \in \mathcal{R}_1$
corresponds to the field implemented by  a rooted tree in $\Tc_{d+1}$ with root degree at most $m + 2k \leq d$. Because 
the tree consisting of a single vertex implements the field $\lambda$ we can apply Lemma~\ref{lem: Tree Building} to see that
every element in the set 
\[
\mathcal{R}_2 = \{f_{\lambda,d}^n(r): r \in \mathcal{R}_1,n\geq 1\}
\]
corresponds to the field implemented by a rooted tree in $\Tc_{d+1}$ with root degree at most $d$.
Because $b$ is chosen such that $|f_d'(z)| > 1$ for 
all $z \in \mathbb{S}-\{1\}$ we find that $f_{\lambda,d}^N(A) = \mathbb{S}$ for a sufficiently large $N$ and thus $\mathcal{R}_2$ 
is dense in $\mathbb{S}$, which shows that in this case statement~{(\ref{item: Case1 lemma})} of the lemma holds.

Finally we assume that $R_{k}(\xi) \in \Arc(\overline{\lambda_k},\lambda_k)$. W.l.o.g. assume that $\xi$ lies
in the upper half-plane. Let $\alpha = |f_k'(R_k(\xi))|$.
It follows from the fact that $R_k(\xi) \in \Arc(\xi, \lambda_k)$ that $\alpha \in (1/3,1)$. Let $\epsilon_1,\epsilon_2>0$ be two reals whose value will be determined later. Let $\delta$ be the constant obtained from 
applying Corollary~\ref{cor: Cantor on Circle} to $\alpha$ and $\epsilon=\epsilon_1$. Now choose $n_1,n_2$ such that
$\xi_{n_1},\xi_{n_2}$ have the following properties.
\begin{enumerate}[(a)]
	\item
	$\xi_{n_1}$ and $\xi_{n_2}$ lie in the upper half-plane, $\Arg(\xi_{n_1}) < \Arg(\xi_{n_2})$, $\Arc{[R_{k}(\xi_{n_1}), R_{k}(\xi_{n_2})]} \subseteq \Arc{(1,\lambda_k)}$ and 
	$\Arg(R_k(R_k(\xi_{n_1}))) > \Arg(R_k(\xi))$.
	\item
	For all $z \in \Arc{[R_{k}(\xi_{n_1}), R_{k}(\xi_{n_2})]}$ we have $||f_k'(z)|-\alpha|<\delta$.
	\item
	\label{item: derivative bound}
	For all $z_1,z_2\in\Arc{[R_{k}(\xi_{n_1}), R_{k}(\xi_{n_2})]}$ we have $||R_k'(z_1)/R_k'(z_2)|-1|<\epsilon_2$.
\end{enumerate}
That it is possible to choose $n_1, n_2$ such that the first two properties hold follows from the fact that 
both $R_k$ and the derivative of $f_k$ are continuous on $\Arc{[\overline{\lambda_k},\lambda_k]}$. The existence
of $n_1,n_2$ satisfying the third property follows from the fact that the derivative of $z \mapsto R_k(z)$
is continuous and non-zero on $\Arc{(1,\lambda_k)}$. 

Let $g_0 = f_{\xi_{n_1},k}$ and $g_1 = f_{\xi_{n_2},k}$. Since $\xi_{n_1},\xi_{n_2}$ are implemented by rooted trees in $\Tc_{d+1}$ with root degrees at most $m$, we have  by Lemma~\ref{lem: Tree Building} that,  if $r$ is implemented by a rooted tree
in $\Tc_{d+1}$, then $g_i(r)$ is the field implemented by a tree in $\Tc_{d+1}$ and root degree $m+k\leq d$. 
Let $A = \Arc{[R_{k}(\xi_{n_1}), R_{k}(\xi_{n_2})]}$ and note that the maps $g_0,g_1$ have the respective endpoints of 
$A$ as fixed points. Furthermore $||g_i'(z)|-\alpha| < \delta$ for all $z \in A$ and thus it follows from
Corollary~\ref{cor: Cantor on Circle} that there is a triple $\omega_1,\omega_2,\omega_3 \in \Omega$ such that 
for all triples $p_i \in g_{\omega_i}(A)$ we have $\Arg(p_1) < \Arg(p_2) < \Arg(p_3)$ and 
\[
	\left|\frac{\ell(\Arc[p_1,p_2])}{\ell(\Arc[p_2,p_3])} - 1\right| < \epsilon_1.
\]
The orbit of $\xi_{n_2}$ under iteration of $g_1$ converges to $R_k(\xi_{n_2})$ approaching from an 
anti-clockwise direction and thus there is some number $N$ such that if we let $\omega_N$ be the 
constant $1$ sequence of length $N$ that $g_{\omega_N}(\xi_{n_2}) \in A$. For $i = 1,2,3$ we define 
$\zeta_i = g_{\omega_i \oplus \omega_N}(\xi_{n_2})$ and note that each $\zeta_i$ is contained in the interval $(\overline{\lambda_k},\lambda_k)$ and is implemented by a rooted tree in $\Tc_{d+1}$ with root degree $m+k$. Furthermore we have
$\Arg(\zeta_1) < \Arg(\zeta_2) < \Arg(\zeta_3)$ and
\begin{equation}
	\label{eq: zeta arcs}
	\left|\frac{\ell(\Arc[\zeta_1,\zeta_2])}{\ell(\Arc[\zeta_2,\zeta_3])} - 1\right| < \epsilon_1.
\end{equation}
Let $h_i = f_{\zeta_i,k}$. Analogously to above, if $r$ is implemented by a rooted tree in $\Tc_{d+1}$, then $h_i(r)$ is implemented by a rooted tree in $\Tc_{d+1}$ with root degree at most $m+2k\leq d$. Redefine $A = \Arc{[R_k(\zeta_1),R_k(\zeta_3)]}$. We will show that we 
can choose $\epsilon_1$ and $\epsilon_2$ sufficiently small such that $A = h_1(A) \cup h_2(A) \cup h_3(A)$.
To do this define $A_1 = \Arc{[R_k(\zeta_1),R_k(\zeta_2)]}$ and $A_2 = \Arc{[R_k(\zeta_2),R_k(\zeta_3)]}$.
It follows from the mean value theorem that there are $x_i \in \Arc{[\zeta_{i},\zeta_{i+1}]}$ such that
$\ell(A_i) = |R_k'(x_i)| \cdot \ell(\Arc{[\zeta_{i},\zeta_{i+1}]})$ for $i = 1,2$. Because both $x_1$ and 
$x_2$ lie in $\Arc{[R_{k}(\xi_{n_1}), R_{k}(\xi_{n_2})]}$ it follows from 
property (\ref{item: derivative bound}) above that we can write 
$|R_k'(x_1)/R_k'(x_2)| = 1 + r_2$ for some $r_2 \in \mathbb{R}$ with $|r_2| < \epsilon_2$.
We use the bound in (\ref{eq: zeta arcs}) to obtain the following inequality
\begin{alignat*}{2}
\left|\frac{\ell(A_1)}{\ell(A_2)}-1\right| 
&= \left|\frac{|R_k'(x_1)| \cdot \ell(\Arc{[\zeta_{1},\zeta_{2}]})}{|R_k'(x_2)| \cdot \ell(\Arc{[\zeta_{2},\zeta_{3}]})}-1\right|
= \left|(1+r_2)\frac{\ell(\Arc{[\zeta_{1},\zeta_{2}]})}{\ell(\Arc{[\zeta_{2},\zeta_{3}]})}-1\right| \\
&\leq \left|1+r_2\right| \cdot\left|\frac{\ell(\Arc{[\zeta_{1},\zeta_{2}]})}{\ell(\Arc{[\zeta_{2},\zeta_{3}]})}-1\right| + |r_2| 
< \left|1+r_2\right| \cdot \epsilon_1 + |r_2|\\ 
&\leq \epsilon_1 + \epsilon_2 + \epsilon_1 \cdot \epsilon_2.
\end{alignat*}
Let $\epsilon_3 = \epsilon_1 + \epsilon_2 + \epsilon_1 \cdot \epsilon_2$ and note that $\epsilon_3$ can be made 
arbitrarily small by choosing $\epsilon_1$ and $\epsilon_2$ sufficiently small. It follows that there is some 
$r_3 \in \mathbb{R}$ with $|r_3| < \epsilon_3$ such that $\ell(A_1) = (1+r_3)\cdot\ell(A_2)$. Because 
$\Arg(R_k(\zeta_1)) > \Arg(R_k(R_k(\xi_{n_1}))) > \Arg(R_k(\xi))$ we find that $1>|f_k'(z)| > \alpha$
for all $z \in A$ and thus $1>|h_i'(z)| > \alpha$ for all $z \in A$ and $i = 1,2,3$. It follows that 
\begin{align*}
\ell(h_1(A_1 \cup A_2)) + \ell(h_2(A_1)) &> \alpha\cdot(\ell(A_1) + \ell(A_2)) + \alpha \cdot \ell(A_1) = \alpha\cdot(2\ell(A_1) + \ell(A_2)) \\
& = \alpha\cdot\Big(2 + \frac{1}{1+r_3}\Big) \ell(A_1) = \alpha \cdot \frac{3+2r_3}{1+r_3} \cdot \ell(A_1).
\end{align*}
and 
\begin{align*}
\ell(h_2(A_2)) + \ell(h_3(A_1 \cup A_2)) &> \alpha \cdot \ell(A_2) + \alpha\cdot(\ell(A_1) + \ell(A_2))  = \alpha\cdot(\ell(A_1) + 2\ell(A_2)) \\
& = \alpha\cdot\left((1+r_3) + 2\right) \ell(A_2) = \alpha \cdot (3+r_3) \cdot \ell(A_2).
\end{align*}
Because $\alpha > 1/3$ we can choose $\epsilon_3$ sufficiently small such that 
\[
	\ell(h_1(A_1 \cup A_2)) + \ell(h_2(A_1)) > \ell(A_1) 
	\quad\text{ and }\quad
	\ell(h_2(A_2)) + \ell(h_3(A_1 \cup A_2)) > \ell(A_2).
\]
Because $h_1(A_1 \cup A_2)$ and $h_2(A_1)$ share the respective endpoints of
$A_1$ it follows that $A_1 \subseteq h_1(A_1 \cup A_2) \cup h_2(A_1)$. Similarly we 
find that $A_2 \subseteq h_3(A_1 \cup A_2) \cup h_2(A_2)$.
It follows that $A = h_1(A) \cup h_2(A) \cup h_3(A)$. Finally let $s = h_3^N(1)$, where 
we have taken $N$ sufficiently large such that $s \in A$, and consider
\[
	\mathcal{S} = \left\{(h_{i_1} \circ \cdots \circ h_{i_l})(s): l \in \mathbb{Z}_{\geq 1} \text{ and } i_1,\dots,i_l \in \{1,2,3\}\right\}.
\]
It follows from Corollary~\ref{cor: Density in Arc} that $\mathcal{S}$ is a dense subset of $A$. Every 
$r \in \mathcal{S}$ is implemented by a rooted tree in $\Tc_{d+1}$ with root degree $m+2k \leq d$. 
Finally we let 
\[
\mathcal{S}_2 = \{f_{\lambda,d}^n(r): r \in \mathcal{S},n\geq 1\}
\]
and we find, because $|f_d'(z)| > 1$ for all $z \in \mathbb{S}-\{1\}$, that $\mathcal{S}_2$ is dense in $\mathbb{S}$.
Every $r \in \mathcal{S}_2$ is implemented by a tree in $\Tc_{d+1}$.
This shows that in this case item~{(\ref{item: Case1 lemma})} of the lemma holds.
\end{proof}

\begin{lemma}
\label{lem: d-5 after lambda_3}
	Suppose $d \in \mathbb{Z}_{\geq 5}$,
	$b \in \big(0, \frac{d-1}{d+1}\big] \cap \mathbb{Q}$, $\lambda \in \SQ\setminus \{\pm 1\}$
	and $\xi\in \Lambda_3(b)\cap \SQ$ with $\xi\neq \pm1$.

	Suppose there is a rooted tree in $\Tc_{d+1}$ with root degree at most $d-5$ and field $\xi$. Then the set of fields implemented by  rooted 
	trees in $\Tc_{d+1}$ is dense in $\mathbb{S}$.
\end{lemma}

\begin{proof}
It follows from Lemma~\ref{lem: tree beyond lambda_k} that there is 
$\sigma \in \mathbb{S}$ with $|f_{3}'(\sigma)| > 1$ together with 
a sequence $\{\zeta_n\}_{n \geq 1}$
accumulating on $\sigma$ such that every $\zeta_n$ is the field implemented by a tree in $\Tc_{d+1}$  whose root degree is bounded by $(d-5)+3 = d-2$. We can 
now apply Lemma~\ref{lem: cantor implementation} with $k = 1$. It follows that either the set of fields implemented by  rooted 
	trees in $\Tc_{d+1}$ is dense in $\mathbb{S}$, or there is a tree in $\Tc_{d+1}$ with root degree at most $(d-2)+1 = d-1$ and field 
$\zeta \in \Arc{(\lambda_1,\overline{\lambda_1})}\setminus\{-1\}$. We conclude from 
Lemma~\ref{lem: elliptic -> irrational rotation} that $f_{\zeta,1}$ is conjugate to an irrational rotation 
and thus the orbit $\{f_{\zeta,1}^n(1)\}_{n \geq 1}$ is dense in $\mathbb{S}$. Every element in 
this orbit is implemented by  a rooted tree in $\Tc_{d+1}$ and thus we are done.
\end{proof}

\section{Proof of Lemma~\ref{lem: Many Cases Lemma}}
\label{sec: Proof of the main lemma}

We are now ready to prove Lemma~\ref{lem: Many Cases Lemma}, which we restate here for convenience. 
\begin{lemcases}
\statelemcases
\end{lemcases}
\begin{proof}
The proof consists of a careful case analysis. We give a seperate argument first for when $k$ is a power of two and for when
$k + 1$ is a power of two, then for each value of $k$ within the set $\{5,6,9,10,11,12,13,14,17\}$ and lastly we prove the statement for all other $k$. 

We remark that in some cases we will show that the set of fields implemented by rooted trees in $\Tc_{d+1}$ is dense in an arc $A$ of the circle. Since $b$ is  such 
that $|f_d'(z)|>1$ for all $z \in \mathbb{S}\setminus\{1\}$ (see \eqref{eq: Derivative} of Lemma~\ref{lem: properties f}), it follows that for all arcs $A$ 
there is an $N\geq 1$ such that $f_{\lambda,d}^N(A) = \mathbb{S}$. Density of fields in 
the whole unit circle therefore follows from density in $A$.

First suppose $k = 2^m$ is a power of two. In this case
$\xi \in \Arc{[\overline{\lambda_{2^{m-1}}},\lambda_{2^{m-1}}]}\setminus\{1\}$ is implemented by a rooted tree in $\Tc_{d+1}$  with root degree at most $d-2^m$ and with $|f_{2^m}'(\xi)| \geq 1$.
Let $\xi_2 = f_{\xi,1}(\xi)$. By item (\ref{enum: lem5}) of Lemma~\ref{lem: properties f}, we have $\xi\in \Arc{[\overline{\lambda_{1}},\lambda_{1}]}\setminus\{1\}$ and hence $\xi_2\neq \xi$  by item (\ref{enum: lem4}) of the same lemma. Moreover, by Lemma~\ref{lem: Tree Building}, $\xi_2$ is the field of a rooted tree in $\Tc_{d+1}$ with root degree at 
most $d-(2^m - 1)$. If $\xi_2\in  \Arc{(\lambda_{2^{m-1}},\overline{\lambda_{2^{m-1}}})}$, 
then the desired result follows from Lemma~\ref{lem: density powers of 2}. 
Otherwise $\xi,\xi_2 \in \Arc{[\overline{\lambda_{2^{m-1}}},\lambda_{2^{m-1}}]}$ 
and the result follows from applying Lemma~\ref{lem: easy even case} to these 
two parameters.

Now suppose $k+1$ is a power of two, so $k = 2^{m+1} - 1$ for $m\geq 1$. In this 
case $\xi \in \Arc{[\overline{\lambda_{2^{m}-1}},\lambda_{2^{m}-1}]}
\setminus\{1\}$ is  implemented by a rooted tree in $\Tc_{d+1}$ with root degree at most $d-(2^{m+1} - 1)$ and 
with $|f_{2^{m+1} - 1}'(\xi)| \geq 1$. If $\xi \in \Arc{(\lambda_{2^{m}},\overline{\lambda_{2^{m}}})}$ 
the result follows from Lemma~\ref{lem: density powers of 2}. Otherwise, if
$\xi \in \Arc{[\overline{\lambda_{2^m}},\lambda_{2^m}]}$, the result follows from
Lemma~\ref{lem: easy odd case}.

We now continue with the list of individual cases. \vskip 0.2cm
$\mathbf{k = 5:}$ In this case $\xi \in \Arc{[\overline{\lambda_2},\lambda_2]}\setminus\{1\}$ is the 
	field of a rooted tree with root degree at most $d-5$
	and with $|f_5'(\xi)| \geq 1$. If $\xi \in \Arc{(\lambda_3,\overline{\lambda_3})}$ the result 
	follows from Lemma~\ref{lem: d-5 after lambda_3}.
	Otherwise, if $\xi \in \Arc{[\overline{\lambda_3},\lambda_3]}$, the result follows from 
	Lemma~\ref{lem: easy odd case}.\vskip 0.15cm

$\mathbf{k = 6:}$ In this case $\xi \in \Arc{[\overline{\lambda_3},\lambda_3]}\setminus\{1\}$ is implemented by a rooted 
	tree in $\Tc_{d+1}$ with root degree at most $d-6$ and with $|f_6'(\xi)| \geq 1$. 
	Let $\xi_2 = f_{\xi,1}(\xi)$, which is the field of a rooted tree in $\Tc_{d+1}$ with root degree at most $d-5$.
	If
	$\xi_2\in \Arc(\lambda_3,\overline{\lambda_3})$ then the result follows from 
	Lemma~\ref{lem: d-5 after lambda_3}. Otherwise 
	$\xi,\xi_2 \in \Arc{[\overline{\lambda_{3}},\lambda_{3}]}$ and the result follows from 
	applying Lemma~\ref{lem: easy even case} to these two parameters. \vskip 0.15cm

$\mathbf{k = 9:}$ 	In this case $\xi \in \Arc{[\overline{\lambda_4},\lambda_4]}\setminus\{1\}$ is implemented by  a rooted tree in $\Tc_{d+1}$ with root degree at most $d-9$
	and with $|f_9'(\xi)| \geq 1$. Consider the orbit $\{f_{\xi,1}^n(\xi): n \geq 1\}$.
	The elements of this orbit are implemented by trees in $\Tc_{d+1}$ with root degree at most $d-8$ and they  
	accumulate on $R_1(\xi)$. Note that $|f_9'(R_1(\xi))| > 1$. 
	It follows from Lemma~\ref{lem: cantor implementation} that we either obtain the desired 
	density or we obtain a rooted tree with root degree at most $(d-8) + 3 = d-5$ that implements a field in $\Arc{(\lambda_3,\overline{\lambda_3})}$. In this latter case the result follows from 
	applying Lemma~\ref{lem: d-5 after lambda_3} to this tree. \vskip 0.15cm

$\mathbf{k = 10:}$ 	In this case $\xi \in \Arc{[\overline{\lambda_5},\lambda_5]}\setminus\{1\}$ is implemented by a 
	rooted tree in $\Tc_{d+1}$ with root degree at most $d-10$ and with $|f_{10}'(\xi)| \geq 1$.  Then it follows 
	from Lemma~\ref{lem: Lower Bound three maps} that either the orbit of $1$ under the action 
	of the semigroup generated by $f_{\xi,3},f_{\xi,4}$ and $f_{\xi,5}$ is dense in an arc 
	of $\mathbb{S}$, in which case the result follows. Or we can conclude that
	$|f_5'(R_5(\xi))| > \frac{43}{50}$. In that case we consider the orbit
	$\mathcal{R}=\{f_{\xi,5}^n(\xi):n \geq 1\}$. This orbit accumulates on $R_5(\xi)$ and every 
	element is implemented by a rooted tree with root degree at most $d-10+5 = d-5$. If 
	$R_5(\xi) \in \Arc{(\lambda_3,\overline{\lambda_3})}$ then there are also fields
	$\zeta \in \mathcal{R}$ with $\zeta \in \Arc{(\lambda_3,\overline{\lambda_3})}$. In that case 
	we can apply Lemma~\ref{lem: d-5 after lambda_3} to obtain density of the fields.
	Otherwise, if $R_5(\xi) \in \Arc{[\overline{\lambda_3},\lambda_3]}$, then we can 
	find $\zeta_1,\zeta_2 \in \mathcal{R}$ such that
	$\zeta_1,\zeta_2 \in \Arc{[\overline{\lambda_3},\lambda_3]}$ are distinct, lie in the same 
	half-plane and $|f_5'(\zeta_i)| > \frac{43}{50}$
	for $i = 1,2$. It follows that for both fields $\zeta_i$ we have
	\[
		|f_6'(R_3(\zeta_i))|  > |f_6'(\zeta_i)| 
		= \frac{6}{5} \cdot |f_5'(\zeta_i)| 
		> \frac{6}{5} \cdot \frac{43}{50} 
		= \frac{129}{125} > 1.
	\]
	Density of the fields now follows from applying 
	Lemma~\ref{lem: easy even case} to $\zeta_1$ and $\zeta_2$. \vskip 0.15cm
	
$\mathbf{k = 11:}$
	In this case $\xi \in \Arc{[\overline{\lambda_5},\lambda_5]}\setminus\{1\}$ is implemented by a rooted tree in $\Tc_{d+1}$ with root degree at most $d-11$
	and with $|f_{11}'(\xi)| \geq 1$. If $\xi \in \Arc{[\overline{\lambda_6},\lambda_6]}$ 
	the result follows from Lemma~\ref{lem: easy odd case}.
	Otherwise, if $\xi \in \Arc{(\lambda_6,\overline{\lambda_6})}$, we apply
	Lemma~\ref{lem: tree beyond lambda_k} to find a parameter 
	$\sigma\in\mathbb{S}$ with $|f_6'(\sigma)| >1$ together with a sequence
	of fields $\{\zeta_n\}_{n \geq 1}$ accumulating on $\sigma$ such that 
	every $\zeta_n$ is implemented by a rooted tree in $\Tc_{d+1}$ whose root degree is at most $d-11+6 = d-5$. If there is any 
	$\zeta_n \in \Arc{(\lambda_3,\overline{\lambda_3})}$ then density of the fields in the circle
	follows from Lemma~\ref{lem: d-5 after lambda_3}. Otherwise the sequence 
	accumulates on $\sigma$ from inside $\Arc{[\overline{\lambda_3},\lambda_3]}$ 
	and thus we can find $\zeta_{n_1},\zeta_{n_2} \in \Arc{[\overline{\lambda_3},\lambda_3]}$
	that are distinct, lie in the same half-plane and have the property that $|f_6'(\zeta_i)| > 1$
	for $i = 1,2$. The desired density now follows from applying 
	Lemma~\ref{lem: easy even case} to $\zeta_1$ and $\zeta_2$. \vskip 0.15cm

$\mathbf{k = 12:}$
	In this case $\xi \in \Arc{[\overline{\lambda_6},\lambda_6]}\setminus\{1\}$ 
	is implemented by a rooted tree in $\Tc_{d+1}$ with root degree at most $d-12$ and with $|f_{12}'(\xi)| \geq 1$. 
	This case can be done in a very similar way to the $k=9$ case. Consider the orbit
	$\{f_{\xi,1}^n(\xi): n \geq 1\}$. The elements of this orbit are fields of trees in $\Tc_{d+1}$ with root 
	degree at most $d-11$ and they accumulate on $R_1(\xi)$. Note that $|f_{12}'(R_1(\xi))| > 1$. 
	It follows from Lemma~\ref{lem: cantor implementation} that we either obtain the desired 
	density or we obtain a rooted tree with root degree at most $(d-11) + 4 = d-7$ and
	field in $\Arc{(\lambda_4,\overline{\lambda_4})}$. In this latter case the result follows 
	from applying Lemma~\ref{lem: density powers of 2}
	to this tree. \vskip 0.15cm

$\mathbf{k = 13:}$
	In this case $\xi \in \Arc{[\overline{\lambda_6},\lambda_6]}\setminus\{1\}$ is implemented by a rooted tree in $\Tc_{d+1}$ with root degree at most $d-13$
	and with $|f_{13}'(\xi)| \geq 1$.  Then it follows from Lemma~\ref{lem: Lower Bound three maps} 
	that either the orbit of $1$ under the action 
	of the semigroup generated by $f_{\xi,4},f_{\xi,5}$ and $f_{\xi,6}$ is dense in an arc of $\mathbb{S}$, 
	in which case the result follows.
	Or we can conclude that $|f_6'(R_6(\xi))| > \frac{10}{13}$. In that case we consider the 
	orbit $\mathcal{R}=\{f_{\xi,6}^n(\xi):n \geq 1\}$.
	This orbit accumulates on $R_6(\xi)$ and every element is implemented by a rooted tree in $\Tc_{d+1}$ with root 
	degree at most $d-13+6 = d-7$. If 
	$R_6(\xi) \in \Arc{(\lambda_4,\overline{\lambda_4})}$ then there is also a field $\zeta \in \mathcal{R}$ with $\zeta \in 
	\Arc{(\lambda_4,\overline{\lambda_4})}$. In that case we can apply 
	Lemma~\ref{lem: density powers of 2} to obtain density of the fields.
	Otherwise, if $R_6(\xi) \in \Arc{[\overline{\lambda_4},\lambda_4]}$, 
	then we can find $\zeta_1,\zeta_2 \in \mathcal{R}$ such that 
	$\zeta_1,\zeta_2 \in \Arc{[\overline{\lambda_4},\lambda_4]}$ are distinct, 
	lie in the same half-plane and $|f_6'(\zeta_i)| > \frac{10}{13}$
	for $i = 1,2$. It follows that for both fields $\zeta_i$ we have
	\[
		|f_8'(R_4(\zeta_i))|  > |f_8'(\zeta_i)| 
		= \frac{8}{6} \cdot |f_6'(\zeta_i)| 
		> \frac{8}{6} \cdot \frac{10}{13} 
		= \frac{40}{39}
		> 1.
	\]
	Density of the fields now follows from applying Lemma~\ref{lem: easy even case} 
	to $\zeta_1$ and $\zeta_2$. \vskip 0.15cm

$\mathbf{k = 14:}$
	In this case $\xi \in \Arc{[\overline{\lambda_7},\lambda_7]}\setminus\{1\}$ is 
	implemented by a rooted tree in $\Tc_{d+1}$ with root degree at most $d-14$
	and with $|f_{14}'(\xi)| \geq 1$.  Then it follows from Lemma~\ref{lem: Lower Bound three maps} 
	that either the orbit of $1$ under the action 
	of the semigroup generated by $f_{\xi,5},f_{\xi,6}$ and $f_{\xi,7}$ 
	is dense in an arc of $\mathbb{S}$, in which case the result follows.
	Or we can conclude that $|f_7'(R_7(\xi))| > \frac{89}{98}$. In that case 
	we consider the orbit $\mathcal{R}=\{f_{\xi,7}^n(\xi):n \geq 1\}$.
	This orbit accumulates on $R_7(\xi)$ and every element is implemented by a rooted tree in $\Tc_{d+1}$ with root degree at most $d-14+7 = d-7$. If 
	$R_7(\xi) \in \Arc{(\lambda_4,\overline{\lambda_4})}$ then there is also a field $\zeta \in \mathcal{R}$ with $\zeta \in \Arc{(\lambda_4,\overline{\lambda_4})}$. 
	In that case we can apply Lemma~\ref{lem: density powers of 2} to 
	obtain density of the fields. Otherwise, if $R_7(\xi) \in \Arc{[\overline{\lambda_4},\lambda_4]}$, 
	then we can find $\zeta_1,\zeta_2 \in \mathcal{R}$ such that 
	$\zeta_1,\zeta_2 \in \Arc{[\overline{\lambda_4},\lambda_4]}$ are distinct, lie in 
	the same half-plane and $|f_7'(\zeta_i)| > \frac{89}{98}$
	for $i = 1,2$. It follows that for both fields $\zeta_i$ we have
	\[
		|f_8'(R_4(\zeta_i))|  
		> |f_8'(\zeta_i)| 
		= \frac{8}{7} \cdot |f_7'(\zeta_i)| 
		> \frac{8}{7} \cdot \frac{89}{98} 
		= \frac{356}{343}
		> 1.
	\]
	Density of the fields now follows from applying Lemma~\ref{lem: easy even case} 
	to $\zeta_1$ and $\zeta_2$.\vskip 0.15cm

$\mathbf{k = 17:}$
	In this case $\xi \in \Arc{[\overline{\lambda_8},\lambda_8]}\setminus\{1\}$ is 
	implemented by a rooted tree in $\Tc_{d+1}$ with root degree at most $d-17$
	and with $|f_{17}'(\xi)| \geq 1$. If $\xi \in \Arc{[\overline{\lambda_9},\lambda_9]}$ 
	the result follows from Lemma~\ref{lem: easy odd case},
	therefore we assume that $\xi \in \Arc{(\lambda_9, \overline{\lambda_9})}$.
	We apply
	Lemma~\ref{lem: tree beyond lambda_k} to find a parameter 
	$\sigma\in\mathbb{S}$ with $|f_9'(\sigma)| >1$ together with a sequence
	of fields $\{\zeta_n\}_{n \geq 1}$ accumulating on $\sigma$ such that 
	every $\zeta_n$ is implemented by a rooted tree in $\Tc_{d+1}$ whose root degree is at most $d-17+9= d-8$.
	It follows from Lemma~\ref{lem: cantor implementation} that we either obtain the 
	required density of fields or there is a tree in $\Tc_{d+1}$ whose root degree is bounded by $d-5$ 
	with field inside $\Arc{(\lambda_3, \overline{\lambda_3})}$. In the latter case
	the result follows from Lemma~\ref{lem: d-5 after lambda_3}.\vskip 0.15cm
	
Finally we complete the proof for $k > 17$. In that case write $k = 2m$ if $k$ is even 
and $k = 2m+1$ if $k$ is odd. Note that $m \geq 9$. We are then given that
$\xi \in \Arc{[\overline{\lambda_m},\lambda_m]}\setminus\{1\}$ is implemented by a rooted tree in $\Tc_{d+1}$
with root degree at most $d-k$ and with $|f_{k}'(\xi)| \geq 1$. It follows from 
Corollary~\ref{cor: remaining cases} that the orbit of $1$ under the action of the semigroup 
generated by $f_{\xi,m-3},f_{\xi,m-2},f_{\xi,m-1}$ and $f_{\xi,m}$ is dense in an arc 
of $\mathbb{S}$ from which our desired conclusion follows. This finishes the proof of Lemma~\ref{lem: Many Cases Lemma}.
\end{proof}

\section{Fast implementation of fields}\label{sec:fast}
In this section, we bootstrap Theorem~\ref{thm: Density of Ratios} to obtain fast algorithms for implementing fields which will be important in our reductions. For a number $\alpha=p/q\in \mathbb{Q}$ with $\mathrm{gcd}(p,q)=1$, we use $\emph{size}(\alpha)$ to denote the total number of bits needed to represent $p,q$,  and we extend this to numbers in $\CQ$ by adding the sizes of the real and imaginary parts. For $\alpha_1,\hdots,\alpha_t\in \CQ$, we denote by $\size{\alpha_1,\hdots,\alpha_t}$ the total of the sizes of $\alpha_1,\hdots,\alpha_t$.

\newcommand{\statelemmainlemma}{Fix an integer $\Delta\geq 3$, a rational number $b\in (0,1)$ and $\lambda\in \SQ(\Delta-1,b)$. Then, there is an algorithm, which on input $\hat{\lambda}\in \SQ$ and rational $\epsilon>0$, returns in time $poly(\size{\hat{\lambda},\epsilon}))$ a rooted tree $T$ in $\Tc_\Delta$ with root degree $1$ that implements a field $\lambda'$ such that $|\lambda'-\hat{\lambda}|\leq \epsilon$.}

\begin{lemma}\label{lem:mainlemma}
\statelemmainlemma
\end{lemma}

\begin{proof}[Proof of Lemma~\ref{lem:mainlemma}]
Let $d=\Delta-1$.  We start by setting up some parameters that will be useful.

Let $\lambda_1$ be as in Lemma~\ref{lem: properties f}.  As $\tilde{\lambda}$ approaches $\lambda_1$ from inside $\Arc{(1,\lambda_1)}$ we know
that $R_1(\tilde{\lambda})$ approaches $R_1(\lambda_1)$. Since $|f_1'(R_1(\lambda_1))| = 1$
there must be $\tilde{\lambda} \in \Arc{(1,\lambda_1)}$ such that $|f_1'(R_1(\xi))| \in (\frac{1}{2},1)$ for all $\xi \in \Arc{(\tilde{\lambda},\lambda_1)}$. 
By Theorem~\ref{thm: Density of Ratios}, there exist trees $T_1,T_2$ in $\Tc_{d+1}$ with root degree $1$ and fields $\xi_1,\xi_2 \in \Arc{(\tilde{\lambda},\lambda_1)} \cap \SQ$ with $\Arg(\xi_1) < \Arg(\xi_2)$.  Because the map 
$\xi \mapsto R_1(\xi)$ is orientation preserving with nonzero derivative
we have $\Arg(R_1(\xi_1)) < \Arg(R_1(\xi_2))$.
For $i\in \{1,2\}$, the fixed point $R_1(\xi_i)$  is a solution to the quadratic equation $\xi_i (z+b) = z (b z + 1)$, and hence we can approximate it with any desired rational precision $\tau>0$ in time $poly(\size{\tau})$.

Let $I = \Arc(R_1(\xi_1),R_1(\xi_2))$ and note that this arc is contained in the upper half-plane. 
We will show that the arc $I$ gets mapped onto $\mathbb{S}$ in a fixed number of applications of $f_{\lambda,d}$.
The idea of the algorithm is then to find a small enough neighborhood of a point in $I$ that gets mapped
close to the field that we are trying to (approximately) implement. Then we use that we are able to quickly and accurately 
approach any value inside $I$ using $f_{\xi_1,1}$ and $f_{\xi_2,1}$. This algorithm is very similar 
to the proof of Lemma~\ref{lem: Density in Arc}. 

We now show that $I$ gets mapped onto $\mathbb{S}$ in a fixed number of applications of $f_{\lambda,d}$.
We first consider the case that $b\in (0,\frac{d-1}{d+1}]$.
Let $C_1 = \lvert f_{\lambda,d}'(1) \rvert = d\frac{1-b}{1+b}$ and let
$C_2 = \lvert f_{\lambda,d}'(-1) \rvert = d\frac{1+b}{1-b}$. Note that $C_1$ and $C_2$ are both greater than 
one and that for any $z \in \mathbb{S}$ the inequality $C_1 \leq \lvert f_{\lambda,d}'(z) \rvert \leq C_2$
holds (cf. item~(\ref{enum: lem1}) of Lemma~\ref{lem: properties f}). This means that for any circular arc $J$ and integer $n$ we get
\begin{equation}
	\label{eq: Length Bound}
	C_1^n \cdot \ell(J) \leq \ell(f_{\lambda,d}^n(J)) \leq C_2^n \cdot \ell(J).
\end{equation}
From this, we deduce that $f_{\lambda,d}^N(I)=\mathbb{S}$, where $N = \big\lceil \tfrac{\log(2\pi/\ell(I))}{\log(C_1)} \big\rceil$.

Next, in case $b\in (\frac{d-1}{d+1},1)$, we recall the conformal metric $\mu$ from the proof of Lemma~\ref{lem: expanding orbit}. 
Let us denote the length of a circular arc $J$ with respect to this metric by $\text{length}(J)$ and denote $c=\text{length}(\mathbb{S})$.
Since there exists a constant $\kappa>1$ such that $f_{d,\lambda}$ is uniformly expanding on $\mathbb{S}$ with a factor $\kappa$ with respect to this metric, it follows that $f_{\lambda,d}^N(I)=\mathbb{S}$, where $N = \big\lceil \tfrac{\log(c/\text{length}(I))}{\log(\kappa)} \big\rceil$.
Note that the right-hand side of \eqref{eq: Length Bound} is also valid for $b\in (\frac{d-1}{d+1},1)$ (with $C_2$ defined in the same way).

Let  $ x_0, \dots, x_m $  be points  such that the clockwise arcs between $x_{i-1}$ and 
$x_i$ form a partition of $I$ with $x_0=R_1(\xi_1)$,  $x_m = R_1(\xi_2)$ and chosen so that $x_1,\hdots, x_{m-1}\in \SQ$ and  the length of an 
arc between two subsequent points is less than $2\pi/C_2^N$. In this way we ensure that these arcs
are not mapped onto the whole circle by $N$ applications of $f_{\lambda,d}$ and thus each arc 
is bijectively mapped to an arc on the unit circle by $f_{\lambda,d}^N$.

We now  describe an algorithm that, on input $\hat{\lambda} \in \SQ$ and rational $\epsilon >0$,
yields in $poly(\size{\hat{\lambda},\epsilon})$ a rooted tree $\hat{T}$ in $\Tc_{d+1}$ with $\mathcal{O}(\log(\epsilon^{-1}))$ vertices whose field has 
distance at most $\epsilon$ from  $\hat{\lambda}$; we will account later for the degree  of the root.  We assume for convenience that $\epsilon \ll \ell(I)$.

The first step of the algorithm is to find $i \in \{1, \dots, m\}$ such that 
$\hat{\lambda} \in \Arc{[f_{\lambda,d}^N(x_{i-1}),f_{\lambda,d}^N(x_{i})]}$.  We know that such an arc must exist 
because $I$ is mapped surjectively onto $\mathbb{S}$ by $f_{\lambda,d}^N$ and, since  $f_{\lambda,d}^N(z)$ is a rational function of $z$ with fixed degree, we can find $i$ in time $poly(\size{\hat{\lambda}})$.  Now we consider the bijective map
\[
	f_{\lambda,d}^N: \Arc{[x_{i-1},x_i]} \to \Arc{[f_{\lambda,d}^N(x_{i-1}),f_{\lambda,d}^N(x_{i})]}.
\] 
Analogously, with $n = \lceil \log_{3/2}(\ell{(\Arc{[x_{i-1},x_i]})} \cdot C_2^N / \epsilon) \rceil$ applications 
of $f_{\lambda,d}^N$, we can determine using binary search  in time $poly(\size{\hat{\lambda},\epsilon})$  an arc $J \subseteq \Arc{[x_{i-1},x_i]}$ with endpoints in $\SQ$ such that $\hat{\lambda} \in f_{\lambda,d}^N(J)$ and whose length satisfies
\[
	3^{-n} \cdot \ell(\Arc{[x_{i-1},x_i]})\leq \ell(J) \leq (2/3)^{n} \cdot \ell(\Arc{[x_{i-1},x_i]}) \leq \epsilon/ C_2^N
\]
 Note that the length of $J$ is bounded below 
by $C_3 \cdot \epsilon^5$, where $C_3$ is a constant independent of $\hat{\lambda}$ or~$\epsilon$. It follows from
(\ref{eq: Length Bound}) that $\ell{(f_{\lambda,d}^N(J))} \leq \epsilon$, which means 
that the arc $J$ is mapped by $f_{\lambda,d}^N$ to an arc of length at most $\epsilon$, that includes $\hat{\lambda}$.  We will next show how to construct in $poly(\size{\hat{\lambda},\epsilon})$ a rooted tree $T$ in $\Tc_{d+1}$ with $s=\mathcal{O}(\log(\epsilon^{-1}))$ vertices that implements a field $w\in J$. Then, using Lemma~\ref{lem: Tree Building},\footnote{\label{foot:size}Lemma~\ref{lem: Tree Building} describes how to construct a tree of size $s \cdot d +1$
with field $f_{\lambda,d}(z)$ from a tree of size $s$ and field $z$. Repeating this construction $N$ times yields the construction 
of $\hat{T}$ from $T$.}   we obtain a
rooted tree $\hat{T}$ with $ (d^N-1)/(d-1) +d^N  s $ vertices that implements the field $\lambda'=f_{\lambda,d}^N(w)$ with $|\lambda'-\hat{\lambda}|\leq \epsilon$.

To construct $T$, we first fix some constants. Let $C_4 = \lvert f_{1}'(R_1(\xi_1)) \rvert$ and $C_5 = \lvert f_{1}'(R_1(\xi_2)) \rvert$ and 
note that $C_4, C_5 \in (\frac{1}{2},1)$. We also have
$C_4 \leq |f_{1}'(z)| \leq C_5$ for all $z \in I$. It follows that
$f_{\xi_2,1}(I) = \Arc{[f_{\xi_2,1}(R_1(\xi_1)),R_1(\xi_2)]}$ is contained 
in $I$ and its length is strictly bigger than $\ell(I)/2$. Furthermore it follows that 
$f_{\xi_1,1}^{-1}(\Arc{[R_1(\xi_1),f_{\xi_2,1}(R_1(\xi_1))]}) = \Arc{[R_1(\xi_1),f_{\xi_1}^{-1}(f_{\xi_2}(R_1(\xi_1)))]}$ is strictly 
contained inside $I$. Let $J_0 = J$ and for $k\geq 0$, as long as $f_{\xi_2,1}(R_1(\xi_1)) \not \in J_k$, define
\[
	J_{k+1}
	=
	\begin{cases} 
      f^{-1}_{\xi_1,1}(J_k) & \text{if }J_k \subset \Arc{[R_1(\xi_1),f_{\xi_2,1}(R_1(\xi_1))]}\\
      f^{-1}_{\xi_2,1}(J_k) & \text{if }J_k \subset \Arc{[f_{\xi_2,1}(R_1(\xi_1))),R_1(\xi_2)]}.
   \end{cases}
\]
We have that $J_k \subseteq I$ for every $k$ and
$\ell(J_k) \geq C_5^{-k} \cdot \ell{(J_0)} \geq C_3 \cdot C_5^{-k} \cdot \epsilon^5$.
Because 
$C_5 < 1$, we deduce that there is $N_1 \geq 0$ such that $f_{\xi_2,1}(R_1(\xi_1)) \in J_{N_1}$ where $N_1$ is bounded above by
\[
	\Big \lceil \tfrac{\log(C_3 \cdot \epsilon^5 / \ell(I))}{\log(C_5)} \Big \rceil = \mathcal{O}(\log(\epsilon^{-1})).
\]
Let $i_1, \dots, i_{N_1}$ be the sequence of indices such that $f_{\xi_{i_k}}(J_{k}) = J_{k-1}$ and note that these can be computed in $poly(\size{\hat{\lambda},\epsilon})$ time.
Let $K = f_{\xi_2,1}^{-1}(J_{N_1})$. We see that $R_1(\xi_1) \in K$ and
\[
	\big(f_{\xi_{i_1},1}\circ\cdots\circ f_{\xi_{i_{N_1}},1} \circ f_{\xi_2,1}\big)(K) = J.
\]
Furthermore, because the maps $f_{\xi_i,1}^{-1}$ are expanding on $I$, we find $\ell(K) 
\geq \ell(J) \geq C_3 \cdot \epsilon^5$. This means that there is an arc of length at least 
$\frac{1}{2} \cdot C_3\cdot \epsilon^5$ extending from $R_1(\xi_1)$, going either clockwise or 
counterclockwise, contained in $K$. In the case that such a clockwise arc exists, i.e. 
$\Arc[R_1(\xi_1)\cdot e^{-i \frac{1}{2}C_3 \epsilon^5},R_1(\xi_1)]\subseteq K$, we 
see that, because $R_1$ is an attracting fixed point of $f_{\xi_1}$, there is some $N_2$, specified below, 
such that $f_{\xi_1,1}^{N_2}(\xi_1) \in K$. Using that for integers $n$ we have
\[
	\ell(\Arc{[f_{\xi_1,1}^{n}(\xi_1),R_1(\xi_1)]}) 
	= 
	\ell(f_{\xi_1,1}^{n}(\Arc{[\xi_1,R_1(\xi_1)]}))
	\leq
	C_4^n \cdot \ell(\Arc{[\xi_1,R_1(\xi_1)])}
	< C_4^n \cdot 2 \pi,
\]
we see that it suffices to take $N_2 = \big\lceil \tfrac{\log(C_3 \cdot \epsilon^5/(4 \pi))}{\log(C_4)}\big \rceil 
	= \mathcal{O}(\log(\epsilon^{-1}))$. In the case that such a clockwise arc does not exist, we find that a counterclockwise arc of length
$\frac{1}{2} \cdot C_3\cdot \epsilon^5$ is contained in $K$. Note that there is some integer $N_c$
independent of $\hat{\lambda}$ and $\epsilon$ such that $f_{\xi_2,1}^{N_c}(\xi_1) \in I$. The same analysis 
as above shows that then $(f_{\xi_1,1}^{N_2} \circ f_{\xi_2,1}^{N_c})(\xi_1) \in K$. We 
let $N_3$ be equal to zero if a clockwise arc of sufficient length is contained in $K$ and 
otherwise we let $N_3 = N_c$. 
We conclude that 
\begin{equation}
	\label{eq: final par}
	\big(f_{\lambda,d}^N \circ f_{\xi_{i_1},1}\circ\cdots\circ f_{\xi_{i_{N_1}},1} \circ f_{\xi_2,1} 
			\circ f_{\xi_1,1}^{N_2} \circ f_{\xi_2,1}^{N_3}\big)(\xi_1)
\end{equation}
has a distance at most $\epsilon$ from $\hat{\lambda}$. By repeatedly applying the constructions laid out in Lemma~\ref{lem: Tree Building} (cf. Footnote~\ref{foot:size}), we conclude that we can construct a tree $T$ in $\Tc_{d+1}$ whose field is given by the value in (\ref{eq: final par}) and with   $\mathcal{O}(\log(\epsilon^{-1}))$ vertices. 

This finishes the description of the algorithm, modulo that the root of the tree we constructed has degree  $d$. To obtain a rooted tree with root degree 1, we run
the algorithm described on input $f_{\lambda,1}^{-1}(\hat{\lambda})$ and $\epsilon \cdot \frac{d-1}{d+1}$ to obtain a rooted tree with 
root degree $d$ and field $\zeta$ with $|f_{\lambda,1}^{-1}(\hat{\lambda})- \zeta| < \epsilon \cdot \frac{d-1}{d+1}$.
Attaching one new vertex by an edge to this root yields a rooted tree with root degree $1$ and field $f_{\lambda,1}(\zeta)$ which satisfies, using Item~\ref{enum: lem1} of Lemma~\ref{lem: properties f}, that
\[|\hat{\lambda} - f_{\lambda,1}(\zeta)| \leq |f_{\lambda,1}^{-1}(\hat \lambda)- \zeta| \cdot \max_{z \in S} |f_{\lambda,1}'(z)|
< \epsilon \cdot \frac{d-1}{d+1} \cdot \frac{d+1}{d-1} = \epsilon,
\]
as wanted. This finishes the proof of Lemma~\ref{lem:mainlemma}.
\end{proof}

\section{Reduction}\label{sec:reduction}
In this section, we prove our inapproximability results.  Throughout this section, we use $\Gc_\Delta$ to denote the set of all graphs  with maximum degree at most $\Delta$. We start in Section~\ref{sec:prelims} with some preliminaries that will be used in our proofs, Section~\ref{sec:reductionstep} gives the main reduction, and we show how to use this in Section~\ref{sec:proofmain} to conclude the proof of Theorem~\ref{thm:main2}.
\subsection{Preliminaries}\label{sec:prelims}

We will  use the following lemma from \cite{PetersRegts2018}.
\begin{lemma}[\mbox{\cite{PetersRegts2018}}]\label{lem:zerofree}
Let $\Delta\geq 3$ be an integer and let $\lambda\in \SQ$ with $\lambda\neq -1$. Then, there exists $\eta=\eta(\Delta,\lambda)>1$ such that, for all $b\in (1/\eta,\eta)$, for all graphs $G\in \Gc_\Delta$, it holds that $Z_G(\lambda,b)\neq 0$.
\end{lemma}

For a graph $G$ and vertices $u,v$ in $G$, let $Z_{G,\plm u,\plm v}(\lambda,b)$ denote the contribution to the partition function when $u,v$ are assigned the spins $\plm$, respectively. For a configuration $\sigma$ on $G$, we use $w_{G,\sigma}(\lambda,b)$ to denote the weight $\lambda^{|n_\pl(\sigma)|} b^{\delta(\sigma)}$ of $\sigma$. We will use the following observation.
\begin{lemma}\label{lem:conjugation}
Let $\lambda \in \mathbb{S}$ and $b\in \mathbb{R}$. Then, for an arbitrary graph $G=(V_G,E_G)$  and vertices $u,v$ of $G$ it holds that
\[Z_{G,\pl u,\pl v}(\lambda,b)=\lambda^{|V(G)|}\,\conj{Z_{G,\mi u,\mi v}(\lambda,b)},\quad Z_{G,\pl u,\mi v}(\lambda,b)=\lambda^{|V(G)|}\,\conj{Z_{G,\mi u,\pl v}(\lambda,b)}.\]
\end{lemma}
\begin{proof}
For an assignment $\sigma:V_G\rightarrow \{\pl,\mi \}$, let $\bar{\sigma}:V_G\rightarrow \{\pl,\mi \}$ be the assignment obtained by interchanging the assignment of $\pl$'s with $\mi$'s. Then 
\[w_{G,\bar{\sigma}}(\lambda,b)=\lambda^{|n_{\pl}(\bar{\sigma})|}b^{\delta(\bar{\sigma})}=\lambda^{|V_G|-|n_{\pl}(\sigma)|}b^{\delta(\sigma)}=\lambda^{|V(G)|}\conj{w_{G,\sigma}(\lambda,b)}.\]
The result follows by summing over the relevant $\sigma$ for each of $Z_{G,\pl u,\pl v}(\lambda,b)$ and $Z_{G,\pl u,\mi v}(\lambda,b)$.
\end{proof}

The following lemma will be useful in general for handling rational points on the circle. Ideally, we would like to describe a number on $\mathbb{S}$ by a rational angle, but this may not correspond to a rational cartesian point which would complicate computations. However, rational points are dense on the circle and we can compute one arbitrarily close to a given angle, as follows. 
\begin{lemma}\label{lem:rationalangle}
Given a rational angle $\theta \in [0 ,2\pi)$ and $\epsilon \in (0,1)$, there exists a number $\hat{\theta}$ such that $|\theta - \hat{\theta}|< \epsilon$ and $\cos \hat{\theta}, \sin \hat{\theta} \in \mathbb{Q}$ are rational numbers of size at most $poly(\size{\theta,\epsilon})$. Furthermore, we can compute $\cos \hat{\theta}$ and $\sin \hat{\theta}$ in time $poly(\size{\theta,\epsilon})$.
\end{lemma}
\begin{proof}
By symmetry, we may assume that $\theta \in [0, \pi/4]$. Given $\theta$, take a rational approximation $r$ of $\tan(\theta / 2)$ such that $|\tan( \theta / 2) - r| < \epsilon / 2$. We claim that $\hat{\theta}=2 \arctan(r)$ has the desired properties. 

Write $s,c,t$ respectively for $\sin \hat{\theta}, \cos \hat{\theta}, \tan \hat{\theta}$. Using the tan double angle formula we have
$s/c = t = 2r/(1-r^2)$. We also know that $s^2 + c^2 = 1$. Solving these simultaneously gives that $s = 2r/(1+r^2)$ and $c = (1 - r^2)/(1+r^2)$, which are both rational since $r$ is rational.

Also writing $f(x) = 2\arctan(x)$ for $x \in [0,1]$, note that $f'(x) = 2/(1+x^2) \in [1,2]$ for $x \in [0,1]$. Hence $|f(x) - f(y)| \leq 2|x-y|$ for $x,y \in [0,1]$ and so $|\theta - \theta'| < \epsilon$.

Finally we can compute $r$ in $poly(\size{\theta,\epsilon})$ using a series expansion of $\tan$ from which we can compute  $s$ and $c$ from the formulas above.
\end{proof}

Finally, we will use the following well-known lemma for continued-fraction approximation.
\begin{lemma}[{\cite[Corollary 6.3a]{Schrijver}}]\label{lem:diophantine}
There is a poly-time algorithm which, on input a rational number $\alpha$ and integer  $K\geq1$, decides whether there exists a rational number $p/q$ with $1 \leq q \leq K$ and  $|\alpha - (p/q)| < 1/2K^2$,  and if so, finds this (unique) rational number.
\end{lemma}

\subsection{The reduction}\label{sec:reductionstep}
To prove Theorem~\ref{thm:main2}, we will show how to use a poly-time algorithm for $\FactorIsing{K}$ and $\ArgIsing{\rho}$ to compute exactly $Z_G(\lambda,\hat{b})$ on graphs $G$ of maximum degree three  for some appropriate value of $\hat{b}$ that we next specify.

Let $\eta=\eta(3,\lambda)>1$ be as in Lemma~\ref{lem:zerofree}, so that 
\begin{equation}\label{eq:nonzero1}
Z_G(\lambda,b')\neq 0 \mbox{ for all } b'\in (1/\eta,\eta) \mbox{ and } G\in \mathcal{G}_3.
\end{equation}
For $k=2,3,\hdots,$ let $P_k$ be the path with $k$ vertices whose endpoints are labeled $u_k,v_k$ and all vertex activities are equal to 1. Then, it is not hard to see that 
\begin{equation}\label{eq:r4f4123}
\left[\begin{array}{cc} Z_{P_k,\pl u_k,\pl v_k}(1,b)& Z_{P_k,\pl u_k,\mi v_k}(1,b) \\ Z_{P_k,\mi u_k,\pl v_k}(1,b)&Z_{P_k,\mi u_k,\mi v_k}(1,b)\end{array}\right]=\left[\begin{array}{cc} 1& b \\ b& 1\end{array}\right]^{k-1}.
\end{equation}
Clearly, for all $k$ it holds that 
\begin{equation}\label{eq:g4g4635335r}
Z_{P_k, +u_k,\pl v_k}(1,b)=Z_{P_k,\mi  u_k,\mi v_k}(1,b) \mbox{ and } \frac{Z_{P_k,\pl u_k,\mi v_k}(1,b)}{Z_{P_k,\mi u_k,\mi v_k}(1,b)}=\frac{Z_{P_k,\mi u_k,\pl v_k}(1,b)}{Z_{P_k,\mi u_k,\mi v_k}(1,b)}=:b_k.
\end{equation}
Moreover, using \eqref{eq:r4f4123}, we have that there exists $k$ such that 
\begin{equation}\label{eq:2wsx}
1/\eta<\hat{b}=b_k<\eta.
\end{equation}
By the choice of $k$ and \eqref{eq:nonzero1}, we conclude that 
\begin{equation}\label{eq:nonzero}
Z_G(\lambda,\hat{b})\neq 0 \mbox{ for all } G\in \mathcal{G}_3.
\end{equation}

The main step in the reduction is captured by  the following  lemma.
\begin{lemma}\label{lem:maincomp}
Let $\Delta\geq 3$ be an integer, $b \in(0,1)$ be a rational and let $\lambda\in \CQ(\Delta-1,b)$   Let $K=1.001$ and $\rho=\pi/40$.
Assume that a poly-time algorithm exists for either $\FactorIsing{K}$ or $\ArgIsing{\rho}$. Then, there exists a poly-time algorithm that on input a graph $G\in \mathcal{G}_3$ and an edge $e=\{u,v\}$ of $G$, outputs the value of the ratio 
\[R_{G,e}=\frac{\hat{b}^2 z_{\pl\pl}+ \hat{b}(z_{\pl\mi}+z_{\mi\pl})+z_{\mi\mi}}{\hat{b}^2z_{\mi\mi}+\hat{b}(z_{\pl\mi}+z_{\mi\pl})+z_{\pl\pl}}, \mbox{\, where \, } z_{\plm\plm}:=Z_{G\backslash e,\plm u,\plm v}(\lambda,\hat{b}).\]
The algorithm also outputs the value of the ratio $R_{G,e}'=z_{\mi\mi}/z_{\pl\pl}$, provided that $z_{\pl\pl}\neq 0$.
\end{lemma}
\begin{remark}
As it will be shown in the proof of Lemma~\ref{lem:maincomp}, the ratio $R_{G,e}$ is well-defined for all graphs $G\in \Gc_3$ and edges $e$ in $G$ using the zero-free region in Lemma~\ref{lem:zerofree} and the choice of $\hat{b}$. It is harder to show that $R_{G,e}'$ is well-defined (we cannot use Lemma~\ref{lem:zerofree} directly) and hence the need for the assumption that $z_{\pl\pl}\neq 0$ in Lemma~\ref{lem:maincomp}.
\end{remark}
\begin{proof}
Suppose that $G=(V,E)$ with $n=|V|$ and $m=|E|$. Let 
\begin{equation}\label{eq:rtrptp}
\begin{aligned}
r&=\hat{b}^2 z_{\pl\pl}+ \hat{b}(z_{\pl\mi}+z_{\mi\pl})+z_{\mi\mi},\quad r'=(\hat{b}^2-1)^2z_{\mi\mi},\\
t&=\hat{b}^2z_{\mi\mi}+\hat{b}(z_{\pl\mi}+z_{\mi\pl})+z_{\pl\pl},\quad t'=(\hat{b}^2-1)^2z_{\pl \pl}.
\end{aligned}
\end{equation}
We first show that $r,t\neq 0$. Consider the graph $H=(V_H,E_H)$ obtained from $G$ by subdividing edge $e$, i.e., we remove edge $e=\{u,v\}$ and then add a new vertex $s$ which is connected to both $u,v$. Note that $H$ is obtained from $G\backslash e$ by adding the edges $\{s,u\}, \{s,v\}$, so it is not hard to see that
\[Z_{H}(\lambda,\hat{b})=\lambda t+r.\]
Note that $H$ is a graph of maximum degree $\Delta$ and we have $Z_{H}(\lambda,\hat{b})\neq 0$ from \eqref{eq:nonzero}. Moreover, from Lemma~\ref{lem:conjugation}, we have $r=\lambda^{n} \conj{t}$. Combining these, we obtain that $r,t\neq 0$. From assumption, we also have that $t'\neq 0$.

We will show how to compute the ratios $R_{\goal}=-\frac{r}{t}$ and $R_{\goal}'=-\frac{r'}{t'}$ (note that these are well-defined since $t,t'\neq 0$). By Lemma~\ref{lem:conjugation}, we have that $r=\lambda^{n} \conj{t}$ \mbox{ and } $r'=\lambda^{n} \conj{t'}$, so $R_{\goal},R_{\goal}'\in \SQ$. In fact, letting $p,p',p'',q$ be integers such that $\hat{b}=p/q$ and $\lambda=(p'+\im p'')/q$, then we have that $R_{\goal},R_{\goal}'\in \mathcal{R}\cap \SQ$, where 
\[\mathcal{R}=\Big\{\frac{P+\im Q}{P'+\im Q'}\mid\, P,Q, P',Q'\in \{-M,\hdots,0, \hdots, M\}\Big\}\mbox{ and }M:=2^{n} |p|^{m} (|p'|+|p''|)^{n}q^{m+n}.\]
Let $\epsilon=1/(10M)^{16}$. Note that for any two distinct numbers $z,z'\in \mathcal{R}$ it holds that $|z-z'|\geq 10\epsilon$, so if we manage to produce $\hat{R},\hat{R}'\in \SQ$ with $poly(n)$ size so that $|R_{\goal}-\hat{R}|\leq \epsilon$ and $|R_{\goal}'-\hat{R}'|\leq \epsilon$, we can in fact compute $R_{\goal}$ and $R_{\goal}'$ in time $poly(n,\size{\epsilon})=poly(n)$.\footnote{We give briefly the details for $R_{\goal}$, the details for $R_{\goal}'$ are similar. For $r \in \mathbb{N}$ let $\mathbb{Q}_r$ denote the set of rationals with denominator between $1$ and $r$. Since $R_{\goal}\in \mathcal{R}\cap\SQ$ and $\hat{R}\in \SQ$ we have that there exist  $\alpha,\beta\in \mathbb{Q}_{2M^2}$ and  $\hat{\alpha},\hat{\beta}\in \mathbb{Q}$ such that  $R_{\goal}=\alpha+\im \beta$ and $\hat{R}=\hat{\alpha}+\im \hat{\beta}$. From $|R_\goal-\hat{R}|\leq \epsilon$, we have $|\alpha-\hat{\alpha}|, |\beta-\hat{\beta}|\leq \epsilon$. By Lemma~\ref{lem:diophantine} (applied to $\hat{\alpha},\hat{\beta}$ and $K = 2M^2$), in poly(n) time, we can compute rationals  $\alpha',\beta'\in \mathbb{Q}_{2M^2}$ such that $|\hat{\alpha}-\alpha'|,|\hat{\beta}-\beta'|\leq 1/(8M^4)$ and hence $|{\alpha}-\alpha'|,|{\beta}-\beta'| \leq  \epsilon+1/(8M^4) \leq 1/(4M^4)$. Now, for distinct   $\gamma,\delta\in \mathbb{Q}_{2M^2}$ we have that $|\gamma-\delta|\geq 1/(2M^2)$, so it must be that $\alpha=\alpha'$ and $\beta=\beta'$, completing the  computation of $R_{\goal}$.} 

We first focus on how to compute $\hat{R}\in \SQ$ so that $|R_{\goal}-\hat{R}|\leq \epsilon$. At this point, it will be helpful to represent  complex numbers on the unit circle $\mathbb{S}$ with their arguments. Let $\theta_{\goal}=\Arg(R_{\goal})$ and $g(\theta):=t \emm^{\im \theta}+r$. Note that
\begin{equation}\label{eq:gnorm}
\begin{aligned}
|g(\theta)|&=|g(\theta)-g(\theta_{\goal})|=|t||\emm^{\im \theta}-\emm^{\im \theta_{\goal}}|=2|t|\big|
\sin((\theta-\theta_{\goal})/2)\big|,
\\
\Arg(g(\theta))&=(\theta-\theta_\goal)/2+\Arg(t) \mod 2\pi,
\end{aligned}
\end{equation}
the latter provided $\theta\neq \theta_\goal$.

We will compute in $poly(n)$ time a rational  $\hat{\theta}$ such that $|\hat{\theta}-\theta_{\goal}|\leq \epsilon/2$, yielding the desired $\hat{R}$ (via Lemma~\ref{lem:rationalangle}).

Let $\tau=1/500$ and $\kappa=\epsilon/10^3$. We will show that a poly-time algorithm  for $\FactorIsing{K}$ can be used to compute, for every rational  $\theta$, a positive number $\hat{g}_\theta$ in time $poly(n,\size{\theta})$  such that, whenever $|\theta-a|\geq \kappa$ for every $a\in \arg(R_{\goal})$, it holds that
\begin{equation}\label{eq:fptasgtheta}
(1-\tau)|g(\theta)|\leq \hat{g}_\theta\leq (1+\tau)|g(\theta)|.
\end{equation}
When $|\theta- a|\leq \kappa$ for some $a\in \arg(R_{\goal})$, there is no guarantee on the value of $\hat{g}_\theta$. Similarly, we will show that a poly-time algorithm  for $\ArgIsing{\rho}$ can be used to compute, for every rational  $\theta$, a positive number $\hat{a}_\theta$ in time $poly(n,\size{\theta})$  such that, whenever $|\theta-a|\geq \kappa$ for every $a\in \arg(R_{\goal})$, it holds that
\begin{equation}\label{eq:fptasgthetaarg}
|\Arg(g(\theta))-\hat{a}_\theta|\leq 2\rho=\pi/20.
\end{equation}
Using these, we compute the desired $\hat{\theta}$ via binary search following similar techniques as in \cite{GG,GJ,BGGS}, though in our case the details are a bit different because we have to work on the unit circle. For the norm, we will utilise that $|g(\theta)|$ is increasing in the interval $[\theta_{\goal},\theta_{\goal}+\pi]$ and decreasing in the interval $[\theta_\goal-\pi,\theta_\goal]$, whereas for the argument we will utilise that $\Arg(g(\theta))$ changes abruptly around $\theta_\goal$ (roughly by $\pi$). In particular, we proceed as follows.
\vskip 0.2cm
\noindent
{\bf Algorithm for $\FactorIsing{K}$ (Step 1):} We first find an interval of length $<2\pi/3$ with rational endpoints containing  $\theta_{\goal}$ in $poly(n)$ time. For  $j = 0, \ldots, 18$ let $\theta_{j} = j/3$,  $g_j = |g(\theta_j)|$ and $\hat{g}_j=\hat{g}_{\theta_j}$; note that the $\hat{g}_j$'s can be computed in $poly(n)$ time. For convenience, extend these definitions by setting  $\theta_{19h+j}=\theta_j+2h \pi$, $g_{19h+j} = g_j$  and $\hat{g}_{19h+j}=\hat{g}_j$ for every integer $h$ and $j = 0, \ldots, 18$. Note that for all $j$ we have that $1/3\geq  |\theta_{j+1}-\theta_{j}|\geq 1/4>\pi/15$.

Consider an index $j\in\{0,\hdots,18\}$ such that $\arg(R_{\goal})$ does not intersect with the intervals $[\theta_j-\kappa,\theta_{j+1}+\kappa]$ and $[\theta_j-\pi, \theta_{j+1}-\pi]$. Then, we have that 
\begin{equation}\label{eq:w1sww2333s}
(1-\tau)g_{j}\leq \hat{g}_{j}\leq (1+\tau)g_{j}, \quad (1-\tau)g_{j+1}\leq \hat{g}_{j+1}\leq (1+\tau)g_{j+1}.
\end{equation}
We claim that $g_{j+1} - g_j$ has the same sign as $\hat{g}_{j+1} - \hat{g}_j$. To see this, assume wlog $g_{j+1} - g_j > 0$, the other possibility follows in a similar way. Observe that we must have $\theta_j,\theta_{j+1}\in (\theta_{\goal},\theta_{\goal}+\pi)$, as $\theta\mapsto \sin(\theta/2-\theta_{\goal}/2)$ is increasing on $(\theta_{\goal},\theta_{\goal}+\pi)$  and so
\begin{equation}\label{eq:tb56566}
\begin{aligned}
g_{j+1} - g_j &= |g(\theta_{j+1})| - |g(\theta_j)|\geq 2|t|\min_{\phi \in [0, \pi/2 - \pi/30)]} [\sin(\phi + \pi / 30)-\sin\phi] \\
&\geq 2|t|\big[\sin(\pi/2)-\sin (\pi/2-\pi/30)\big] \geq |t|/100.
\end{aligned}
\end{equation}
On the other hand, if $\hat{g}_{j+1} - \hat{g}_j <0$, from \eqref{eq:w1sww2333s} we have $(1- \tau)g_{j+1} - (1+ \tau)g_j < 0$. This gives $g_{j+1} - g_j \leq \tau(g_{j+1} + g_j) \leq 2 \tau |t|$, a contradiction to the above.

Let $j^*$ be such that $\theta_{\goal} \in [\theta_{j^*}, \theta_{j^*+1})$. From \eqref{eq:gnorm}, the sequence $g_{j}$ is decreasing till $j^*$ and increasing after $j^*+1$. From the claim above, the sequence $\hat{g}_j$ must therefore be decreasing for indices $j$ in $[j^*-8,j^*-1]$ and increasing for indices $[j^*+2,j^*+9]$. 
Therefore, from the values of $\hat{g}_j$'s we can find $\hat{j}$ so that $\theta_{\goal}\in [\theta_{\hat{j}-3},\theta_{\hat{j}+3}]$. By enlarging slightly the interval $[\theta_{\hat{j}-3},\theta_{\hat{j}+3}]$, we obtain the desired interval of length $<2\pi/3$ with rational endpoints.

\vspace{0.2 cm}
\noindent
{\bf Algorithm for $\FactorIsing{K}$ (Step 2):} Given an interval $[\theta_1, \theta_2]$ with rational endpoints containing $\theta_{\goal}$ with $|\theta_1 - \theta_2|=\ell$ and $\ell\in (100\kappa,2\pi/3)$, we show how to find  in $poly(n,\size{\theta_1,\theta_2})$ time an interval with rational endpoints that is a factor of $1/2$ smaller in length and also contains $\theta_{\goal}$. The analysis will be similar to step 1.

For $j = 0, \ldots, 19$ define $\phi_j = \theta_1 + (\theta_2 - \theta_1)j /19$ and let $g_j = |g(\phi_j)|$ and $\hat{g}_j=\hat{g}_{\phi_j}$. Since $\theta_{\goal} \in [\theta_1, \theta_2]$ and $|\theta_1 - \theta_2|=\ell$, for any $\theta \in  [\theta_1, \theta_2]$ we have $|g(\theta)| \leq 2|t|\sin(\ell/2) \leq \ell |t|$. In particular we have $g_j \leq \ell |t|$ for all $j$.

Moreover, for an index $j$ such that $\theta_{\goal}\notin [\phi_j,\phi_{j+1}]$ we claim that $g_{j+1} - g_j$ has the same sign as $\hat{g}_{j+1} - \hat{g}_j$. To prove the claim assume $g_{j+1} - g_j \geq  0$, so $\phi_{j+1}, \phi_j \geq \theta_{\goal}$, the other possibility follows in a similar way. The derivative of $|g(\theta)|$ in the interval $[\theta_{\goal},\theta_{\goal}+\ell]$ is bounded below by $|t|\cos(\ell/2)\geq |t|/2$, so  by the mean value theorem we have that 
\begin{align*}
g_{j+1} - g_j \geq \frac{|t|}{2}(\phi_{j+1} - \phi_{j})\geq  |t|\ell/50.
\end{align*}
On the other hand if $\hat{g}_{j+1} - \hat{g}_j <0$ then as before we have $(1- \tau)g_{j+1} - (1+ \tau)g_j < 0$, which implies $g_{j+1} - g_j \leq \tau(g_{j+1} + g_j) \leq 2\tau\ell |t|$, a contradiction to above. This proves the claim.

Using the claim, we can conclude just as we did in step 1 and find an index $\hat{j}$ so that $\theta_{\goal}\in [\phi_{\hat{j}-3},\phi_{\hat{j}+3}]$, giving the desired interval.

\vspace{0.2 cm}
\noindent
{\bf Algorithm for $\ArgIsing{\rho}$:} Given a rational endpoint $\theta_1$ and a rational length $\ell\in (100\kappa,\tfrac{63}{10}]$ such that $\theta_{\goal}$ lies in the interval $[\theta_1, \theta_2]$ for some $\theta_2\leq \theta_1+\ell$,  we show how to find  in $poly(n,\size{\theta_1,\ell})$ time a  rational endpoint $\theta_1'$ and a rational length $\ell'$ such that $\ell'\leq \ell/4$ and $\theta_{\goal}\in [\theta_1', \theta_2']$ for some $\theta_2'\leq \theta_1'+\ell'$. 

For $j = 0, \ldots, 25$ define $\phi_j = \theta_1 + \ell j /26$ and let $a_j = \Arg(g(\phi_j))$, $\hat{a}_j=\hat{a}_{\phi_j}$. For convenience, extend these definitions by setting  $\phi_{26h+j}=\phi_j$, $a_{26h+j}= a_{j}$  and $\hat{a}_{26h+j}= \hat{a}_{j}$ for every integer $h$ and $j = 0, \hdots, 25$. For indices $j,j'$, let 
\[D_{j,j'}=\min\{|a_{j'} - a_j|,2\pi-|a_{j'} - a_j|\}\mbox{ and }\widehat{D}_{j,j'}=\min\{|\hat{a}_{j'} - \hat{a}_j|,2\pi-|\hat{a}_{j'} - \hat{a}_j|\}.\]
Consider an index $j$ such that $\theta_{\goal}\notin [\phi_j-\kappa,\phi_{j+1}+\kappa]$. Then, we have that  $D_{j,j+1}=|\phi_{j+1}-\phi_j|/2\leq \pi/10$ and hence $\widehat{D}_{j,j+1}\leq \pi/5$. On the other hand, for an index $j$ such that $\theta_{\goal}\in [\phi_{j},\phi_{j+1}]$ we have that $D_{j-1,j+1}=\pi-|\phi_{j+1}-\phi_{j-1}|/2\geq 4\pi/5$ and similarly $D_{j,j+2}\geq 4\pi/5$. Therefore, at least one of $\widehat{D}_{j-1,j+1}\geq 3\pi/5$, $\widehat{D}_{j,j+2}\geq 3\pi/5$ must hold.
Therefore, using the $\hat{a}_j$'s, we can find an index $\hat{j}$ so that $\theta_{\goal}\in [\phi_{\hat{j}-2},\phi_{\hat{j}+2}]$, giving the desired interval.

\vskip 0.2cm

By repeating the above, we conclude that, using a poly-time  algorithm for either the problem $\FactorIsing{K}$ or $\ArgIsing{\rho}$, we can compute in $poly(n)$ time a rational  $\hat{\theta}$ such that $|\hat{\theta}-\theta|\leq 400\kappa\leq \epsilon/2$, yielding the desired $\hat{R}$ (via Lemma~\ref{lem:rationalangle}). We thus focus on proving that for a rational  $\theta$ we can obtain in time $poly(n,\size{\theta})$ values $\hat{g}_\theta, \hat{a}_\theta$   satisfying \eqref{eq:fptasgtheta} and \eqref{eq:fptasgthetaarg}, respectively.

Let $\epsilon_2=\kappa \epsilon/10^{5}$, $\epsilon_1:=\epsilon_2/\big(2^{4n}(2\hat{b})^{2m}\big)$, $\epsilon_0=\epsilon_1/(k4^k)$. By Lemmas~\ref{lem:mainlemma} and~\ref{lem:rationalangle}, for a rational number $\phi$, we can construct in time $poly(n,\size{\phi})$ a rooted tree $T_\phi$ in $\Tc_{\Delta}$  with root $x_\phi$ that has degree 1 and implements a field $\lambda_\phi$ such that $|\lambda_\phi-\emm^{\im \phi}|\leq \epsilon_0$.  For convenience, let 
\begin{equation}\label{eq:Qphi}
Q_\phi^{\plm}:=Z_{T_\phi,\plm x_\phi}(\lambda,b) \mbox{ and note that } \Big|\frac{Q_\phi^{\pl}}{Q_\phi^{\mi}}-\emm^{\im \phi}\Big|\leq \epsilon_0.
\end{equation}
Let $T_\theta,T_{0}$ be the trees obtained for $\phi=\theta,0$ and note that that $T_\theta, T_0$ implement the vertex activities $\emm^{\im \theta}, 1$ respectively (with precision $\epsilon_0$).

Recall that $P_k$ is the path with $k$ vertices and endpoints $u_k,v_k$, we denote by $V_{P_k}$ the set of its vertices. Let $P_{k, T_{0}}$ be the tree obtained from $P_k$  by attaching $k-2$ disjoint copies of the graph $T_{0}$ to the internal vertices of the path, i.e., for $i=1,\hdots, k-2$, identify the root $x_{0}$ of the $i$-th copy of $T_{0}$ with the $i$-th internal vertex of the path. For convenience, let
\begin{equation}\label{eq:4g45g644abc12}
A_{\plm\plm}:=Z_{P_{k, T_{0}},\plm u_k,\plm v_k}(\lambda,b).
\end{equation}

Recall that $H=(V_H,E_H)$ denotes the graph obtained by subdividing edge $e$ of $G$. Let $H_\theta\in \mathcal{G}_\Delta$ be the graph obtained from $H$ by replacing every edge $\{x,y\}$ of $H$ by a distinct copy of $P_{k, T_{0}}$ (identifying $x$ with $u_k$ and $y$ with $v_k$) and attaching the tree $T_\theta$ on the vertex $s$ of $H$ (identifying $s$ with the root $x_\theta$).  Effectively, the construction of $H_{\theta}$ is so that the Ising model on $H_\theta$ with edge activities equal to $b$ and vertex activities equal to $\lambda$  corresponds to an Ising model on $H$ with edge activities equal to $\hat{b}$, and vertex activities equal to $\lambda$ apart from that of vertex $s$ which is set to $\emm^{\im \theta}$. In this latter model, the contribution to the partition function from configurations where $s$ is set to $\pl$ is given by $t$ and the contribution to the partition function from configurations where $s$ is set to $\mi$ is given by $r$, where $t,r$ are as in \eqref{eq:rtrptp}. Based on this, we will soon show that
\begin{equation}\label{eq:5655644ff4f34}
\Big|\frac{Z_{H_\theta}(\lambda,b)}{Q^{\mi}_\theta(A_{\pl\pl})^{m+1}}-g(\theta)\Big|\leq  \epsilon_2.
\end{equation}
From \eqref{eq:5655644ff4f34}, we obtain the desired approximations $\hat{g}_{\theta}, \hat{a}_\theta$ that satisfy \eqref{eq:fptasgtheta}, \eqref{eq:fptasgthetaarg} respectively,  as follows. 
First, observe that $|g(\theta)|\geq |t|\kappa/2\geq 10\epsilon_2/\tau$ since $|\theta-a|\geq \kappa$ for every $a\in \arg(R_{\goal})$. 
Second, $T_\theta$ and $P_{k, T_{0}}$ are trees of size $poly(n, \size{\theta})$, so we can compute $Q^{\mi}_\theta$ and $A_{\pl\pl}$   in time $poly(n, \size{\theta})$. Using a poly-time algorithm for $\FactorIsing{K}$, we can compute $\hat{Z}_{\theta}$ in time $poly(n,\size{\theta})$ which is within a factor of $1\pm \tau$ from $|Z_{H_\theta}(\lambda,b)|$,  thus yielding $\hat{g}_{\theta}=\frac{\hat{Z}_{\theta}}{|Q^{\mi}_\theta|\,|A_{\pl\pl}|^{m+1}}$ that satisfies \eqref{eq:fptasgtheta}. 
Similarly, using a poly-time algorithm for $\ArgIsing{\rho}$, we can compute $\hat{A}_{\theta}$ in time $poly(n,\size{\theta})$ which is within distance $\rho$ from $\mathrm{Arg}(Z_{H_\theta}(\lambda,b))$. 
Noting that the argument $\alpha$ of $\frac{Z_{H_\theta}(\lambda,b)}{Q^{\mi}_\theta(A_{\pl\pl})^{m+1}}-g(\theta)$ satisfies $\sin(\alpha)\leq \epsilon_2/g(\theta)$, from which it follows that $\alpha\leq \rho$. 
Hence $\hat{a}_{\theta}=\hat{A}_{\theta}-\Arg(Q^{\mi}_\theta)-(m+1)\Arg(A_{\pl\pl})$ $(\mathrm{mod}\ 2\pi)$ satisfies \eqref{eq:fptasgthetaarg}.  

It remains to prove \eqref{eq:5655644ff4f34}. We first claim that
\begin{equation}\label{eq:4g45g644abc}
\Big|\frac{A_{\plm\plm}}{(Q^{\mi}_{0})^{k-2}}-Z_{P_{k},\plm u_k,\plm v_k}(1,b)\Big|\leq \frac{\epsilon_1}{4}Z_{P_{k},\plm u_k,\plm v_k}(1,b).
\end{equation}
Indeed, for a fixed $\sigma:V_{P_k}\rightarrow \{\pl,\mi\}$, the aggregate contribution to $Z_{P_{k,T_{0}}}(1,b)$ from configurations on $P_{k,T_{0}}$ that agree with $\sigma$ on $V_{P_k}$ is $(Q^{\pl}_{0})^{n_{\pl}(\sigma)}(Q^{\mi}_{0})^{n_{\mi}(\sigma)}w_{P_k,\sigma}(1,b)$ where $n_{\plm}(\sigma)$ is the number of internal vertices in $P_k$ that have spin $\plm$ under $\sigma$, so \eqref{eq:4g45g644abc} follows from aggregating over the relevant $\sigma$ and observing that\footnote{Here, and in the follow-up estimates, we use that for complex numbers $c_1,\hdots,c_i$ and $d_1,\hdots,d_i$ it holds that $\big|\prod^{i}_{j=1}c_j-\prod^{i}_{j=1}d_j\big|\leq \sum^{i}_{j=1}|c_j-d_j|\prod^{j-1}_{j'=1}|c_j|\prod^{i}_{j'=j+1}|d_j|$.} $\big|\frac{(Q^{\pl}_{0})^{j}}{(Q^{\mi}_{0})^{j}}-1\big|\leq k\epsilon_0$ for all $j=0,\hdots,k$. 
From \eqref{eq:g4g4635335r} and \eqref{eq:4g45g644abc}, it follows that $A_{\plm,\plm}\neq 0$ and 
\begin{equation}\label{eq:4g6g6677712}
\Big|\frac{A_{\mi\pl}}{A_{\pl\pl}}-\hat{b}\Big|\leq \epsilon_1,\quad\Big|\frac{A_{\mi\mi}}{A_{\pl\pl}}-1\Big|\leq \epsilon_1, 
\end{equation}
Now, for $\sigma: V_H\rightarrow\{\pl,\mi\}$ with $\sigma(s)=\pl$, let $W^{\pl}_\sigma$ be the aggregate weight of configurations on $H_\theta$ that agree with $\sigma$ on  $V(H)$. Define analogously $W^{\mi}_\sigma$. Then, we have that
\[W^{\plm}_\sigma=Q^{\plm}_\theta(A_{\pl\pl})^{m_{\pl\pl}(\sigma)}(A_{\pl\mi})^{m_{\pl\mi}(\sigma)}(A_{\mi\mi})^{m_{\mi\mi}(\sigma)},\]
where $m_{\pl\pl}(\sigma),m_{\pl\mi}(\sigma),m_{\mi\mi}(\sigma)$ denote the number edges of $E_H$ whose endpoints are assigned $\pl\pl,\pl\mi,\mi\mi$, respectively. Since the total number of edges in $E_H$ is $m+1$, we obtain
\begin{equation}\label{eq:4g5rf4556}
\bigg|\frac{W^{\pl}_\sigma}{Q^{\mi}_\theta(A_{\pl\pl})^{m+1}}-\emm^{\im \theta}\, w_{H,\sigma}(\lambda,\hat{b})\bigg|\leq \epsilon_2/10^n, \quad \bigg|\frac{W^{\mi}_\sigma}{Q^{\mi}_\theta(A_{\pl\pl})^{m+1}}- w_{H,\sigma}(\lambda,\hat{b})\bigg|\leq \epsilon_2/10^n.
\end{equation}
Observe also that the quantities $t,r$, as defined in \eqref{eq:rtrptp}, are such that
\[t=\sum_{\sigma: V_H\rightarrow \{\pl,\mi\}; \sigma(s)=\pl}w_{H,\sigma}(\lambda,\hat{b}) \mbox{ and } r=\sum_{\sigma: V_H\rightarrow \{\pl,\mi\}; \sigma(s)=\mi}w_{H,\sigma}(\lambda,\hat{b}),\]
so summing \eqref{eq:4g5rf4556} over all $\sigma$ gives \eqref{eq:5655644ff4f34}. This finishes the proof of \eqref{eq:5655644ff4f34} and hence completes the computation of $R_{\goal}$ in $poly(n)$ time.

The computation of $R_{\goal}'$ is completely analogous, once we establish an analogue of \eqref{eq:5655644ff4f34}. In particular, let $H'$ be the graph obtained from $H$ by removing vertex $s$ and adding the vertices $u',v',s'$ and the edges $\{u,u'\},\{u',s'\},\{s',v'\},\{v',v\}$; note that $H'$ is obtained from $G$ by replacing the edge $e$ by a path with three vertices. We construct $H_\theta'$ from $H'$ as above, with a minor twist: we replace every edge $\{x,y\}$ of $H'$ with a distinct copy of $P_{k, T_{0}}$ (identifying $x$ with $u_k$ and $y$ with $v_k$), we attach the rooted tree $T_\theta$ on the vertex $s'$ of $H'$ (identifying $s'$ with the root $x_\theta$), and we attach two distinct copies of the rooted tree $T_\pi$ on the vertices $u',v'$ of $H'$ (identifying $u',v'$ with the corresponding roots $x_\pi$ in the two copies of $T_\pi$). Note the use of the tree $T_\pi$ in the construction of $H'$ which, analogously\footnote{Even though $\pi$ is irrational, it holds that $\emm^{\im \pi}=-1$ and we can therefore construct $T_\pi$ satisfying \eqref{eq:Qphi} for $\phi=\pi$ using Lemma~\ref{lem:mainlemma}.} to \eqref{eq:Qphi},  implements the field $\emm^{\im \pi}=-1$ (with precision $\epsilon_0$).  Effectively, the construction of $H_{\theta}'$ is so that the Ising model on $H_\theta'$ with edge activities equal to $b$ and vertex activities equal to $\lambda$ corresponds to an Ising model on $H'$ with edge activities equal to $\hat{b}$, and vertex activities equal to $\lambda$ apart from those of $u',s',v'$ which are set to $-1,\emm^{\im \theta},-1$, respectively. In this latter model, the contribution to the partition function from configurations where $s'$ is set to $\pl$ is given by $t'=(\hat{b}^2-1)^2z_{\pl\pl}$ and the contribution to the partition function from configurations where $s'$ is set to $\mi$ is given by $r'=(\hat{b}^2-1)^2z_{\mi\mi}$. Based on this, we obtain similarly to above, the following analogue of \eqref{eq:5655644ff4f34}:
\begin{equation}\label{eq:5655644ff4f34b}
\Big|\frac{Z_{H_\theta'}(\lambda,b)}{Q^{\mi}_\theta (Q^{\mi}_\pi)^2(A_{\pl\pl})^{m+2}}-\big(t'\emm^{\im \theta}+r'\big)\Big|\leq  \epsilon_2.
\end{equation}
Having \eqref{eq:5655644ff4f34b} at hand, the computation of $R_{\goal}'$ can be carried out using exactly the same procedure as for $R_{\goal}$. This finishes the proof of Lemma~\ref{lem:maincomp}.
\end{proof}

\subsection{Proof of our main theorem}\label{sec:proofmain}
We are now ready to finish the proof of Theorem~\ref{thm:main2}, which we restate here for convenience.
\begin{thmmaintwo}
\statethmmaintwo
\end{thmmaintwo}
\begin{proof}[Proof of Theorem~\ref{thm:main2}]
Let $b\in (0,1)$ be a rational number and let $\ \lambda\in \SQ(\Delta-1,b)$.
By Theorem~\ref{thm: Density of Ratios} it suffices to show that $\FactorIsing{K}$ and $\ArgIsing{\rho}$ are $\numP$-hard.
To prove the $\numP$-hardness for these problems, we will show that, assuming a poly-time algorithm for either $\FactorIsing{K}$ or $\ArgIsing{\rho}$, on input of a graph $G\in\mathcal{G}_3$ we can compute $Z_{G}(\lambda,\hat{b})$ in poly-time, which is $\numP$-hard by \cite[Theorem 1.1]{KowCai16}.
In fact, it suffices to compute in poly-time, for an arbitrary edge $e$ of $G$, the ratio $\frac{Z_{G}(\lambda,\hat{b})}{Z_{G\backslash e}(\lambda,\hat{b})}$ since then we can compute $Z_{G}(\lambda,\hat{b})$ using a telescoping product over the edges of the graph $G$.

So fix an arbitrary edge $e=\{u,v\}$ of $G$ and let $z_{\plm\plm}:=Z_{G,\plm u,\plm v}(\lambda,\hat{b})$. The ratio $r^*:=\frac{Z_{G}(\lambda,\hat{b})}{Z_{G\backslash e}(\lambda,\hat{b})}$ is well-defined since, by the choice of  $\hat{b}$, we have $Z_{G\backslash e}(\lambda,\hat{b})\neq 0$ (cf. \eqref{eq:2wsx} and \eqref{eq:nonzero}). Moreover, we can express $r^*$ using the $z_{\plm\plm}$'s as follows:
\[r^*=\frac{z_{\pl\pl}+z_{\mi\mi}+\hat{b}(z_{\pl\mi}+z_{\mi\pl})}{z_{\pl\pl}+z_{\mi\mi}+z_{\pl\mi}+z_{\mi\pl}}.\]

We will compute $r^*$ using Lemma~\ref{lem:maincomp}. Namely, by Lemma~\ref{lem:maincomp}, we can compute in poly-time the value of the ratio
\begin{equation}\label{eq:r134223}
r=R_{G,e}=\frac{A^2z_{\pl\pl}+ AB(z_{\pl\mi}+z_{\mi\pl})+B^2z_{\mi\mi}}{A^2z_{\mi\mi}+AC(z_{\pl\mi}+z_{\mi\pl})+C^2z_{\pl\pl}}, \mbox{\, where \, }\begin{array}{l}A:=\hat{b}\\B:=1\\ C:=1\end{array}.
\end{equation}
Let $G'$ be the graph obtained from $G\setminus e$ by  adding two new vertices $u',v'$ and adding the edges $\{u,u'\},\{u',v'\},\{v',v\}$.  We next apply Lemma~\ref{lem:maincomp} to the graph $G'$ with the edge $e'=\{u',v'\}$. We first express $Z_{G'\backslash e',\plm u',\plm v'}(\lambda,\hat{b})$ in terms of the $z_{\plm\plm}$'s. We have 
\begin{align*}
Z_{G'\backslash e',\pl u',\pl v'}(\lambda,\hat{b})&=\lambda^2 \big(z_{\pl\pl}+\hat{b} (z_{\pl\mi}+z_{\mi\pl})+\hat{b}^2 z_{\mi\mi}\big),\\
Z_{G'\backslash e',\pl u',\mi v'}(\lambda,\hat{b})&= \lambda\big(\hat{b} z_{\pl\pl}+z_{\pl\mi}+\hat{b}^2 z_{\mi\pl}+\hat{b} z_{\mi\mi}\big),\\
Z_{G'\backslash e',\mi u',\pl v'}(\lambda,\hat{b})&= \lambda\big(\hat{b} z_{\pl\pl}+\hat{b}^2z_{\pl\mi}+ z_{\mi\pl}+\hat{b} z_{\mi\mi}\big),\\
Z_{G'\backslash e',\mi u',\mi v'}(\lambda,\hat{b})&= \hat{b}^2 z_{\pl\pl}+\hat{b} (z_{\pl\mi}+z_{\mi\pl})+ z_{\mi\mi}.
\end{align*} 
Then, by Lemma~\ref{lem:maincomp}, we can compute in poly-time the value of the ratio
\begin{equation}\label{eq:r134223b}
r'=R_{G',e'}=\frac{(A')^2z_{\pl\pl}+ A'B'(z_{\pl\mi}+z_{\mi\pl})+(B')^2z_{\mi\mi}}{(A')^2z_{\mi\mi}+A'C'(z_{\pl\mi}+z_{\mi\pl})+(C')^2z_{\pl\pl}}, \mbox{\, where }\begin{array}{l}A':=\hat{b}(\lambda+1)\\B':=1+\hat{b}^2\lambda\\ C':=\hat{b}^2+\lambda\end{array}.
\end{equation}

We are now in position to complete the computation of $r^*$. We first show how to decide in poly-time  whether $z_{\pl \pl}=0$. We claim that 
\begin{equation}\label{eq:criterion1}
z_{\pl \pl}=0 \Longleftrightarrow r=B/C\mbox{ and }  r'=B'/C'.
\end{equation}
Indeed, if $z_{\pl \pl}=0$, then $z_{\mi\mi}=0$ from Lemma~\ref{lem:conjugation}, and therefore from \eqref{eq:r134223}, \eqref{eq:r134223b} we have that
$r=B/C$ and $r'=B'/C'$. Conversely, using that $A^2\neq B C$ and $(A')^2\neq B' C'$, we have that
\[r=B/C \Longrightarrow Cz_{\pl\pl}=Bz_{\mi\mi},\qquad r'=B'/C' \Longrightarrow C'z_{\pl \pl}=B'z_{\mi \mi},\]
which together imply that $z_{\pl \pl}=0$. 

Using \eqref{eq:criterion1} we can decide in poly-time whether $z_{\pl \pl}=0$. If so, by Lemma~\ref{lem:conjugation}, we have $z_{\mi\mi}=0$ and hence $r^*=\hat{b}$.  So, assume $z_{\pl \pl}\neq 0$, and hence $z_{\mi\mi}\neq 0$ in what follows. We claim that 
\begin{equation}\label{eq:criterion2}
z_{\pl\mi}+z_{\mi\pl}=0 \Longleftrightarrow r=\frac{A^2z_{\pl\pl}+B^2z_{\mi\mi}}{A^2z_{\mi\mi}+C^2z_{\pl\pl}}, r'=\frac{(A')^2z_{\pl\pl}+(B')^2z_{\mi\mi}}{(A')^2z_{\mi\mi}+(C')^2z_{\pl\pl}}.
\end{equation}
The forward direction is again trivial. For the backward direction, we have
\begin{align*}
r&=\frac{A^2z_{\pl\pl}+B^2z_{\mi\mi}}{A^2z_{\mi\mi}+C^2z_{\pl\pl}} \Longrightarrow C z_{\pl\pl}=Bz_{\mi\mi} \mbox{ or }z_{\pl\mi}+z_{\mi\pl}=0,\\ 
r'&=\frac{(A')^2z_{\pl\pl}+(B')^2z_{\mi\mi}}{(A')^2z_{\mi\mi}+(C')^2z_{\pl\pl}} \Longrightarrow C'z_{\pl \pl}=B'z_{\mi \mi} \mbox{ or }z_{\pl\mi}+z_{\mi\pl}=0.
\end{align*}
Since $z_{\pl \pl},z_{\mi\mi}\neq 0$, we therefore obtain that $z_{\pl\mi}+z_{\mi\pl}=0$, proving \eqref{eq:criterion2}. 

Note, we can decide the right-hand of \eqref{eq:criterion2} in poly-time using the value of the ratio $r''=z_{\mi\mi}/z_{\pl\pl}$ from the second part of Lemma~\ref{lem:maincomp}. If it turns out that $z_{\pl\mi}+z_{\mi\pl}=0$, then $r^*=1$ and we are done.  Otherwise, we can use the values of $r$ and $r''$ to compute the ratios  $\frac{z_{\pl \pl}}{z_{\pl\mi}+z_{\mi\pl}}, \frac{z_{\mi\mi}}{z_{\pl\mi}+z_{\mi\pl}}$, which we can then use to compute $r^*$. 

This completes the computation of the ratio $r^*$, and therefore the proof of Theorem~\ref{thm:main2}.
\end{proof}

\section{Equivalence for $\lambda=-1$ with Approximately Counting Perfect Matchings}\label{sec:minusone}
In this section, we show that for $\lambda=-1$, the problem of approximating the partition of the ferromagnetic Ising model on graphs of maximum degree $\Delta$ is equivalent to the problem $\PM$, the problem of approximately counting perfect matchings  on general graphs. The proof follows the technique in \cite{Tutte08}, where the case of negative $b$ but $\lambda=1$ was considered; here however, we need to rework the relevant ingredients. 
The main such ingredient is the following ``high-temperature'' expansion formula for~$\lambda=-1$.
\begin{lemma}\label{lem:connection}
Let $\lambda=-1$ and $b\neq -1$ be an arbitrary number.  Then, for any graph $G=(V,E)$,
\[Z_G(\lambda,b)=(-2)^{|V|}\Big(\frac{1+b}{2}\Big)^{|E|}\sum_{S\subseteq E;\, S\, \mbox{{\small  odd } }} \Big(\frac{1-b}{1+b}\Big)^{|S|},\]
where the sum is over $S\subseteq E$ such that every vertex $v\in V$ has odd degree in the subgraph $(V,S)$.
\end{lemma}
\begin{proof}
Let $G=(V,E)$ be a graph. For a set $S\subseteq E$ and a vertex $v\in V$, we let $d_v(S)$ denote the degree of $v$ in the subgraph $(V,S)$. 

For the purposes of this proof, it will be convenient to view configurations of the Ising model on $G$ as vectors in $\{\pm1\}^{V}$. Now, for a configuration $\sigma\in \{\pm 1\}^{V}$ we use the notation $n_+(\sigma)$ to denote the number of vertices with spin $+1$.   Observe that $n_+(\sigma)=\tfrac{1}{2}\big(|V|+\sum_{v\in V}\sigma_v\big)$ and that for an edge $e=(u,v)$, we have $b^{\mathbf{1}\{\sigma_u\neq \sigma_v\}}=\tfrac{1+b}{2}\big(1+\tfrac{1-b}{1+b}\sigma_u\sigma_v\big)=\rho\big(1+\nu\sigma_u\sigma_v\big)$, where for convenience we set  $\rho:=\tfrac{1+b}{2}$ and $\nu:=\tfrac{1-b}{1+b}$. So, using that $\imm^2=-1$,
\begin{align*}
Z_G(\lambda,b)&=\rho^{|E|}\sum_{\sigma\in\{\pm 1\}^{V}}\lambda^{n_+(\sigma)}\prod_{e=(u,v)\in E}(1+\nu\sigma_u\sigma_v)=\rho^{|E|}\sum_{\sigma\rightarrow\{\pm 1\}^{V}}\lambda^{n_+(\sigma)}\sum_{S\subseteq E}\nu^{|S|}\prod_{v\in V}(\sigma_v)^{d_v(S)}\\
&=\imm^{|V|}\rho^{|E|}\sum_{S\subseteq E}\nu^{|S|}\sum_{\sigma\in\{\pm 1\}^{V}}\prod_{v\in V}\imm^{\sigma_v}(\sigma_v)^{d_v(S)}.
\end{align*}
The latter sum is equal to $\prod_{v\in V}\sum_{\sigma_v\in \{\pm1\}}\imm^{\sigma_v}(\sigma_v)^{d_v(S)}$, which equals $0$ if $d_v(S)$ is even, and $2\imm$ otherwise. Plugging this in the expression above, yields the lemma.
\end{proof}

Now, we are ready to show the main theorem for this section. For counting problems $A,B$ we use the notion of $\mathsf{AP}$-reductions, see \cite{DGGJ}. Roughly, we have that $A\leq_{\mathsf{AP}} B$ if an $\mathsf{FPRAS}$ for $B$ can be converted to an $\mathsf{FPRAS}$ for $A$, and $A\equiv_{\mathsf{AP}} B$ if both $A\leq_{\mathsf{AP}} B$ and $B\leq_{\mathsf{AP}} A$ hold.

\begin{theorem}\label{thm:lambdaminusone}
Let $\lambda=-1$ and $b\in (0,1)$ be a rational. Then, for any connected graph $G$, we have $Z_G(\lambda,b)>0$ if $G$ has an even number of vertices and $Z_G(\lambda,b)=0$, otherwise.  

Moreover, for all integers $\Delta\geq 3$, we have that $\FactorIsingb\equiv_{\mathsf{AP}} \PM$.
\end{theorem}
\begin{proof}
The statement about the sign of $Z_G(\lambda,b)$ follows from Lemma~\ref{lem:connection}, and the fact that every connected graph with an even number of vertices has a spanning subgraph where every vertex has odd degree. We thus focus on proving the $\mathsf{AP}$-equivalence.\vskip 0.1cm

\noindent $\boldsymbol{\PM\leq_{\mathsf{AP}} \FactorIsingb}$. It is well-known that the problem of approximating the number of perfect matchings on general graphs is $\mathsf{AP}$-equivalent to the same problem on graphs of maximum degree 3, see, e.g., \cite[Lemma 28]{GLLZ}. So, let $G=(V_G,E_G)$ be a graph of maximum degree 3, with $n=|V_G|$ and $m=|E_G|$, and let $\mathcal{M}$ be the set of perfect matchings of $G$.  Since we can check whether a graph has a perfect matching in polynomial time, we may further assume that $|\mathcal{M}|>0$ and in particular that $n$ is even. Let $\epsilon\in (0,1)$ be the desired relative error that we want to approximate $|\mathcal{M}|$.

Analogously to \eqref{eq:r4f4123} and \eqref{eq:g4g4635335r}, for $k=1+2\lceil \tfrac{m^2+\ln(1/\epsilon)}{-\ln(1-b)}\rceil$,  let $P_k=(V_k,E_k)$ be the path with $k$ vertices whose endpoints are labeled $u_k,v_k$ and $P_k^*=(V_k^*,E_k^*)$ be the graph obtained from $P_k$ by attaching a vertex $z_i$ to the $i$-th internal vertex $w_i$ of $P_k$, for $i=1,\hdots,k-2$. Let $A_{\plm,\plm}:=Z_{P_k^*,\plm u_k,\plm v_k}(\lambda,b)$. Then,  it is not hard  to see that\footnote{Here, the key observation is that for a configuration $\tau:V_k\rightarrow \{\pl,\mi\}$, the aggregate weight of configurations $\sigma:V_k^*\rightarrow \{\pl,\mi\}$ with $\sigma_{V_k}=\tau$ is $(-1)^{\mathbf{1}\{\tau_{u_k}\neq \tau_{v_k}\}}(1-b)^{k-2}w_{P_k,\tau}(1,b)$. Indeed, if $\tau(w_i)=\pl$ then the contribution of the edge $(w_i,z_i)$ and the external field on $z_i$ is $b+\lambda=b-1$, whereas if $\tau(w_i)=\mi$ the contribution is $1+b\lambda=1-b$. This, combined with the factor $\lambda^{n_\pl(\tau)}$ coming from the external fields on $V_k$, gives the factor $(-1)^{\mathbf{1}\{\tau_{u_k}\neq \tau_{v_k}\}}(1-b)^{k-2}$ above; the remaining contribution is just the weight of $\tau$ on $P_k$ when the external field of all vertices on $P_k$ is equal to $1$.}
\begin{equation*}
\left[\begin{array}{cc} A_{\pl\pl}& -A_{\pl\mi} \\ -A_{\mi\pl}&A_{\mi\mi}\end{array}\right]=(1-b)^{k-2}\left[\begin{array}{cc} 1& b \\ b& 1\end{array}\right]^{k-1}
\end{equation*}
and so 
\begin{align*}
&A_{\pl\pl}= A_{\mi\mi}=\tfrac{1}{2}\big((1+b)^{k-1}+(1-b)^{k-1}\big)(1-b)^{k-2}, \text{ and }
\\
&A_{\pl\mi}= A_{\mi\pl}=\tfrac{1}{2}\big((1-b)^{k-1}-(1+b)^{k-1}\big)(1-b)^{k-2}.
\end{align*}
We next set
\begin{equation}\label{eq:g4g4635335rb}
b_k:= -\frac{A_{\pl\mi}}{A_{\pl\pl}}=-\frac{A_{\mi\pl}}{A_{\mi\mi}},\quad \mbox{and observe that}\quad  1-(1-b)^{k-1}<b_k<1.
\end{equation} 

Let $H=(V_H,E_H)$ be an instance of $\FactorIsingb$ obtained from $G$ by replacing every edge $e=(u,v)$ of $G$ with a distinct copy of $P_k^*$, identifying the endpoints $u,v$ with $u_k,v_k$, respectively. Then, we claim that
\begin{equation}\label{eq:GHAbl}
Z_H(\lambda,b)=(A_{\pl \pl})^{m}Z_G(\lambda,b_k).
\end{equation}
Indeed, for a configuration $\sigma:V_G\rightarrow \{\pl,\mi\}$, let $\Omega_{H,\sigma}=\{\sigma': V_H\rightarrow \{\pl,\mi\}\mid \sigma_{V_G}'=\sigma\}$ be the configurations on $H$ which agree with $\sigma$ on $V_G$, and $Z_{H,\sigma}(\lambda,b)$ be the contribution to $Z_{H}(\lambda,b)$ from configurations in $\Omega_{H,\sigma}$.   Then,  we have
\[Z_{H,\sigma}(\lambda,b)=\lambda^{|n_\pl(\sigma)|}\prod_{e=(u,v)\in E_G}(-1)^{\mathbf{1}_{\sigma_u\neq \sigma_v}}A_{\sigma_u\sigma_v}=(A_{\pl \pl})^{m} \lambda^{|n_\pl(\sigma)|} b_k^{\delta(\sigma)},\]
proving \eqref{eq:GHAbl}.  Note, from Lemma~\ref{lem:connection} we have that
\begin{equation}\label{eq:t65g6hh223}
Z_G(\lambda,b_k)=2^{n}\Big(\frac{1+b_k}{2}\Big)^{m}\sum_{S\subseteq E;\, S\, \mbox{{\small  odd } }} \Big(\frac{1-b_k}{1+b_k}\Big)^{|S|}.
\end{equation}
Perfect matchings in $G$ are in 1-1 correspondence with odd sets $S\subseteq E$ with $|S|= n/2$. Moreover, for any other odd set $S\subseteq E$ we have $|S|>n/2+1$, and hence, using also \eqref{eq:GHAbl}, we obtain
\[\bigg|\frac{Z_H(\lambda,b)}{(A_{\pl \pl})^{m}2^{n}\big(\frac{1+b_k}{2}\big)^{m}\big(\frac{1-b_k}{1+b_k}\big)^{n/2}}-|\mathcal{M}|\bigg|\leq 2^m\Big(\frac{1-b_k}{1+b_k}\Big)\leq \epsilon |\mathcal{M}|.\]
Using therefore an FPRAS for $\FactorIsingb$, we can approximate $Z_H(\lambda,b)$ within relative error $\epsilon$ in time $poly(n,1/\epsilon)$, and compute therefore $|\mathcal{M}|$ within relative error $\epsilon$, finishing the $\mathsf{AP}$-reduction.

\vskip 0.2cm
\noindent $\boldsymbol{\FactorIsingb\leq_{\mathsf{AP}}\PM}$. We first consider the case $\Delta=3$. Let $G=(V,E)$ be a graph of maximum degree $\Delta=3$ that is input to $\FactorIsingb$, and set $n=|V|$, $m=|E|$. We may assume that $n$ is even, since otherwise we can output 0 for the partition function.  By Lemma~\ref{lem:connection} we have that
\begin{equation}\label{eq:t65g6hh223b}
Z_G(\lambda,b)=2^{n}\Big(\frac{1+b}{2}\Big)^{m}\sum_{S\subseteq E;\, S\, \mbox{{\small  odd } }} \Big(\frac{1-b}{1+b}\Big)^{|S|}.
\end{equation}

To formulate this in terms of perfect matchings, we construct a graph $G'=(V',E')$ as follows, resembling the construction in~\cite{Fisher}. For $v\in V$, let $d_v$ be the degree of $v$ in $G$. For a vertex $v\in V $, if $d_v=3$, replace $v$ with a triangle of  vertices $T_v=\{v_1,v_2,v_3\}$; otherwise, keep $v$  in $G'$ as well and let for convenience $T_v=\{v\}$. For every edge $(u,v)\in E$,  add an edge in $G'$ between a vertex in $T_u$ and $T_w$ so that $G'$ has maximum degree~3; note that edges of $G$ that are not incident to degree-3 vertices belong to $G'$ as well. We call internal all  edges of $G'$ whose endpoints belong to some $T_v$ and external all other edges of $G'$.  Note that an edge $e$ of $G$ maps to an external edge $\mathrm{ex}(e)$ of $G'$ bijectively, under the natural mapping. We use $\mathrm{ex}(G')$ to denote  the external edges of $G'$.

For $v\in V$, observe that any perfect matching in $G'$ must contain exactly one external edge incident to a  vertex in  $T_v$ if $|T_v|=1$, and two  or three edges if $|T_v|=3$, either one internal and one external, or three external, respectively. Based on this, we have that  a perfect matching $M'$ in $G'$ maps bijectively to an odd subset  $S$ of $G$, by adding an edge $e$ of $G$ to $S$ iff $\mathrm{ex}(e)\in M'$. Therefore, with $\mathcal{M}'$ denoting the set of perfect matchings in $G'$, we can rewrite \eqref{eq:t65g6hh223b} as
\[Z_G(\lambda,b)=2^{n}\Big(\frac{1+b}{2}\Big)^{m}\sum_{M'\in \mathcal{M}'}\Big(\frac{1-b}{1+b}\Big)^{|M'\cap \mathrm{ex}(G')|}.\]

Let $n'=|V'|\leq 3n$ and $m'=|E'|$. Let $p,q$ be positive integers with $\mathrm{gcd}(p,q)=1$ such that $\frac{p}{q}=\frac{1-b}{1+b}$. Let $G''$ be the multigraph obtained from $G'$ by replacing every external edge $e=(u,v)$ with $p$ parallel edges connecting $u$ to a new vertex $w_e$, $q$ parallel edges connecting $w_e$ to a new vertex $z_e$, and an edge between $z_e$ and $v$; note, internal edges of $G'$ are left intact. Let $\mathcal{M}'$ and $\mathcal{M}''$ be the set of perfect  matchings of $G'$ and $G''$  Then, there is a one-to-many correspondence between perfect matchings $M'\in \mathcal{M}'$ in $G'$ and perfect matchings $M''\in \mathcal{M}''$, where an internal edge $e$ is matched in $M'$ iff $e$ is matched in $M''$, while an external edge $e=(u,v)$ is matched in $M'$ iff $(z_e,v)$ is matched in $M''$. Note that, for an external edge $e$ and a perfect matching $M''$ of $G''$,  if $(z_e,v)$ belongs to $M''$ then $u$ must be matched by one of the $p$ parallel edges connecting $u$ to $w_e$, whereas if $(z_e,v)$ does not belong to $M''$, $w_e$ and $z_e$ must be matched by one of the $q$ parallel edges connecting $u$ to $w_e$; it follows that 
\[|\mathcal{M}''|=\sum_{M\in \mathcal{M}'}p^{|M\cap \mathrm{ex}(G')|}q^{m-|M\cap \mathrm{ex}(G')|}.\]

 Finally, if we let $G'''$ be the graph obtained from $G''$ by replacing every edge of $G''$ with a path of length 3, we have that the set of perfect matchings $\mathcal{M}'''$ off $G'''$ is in 1-1 correspondence with $\mathcal{M''}$, and we see that $2^{n}\big(\frac{1+b}{2}\big)^{m}q^{m}|\mathcal{M}'''|$  equals  $Z_G(\lambda,b)$, completing the $\mathsf{AP}$-reduction for $\Delta=3$.

To handle the case $\Delta\geq 4$, it suffices to show that $\FactorIsingb\leq_{\mathsf{AP}} \FactorIsingc$ since $\mathsf{AP}$-reductions are transitive, see \cite{DGGJ}. Let $G=(V,E)$ be a graph of maximum degree $\Delta$, and set $n=|V|$. Let $V_{\leq 3}=\{v\in V\mid d_v\leq 3\}$ be the set of vertices in $G$ with degree $\leq 3$, and $V_{>3}$ be the set of the remaining vertices.

Construct a graph $G'=(V',E')$ from $G$ by replacing every vertex $v\in V$ with $d_v=t\geq 4$, with a path of  $2t-1$ vertices if $t$ is odd and of $2t-3$ vertices if $t$ is even. We partition the vertices on the path into two sets $T_v,T_v'$ according to their parity, so that the endpoints of the path belong to $T_v$; note that $|T_v|=t$ if $t$ is odd, while $|T_v|=t-1$ if $t$ is even. We keep vertices $v\in V_{\leq 3}$ in $G'$, and for such vertices, we let for convenience $T_v=\{v\}$. Then, for every edge $(u,v)\in E$, we  add an edge in $G'$ between a vertex in $T_u$ and $T_v$ so that, in the end, $G'$ has maximum degree~3 and, further, for vertices $v\in V_{>3}$ with $d_v$ even, exactly one endpoint of the path on $T_v\cup T_v'$ has degree 3 in $G'$ (and the other has degree two). As before, we call an edge in $G'$ internal if both of its endpoints lie within a set $T_v$ for some $v\in V$, and external otherwise.

The key observation is that the aggregate contribution to $Z_{G'}(\lambda,b)$ from configurations on $G'$ where, for some $v\in V$, the vertices in $T_v$ do not get the same spin is zero.\footnote{This follows by observing that for a path with two edges, the aggregate weight of configurations where the endpoints of the path have different spins is equal to 0 (using that $\lambda=-1$).} For a configuration $\sigma$ on $G$, let $\Omega_{G',\sigma}$ be the set of configurations on $G'$ such that all vertices in $T_v$ get the spin $\sigma_v$ and let $Z_{G',\sigma}(\lambda,b)$ be their aggregate contribution to $Z_{G'}(\lambda,b)$, so that, from the observation above, we have
\[Z_{G'}(\lambda,b)=\sum_{\sigma: V\rightarrow\{\pl,\mi\}} Z_{G',\sigma}(\lambda,b).\]
For a configuration $\sigma: V\rightarrow\{\pl,\mi\}$, external edges and the external fields on $V_{\leq 3}$ contribute to $Z_{G',\sigma}(\lambda,b)$ a factor of $\lambda^{|n_\pl(\sigma)\cap V_{\leq 3}|}b^{|\delta_G(\sigma)|}$. For $v\in V_{>3}$ with $\sigma_v=\pl$, the edges in $T_v\cup T_v'$ and the external fields on $T_v\cup T_v'$ contribute to $Z_{G',\sigma}(\lambda,b)$ a factor of $-(1-b^2)^{|T_v|}$, and a factor of $(1-b^2)^{|T_v|}$  if $\sigma_v=\mi$. It follows that  $Z_{G',\sigma}(\lambda,b)=(1-b^2)^{|T|}\lambda^{|n_\pl(\sigma)|}w_{G,\sigma}(\lambda,b)$ where $T=\cup_{v\in V;d_v\geq 4} |T_v|$. It follows that 
\[Z_{G'}(\lambda,b)=(1-b^2)^{|T|}Z_{G}(\lambda,b),\]
therefore completing the $\mathsf{AP}$-reduction, since by construction $G'$ is a graph of maximum degree~3.

This finishes the proof of Theorem~\ref{thm:lambdaminusone}.
\end{proof}

\bibliographystyle{plainurl}
\bibliography{\jobname}

\begin{thebibliography}{10}

\bibitem{Barvinokbook}
A.~Barvinok.
\newblock {\em Combinatorics and complexity of partition functions}.
\newblock Algorithms and Combinatorics. Springer International Publishing,
  2017.

\bibitem{barvinokregts2019}
A.~Barvinok and G.~Regts.
\newblock Weighted counting of solutions to sparse systems of equations.
\newblock {\em Combinatorics, Probability and Computing}, 28(5):696--719, 2019.

\bibitem{Beardon1995}
A.~F. Beardon.
\newblock {\em The geometry of discrete groups}, volume~91 of {\em Graduate
  Texts in Mathematics}.
\newblock Springer-Verlag, 1995.

\bibitem{bencs2019leeyang}
F.~Bencs, P.~Buys, L.~Guerini, and H.~Peters.
\newblock Lee-{Y}ang zeros of the antiferromagnetic {I}sing model.
\newblock {\em arXiv e-prints}, abs/1907.07479, 2019.

\bibitem{bencs2018zerofree}
F.~Bencs, E.~Davies, V.~Patel, and G.~Regts.
\newblock On zero-free regions for the anti-ferromagnetic {P}otts model on
  bounded-degree graphs.
\newblock {\em arXiv e-prints}, abs/1812.07532, 2018.

\bibitem{BGGS}
I.~Bez\'{a}kov\'{a}, A.~Galanis, L.~A. Goldberg, and D.~\v{S}tefankovi\v{c}.
\newblock Inapproximability of the independent set polynomial in the complex
  plane.
\newblock In {\em Proceedings of the 50th Annual ACM SIGACT Symposium on Theory
  of Computing}, STOC 2018, pages 1234--1240, 2018.

\bibitem{matchings}
I.~Bez{\'a}kov{\'a}, A.~Galanis, L.~A. Goldberg, and D.~\v{S}tefankovi\v{c}.
\newblock {The complexity of approximating the matching polynomial in the
  complex plane}.
\newblock In {\em 46th International Colloquium on Automata, Languages, and
  Programming (ICALP 2019)}, volume 132, pages 22:1--22:13, 2019.

\bibitem{bremner}
M.~J. Bremner, A.~Montanaro, and D.~J. Shepherd.
\newblock Average-case complexity versus approximate simulation of commuting
  quantum computations.
\newblock {\em Phys. Rev. Lett.}, 117:080501, 2016.

\bibitem{buys2019location}
P.~Buys.
\newblock On the location of roots of the independence polynomial of bounded
  degree graphs.
\newblock {\em arXiv e-prints}, abs/1903.05462, 2019.

\bibitem{Chio2019}
I.~Chio, C.~He, A.~L. Ji, and R.~K.~W. Roeder.
\newblock Limiting measure of {L}ee--{Y}ang zeros for the {C}ayley tree.
\newblock {\em Communications in Mathematical Physics}, 370(3):925--957, 2019.

\bibitem{Collevecchio2016}
A.~Collevecchio, T.~M. Garoni, T.~Hyndman, and D.~Tokarev.
\newblock The worm process for the {I}sing model is rapidly mixing.
\newblock {\em Journal of Statistical Physics}, 164(5):1082--1102, 2016.

\bibitem{DGGJ}
M.~Dyer, L.~A. Goldberg, C.~Greenhill, and M.~Jerrum.
\newblock The relative complexity of approximate counting problems.
\newblock {\em Algorithmica}, 38(3):471--500, 2004.

\bibitem{Fisher}
M.~E. Fisher.
\newblock On the dimer solution of planar {I}sing models.
\newblock {\em Journal of Mathematical Physics}, 7(10):1776--1781, 1966.

\bibitem{galanis2020complexity}
A.~Galanis, L.~A. Goldberg, and A.~Herrera-Poyatos.
\newblock The complexity of approximating the complex-valued potts model.
\newblock {\em arXiv e-prints}, abs/2005.01076, 2020.

\bibitem{GSVanti}
A.~Galanis, D.~\v{S}tefankovi\v{c}, and E.~Vigoda.
\newblock Inapproximability for antiferromagnetic spin systems in the tree
  nonuniqueness region.
\newblock {\em J. ACM}, 62(6), 2015.

\bibitem{GSV2016}
A.~Galanis, D.~\v{S}tefankovi\v{c}, and E.~Vigoda.
\newblock Inapproximability of the partition function for the antiferromagnetic
  {I}sing and hard-core models.
\newblock {\em Combinatorics, Probability and Computing}, 25(4):500--559, 2016.

\bibitem{GG}
L.~A. Goldberg and H.~Guo.
\newblock The complexity of approximating complex-valued {I}sing and {T}utte
  partition functions.
\newblock {\em Computational Complexity}, 26(4):765--833, 2017.

\bibitem{Tutte08}
L.~A. Goldberg and M.~Jerrum.
\newblock Inapproximability of the {T}utte polynomial.
\newblock {\em Information and Computation}, 206(7):908--929, 2008.

\bibitem{GJ}
L.~A. Goldberg and M.~Jerrum.
\newblock The complexity of computing the sign of the {T}utte polynomial.
\newblock {\em {SIAM} J. Comput.}, 43(6):1921--1952, 2014.

\bibitem{guo2018}
H.~Guo and M.~Jerrum.
\newblock Random cluster dynamics for the {I}sing model is rapidly mixing.
\newblock {\em Ann. Appl. Probab.}, 28(2):1292--1313, 2018.

\bibitem{GLLZ}
H.~Guo, C.~Liao, P.~Lu, and C.~Zhang.
\newblock Zeros of {H}olant problems: locations and algorithms.
\newblock In {\em Proceedings of the Thirtieth Annual ACM-SIAM Symposium on
  Discrete Algorithms}, pages 2262--2278, 2019.

\bibitem{GLL}
H.~Guo, J.~Liu, and P.~Lu.
\newblock Zeros of ferromagnetic 2-spin systems.
\newblock In {\em Proceedings of the 2020 {ACM-SIAM} Symposium on Discrete
  Algorithms}, pages 181--192, 2020.

\bibitem{GuoLu}
H.~Guo and P.~Lu.
\newblock Uniqueness, spatial mixing, and approximation for ferromagnetic
  2-spin systems.
\newblock {\em ACM Trans. Comput. Theory}, 10(4), 2018.

\bibitem{Heilmann1972}
O.~J. Heilmann and E.~H. Lieb.
\newblock Theory of monomer-dimer systems.
\newblock {\em Communications in Mathematical Physics}, 25(3):190--232, 1972.

\bibitem{JS:ising}
M.~Jerrum and A.~Sinclair.
\newblock Polynomial-time approximation algorithms for the {I}sing model.
\newblock {\em SIAM Journal on Computing}, 22(5):1087--1116, 1993.

\bibitem{Jerrum1986}
M.~R. Jerrum, L.~G. Valiant, and V.~V. Vazirani.
\newblock Random generation of combinatorial structures from a uniform
  distribution.
\newblock {\em Theoretical Computer Science}, 43:169--188, 1986.

\bibitem{KowCai16}
M.~Kowalczyk and J.{-}Y. Cai.
\newblock Holant problems for 3-regular graphs with complex edge functions.
\newblock {\em Theory Comput. Syst.}, 59(1):133--158, 2016.

\bibitem{Cuevas}
G.~De las Cuevas, W.~D\"{u}r, M.~Van den Nest, and M.~A. Martin-Delgado.
\newblock Quantum algorithms for classical lattice models.
\newblock {\em New Journal of Physics}, 13(9):093021, 2011.

\bibitem{yang1952statistical}
T.-D. Lee and C.-N. Yang.
\newblock Statistical theory of equations of state and phase transitions. {I}.
  {T}heory of condensation.
\newblock {\em Physical Review}, 87(3):404, 1952.

\bibitem{lee1952statistical2}
Tsung-Dao Lee and Chen-Ning Yang.
\newblock Statistical theory of equations of state and phase transitions. ii.
  {L}attice gas and {I}sing model.
\newblock {\em Physical Review}, 87(3):410, 1952.

\bibitem{LLY}
L.~Li, P.~Lu, and Y.~Yin.
\newblock Correlation decay up to uniqueness in spin systems.
\newblock In {\em Proceedings of the 24th Annual {ACM-SIAM} Symposium on
  Discrete Algorithms (SODA)}, pages 67--84, 2013.

\bibitem{LLZ}
J.~Liu, P.~Lu, and C.~Zhang.
\newblock The complexity of ferromagnetic two-spin systems with external
  fields.
\newblock In {\em Approximation, Randomization, and Combinatorial Optimization.
  Algorithms and Techniques (APPROX/RANDOM 2014)}, volume~28, pages 843--856,
  2014.

\bibitem{LSP}
J.~Liu, A.~Sinclair, and P.~Srivastava.
\newblock A deterministic algorithm for counting colorings with 2-{D}elta
  colors.
\newblock In {\em IEEE 60th Annual Symposium on Foundations of Computer Science
  (FOCS 2019)}, pages 1380--1404, 2019.

\bibitem{liufisher}
J.~Liu, A.~Sinclair, and P.~Srivastava.
\newblock Fisher zeros and correlation decay in the {I}sing model.
\newblock {\em Journal of Mathematical Physics}, 60(10):103304, 2019.

\bibitem{liu2019ising}
J.~Liu, A.~Sinclair, and P.~Srivastava.
\newblock The {I}sing partition function: {Z}eros and deterministic
  approximation.
\newblock {\em Journal of Statistical Physics}, 174(2):287--315, 2019.

\bibitem{quant}
R.~L. Mann and M.~J. Bremner.
\newblock Approximation algorithms for complex-valued {I}sing models on bounded
  degree graphs.
\newblock {\em {Quantum}}, 3:162, 2019.

\bibitem{MendesEtAl1994}
P.~Mendes and F.~Oliveira.
\newblock On the topological structure of the arithmetic sum of two {C}antor
  sets.
\newblock {\em Nonlinearity}, 7(2):329--343, 1994.

\bibitem{Milnor2006}
J.~Milnor.
\newblock {\em Dynamics in one complex variable}, volume 160 of {\em Annals of
  Mathematics Studies}.
\newblock Princeton University Press, third edition, 2006.

\bibitem{PR17}
V.~Patel and G.~Regts.
\newblock Deterministic polynomial-time approximation algorithms for partition
  functions and graph polynomials.
\newblock {\em SIAM Journal on Computing}, 46(6):1893--1919, 2017.

\bibitem{exper}
X.~Peng, H.~Zhou, B.-B. Wei, J.~Cui, J.~Du, and R.-B. Liu.
\newblock Experimental observation of {L}ee-{Y}ang zeros.
\newblock {\em Phys. Rev. Lett.}, 114:010601, 2015.

\bibitem{peters2019}
H.~Peters and G.~Regts.
\newblock On a conjecture of {S}okal concerning roots of the independence
  polynomial.
\newblock {\em Michigan Mathematical Journal}, 68(1):33--55, 2019.

\bibitem{PetersRegts2018}
H.~Peters and G.~Regts.
\newblock Location of zeros for the partition function of the {I}sing model on
  bounded degree graphs.
\newblock {\em Journal of the London Mathematical Society}, 101:765--785, 2020.

\bibitem{Schrijver}
A.~Schrijver.
\newblock {\em Theory of Linear and Integer Programming}.
\newblock John Wiley \& Sons, Inc., 1986.

\bibitem{shao2019contraction}
S.~Shao and Y.~Sun.
\newblock Contraction: a unified perspective of correlation decay and
  zero-freeness of 2-spin systems.
\newblock {\em arXiv e-prints}, abs/1909.04244, 2019.

\bibitem{SST}
A.~Sinclair, P.~Srivastava, and M.~Thurley.
\newblock Approximation algorithms for two-state anti-ferromagnetic spin
  systems on bounded degree graphs.
\newblock {\em J. Stat. Phys.}, 155(4):666--686, 2014.

\bibitem{SlySun}
A.~Sly and N.~Sun.
\newblock Counting in two-spin models on d-regular graphs.
\newblock {\em Ann. Probab.}, 42(6):2383--2416, 2014.

\bibitem{VALIANT}
L.G. Valiant.
\newblock The complexity of computing the permanent.
\newblock {\em Theoretical Computer Science}, 8(2):189--201, 1979.

\end{thebibliography}

\end{document}